\newif \ifcomments \commentstrue
\newif \iffc \fctrue

\commentsfalse

\iffc
    \documentclass[runningheads]{llncs}
\else
    \documentclass[caption=false,letterpaper,conference,compsoc]{IEEEtran}
    \IEEEoverridecommandlockouts
\fi

\usepackage[utf8]{inputenc}
\usepackage{multicol}
\usepackage{fullpage}
\usepackage{url}
\usepackage{booktabs}
\usepackage{adjustbox}
\usepackage{mathtools}
\usepackage{enumitem}
\usepackage{fancyhdr} 

\usepackage{color}
\usepackage[dvipsnames]{xcolor}
\usepackage{placeins}
\usepackage{subfig}
\usepackage[linesnumbered,ruled,vlined]{algorithm2e}
\usepackage{etoolbox}
\usepackage{cite}
\usepackage{tabularx,ragged2e}
\newcolumntype{L}{>{\RaggedRight}X}
\usepackage{lipsum}
\usepackage{float}
\usepackage{epsfig}
\usepackage{comment}
\usepackage{pgfplots}
\usepackage{tikz}
\usepgfplotslibrary{groupplots}
\pgfplotsset{compat=newest}
\allowdisplaybreaks
\usepackage{amsfonts}
\usepackage{amssymb}
\usepackage[utf8]{inputenc}
\usepackage{graphicx}
\usepackage{hyperref}
\usepackage{refcount}
\usepackage{cleveref}
\Crefname{algorithm}{Alg.}{Algs.}
\usepackage{comment}
\usepackage{algorithmicx}

\usepackage[pass]{geometry} 
\usepackage{datetime2} 
\usepackage[normalem]{ulem}

  \let\subparagraph\relax
  \usepackage{titlesec}
  \titleformat{\subsubsection}[block]
    {\normalfont\normalsize\bfseries}{\thesubsubsection}{1em}{}
  \titlespacing*{\subsubsection}{0pt}{1.5ex plus .5ex minus .2ex}{0.7ex}
\DeclareRobustCommand{\Mref}[1]{%
  \ifmmode
    \text{\Cref{#1}}%
  \else
    \Cref{#1}%
  \fi
}

\iffc
\else
    \usepackage{amsthm}
    
    \newtheorem{remark}{Remark}
    \newtheorem{theorem}{Theorem}
    \newtheorem{lemma}[theorem]{Lemma}

    \newtheorem{corollary}[theorem]{Corollary}
    \newtheorem{definition}{Definition}

\fi

\newcommand{\sysname}{MonadBFT\xspace}

\newcommand{\mypara}[1]{\smallskip\noindent\textbf{#1}\hspace*{0.25em}}

\renewcommand{\addcontentsline}[3]{}

\newcommand{\viewduration}{\Theta_{\textup{view}}}
\newcommand{\recoveryduration}{\Theta_{\textup{recovery}}}

\newcommand{\batch}{\kappa}
\newcommand{\intervalduration}{\Theta_{\textup{interval}}}

\newcommand{\ProposalRequest}{\textsf{ProposalRequest}\xspace}
\newcommand{\ProposalResponse}{\textsf{ProposalResponse}\xspace}

\newcommand{\NERequest}{\textup{\textsf{NERequest}}}
\newcommand{\NoEndorsement}{\textsf{NE}\xspace}

\newcommand{\one}{\vspace{1mm}}

\newcommand{\validators}{\textup{\texttt{validators}}}
\newcommand{\requested}{\mathit{requested}}
\newcommand{\NEset}{\mathit{NEset}}
\newcommand{\ine}{\mathit{ne}}

\newcommand{\NE}{\textsf{NE}\xspace}

\newcommand{\hightip}{\textup{\texttt{high\_tip}}\xspace}

\newcommand{\hightipQCview}{\textup{\texttt{high\_tip\_qc\_view}}\xspace}
\newcommand{\highQC}{\textup{\texttt{high\_qc}}\xspace}

\newcommand{\tip}{\textup{\texttt{tip}}\xspace}

\newcommand{\block}{\textup{\texttt{block}}\xspace}

\newcommand{\view}{\textup{\textup{\texttt{view}}}\xspace}
\newcommand{\originalview}{\textup{\textup{\texttt{block\_view}}}\xspace}
\newcommand{\payload}{\textup{\texttt{payload}}\xspace}
\newcommand{\payloadhash}{\textup{\texttt{payload\_hash}}\xspace}
\newcommand{\blockheader}{\textup{\texttt{block\_header}}\xspace}
\newcommand{\blockhash}{\textup{\texttt{block\_hash}}\xspace}
\newcommand{\blockid}{\textup{\texttt{proposal\_id}}\xspace}

\newcommand{\hightipvote}{\textup{\texttt{tip\_vote}}}

\newcommand{\tipsviews}{\textup{\texttt{tips\_views}}\xspace}
\newcommand{\qcsviews}{\textup{\texttt{qcs\_views}}\xspace}
\newcommand{\qc}{\textup{\texttt{qc}}\xspace}
\newcommand{\tc}{\textup{\texttt{tc}}\xspace}
\newcommand{\nec}{\textup{\texttt{nec}}\xspace}
\newcommand{\lastcer}{\textup{\texttt{last\_cer}}\xspace}

\newcommand{\tipview}{\textup{\texttt{tip\_view}}}
\newcommand{\qcview}{\textup{\texttt{qc\_view}}}

\newcommand{\GST}{{\rm GST}}

\newcommand{\iqc}{\mathit{qc}}
\newcommand{\inec}{\mathit{nec}}
\newcommand{\itc}{\mathit{tc}}
\newcommand{\ivote}{\mathit{vote}}
\newcommand{\itip}{\mathit{T}}
\newcommand{\ihighqc}{\mathit{high\_qc}}
\newcommand{\iview}{\mathit{view}}
\newcommand{\ivotes}{\mathsf{votes}}
\newcommand{\itimeoutvotes}{\mathsf{timeout\_votes}}
\newcommand{\ivotesproposalview}{\mathit{votes\_proposal\_view}}

\newcommand{\iblock}{\mathit{B}}
\newcommand{\iblockheader}{H}
\newcommand{\iprop}{\mathit{p}}
\newcommand{\itimeout}{\mathit{M}}
\newcommand{\ileader}{\mathit{L}}
\newcommand{\iparent}{\mathit{parent}}
\newcommand{\igparent}{\mathit{grandparent}}
\newcommand{\itips}{\mathit{tips}}
\newcommand{\iqcs}{\mathit{qcs}}
\newcommand{\iqcview}{\mathit{qc\_view}}
\newcommand{\iblockview}{\mathit{block\_view}}
\newcommand{\ipayload}{\mathit{payload}}
\newcommand{\ipayloadhash}{\mathit{payload\_hash}}
\newcommand{\iblockhash}{\mathit{block\_hash}}
\newcommand{\iproposalid}{\mathit{proposal\_id}}
\newcommand{\icer}{\mathit{cer}}

\newcommand{\itimeouts}{\mathsf{timeouts}}
\newcommand{\imsg}{\mathit{msg}}
\newcommand{\ihightip}{\mathit{high\_tip}}
\newcommand{\iindex}{\mathit{i}}
\newcommand{\iviewcer}{\mathit{view\_cer}}

\newcommand{\tupled}[1]{\langle #1 \rangle}

\newcommand{\Proposal}{\text{\small \textsf{Proposal}}\xspace}

\newcommand{\TC}{\text{\small \textsf{TC}}\xspace}
\newcommand{\NEC}{\text{\small \textsf{NEC}}\xspace}

\newcommand{\IsFreshProposal}{\textnormal{\textsc{IsFreshProposal}}}
\newcommand{\ValidQC}{\textnormal{\textsc{ValidQC}}}
\newcommand{\ValidTC}{\textnormal{\textsc{ValidTC}}}
\newcommand{\ValidVote}{\textnormal{\textsc{ValidVote}}}
\newcommand{\ValidBlock}{\textnormal{\textsc{ValidBlock}}}
\newcommand{\ValidBlockHeader}{\textnormal{\textsc{ValidBlockHeader}}}
\newcommand{\ValidNEC}{\textnormal{\textsc{ValidNEC}}}
\newcommand{\ValidTimeoutMessage}{\textnormal{\textsc{ValidTimeoutMessage}}}
\newcommand{\ValidTip}{\textnormal{\textsc{ValidFreshTip}}}

\newcommand{\SafetyCheck}{\textnormal{\textsc{ValidProposal}}}
\newcommand{\CreateTC}{\textnormal{\textsc{CreateTC}}}
\newcommand{\CreateNEC}{\textnormal{\textsc{CreateNEC}}}
\newcommand{\CreateQC}{\textnormal{\textsc{CreateQC}}}
\newcommand{\CreateTimeoutMsg}{\textnormal{\textsc{CreateTimeoutMessage}}}
\newcommand{\CurrentLeaderReceivesQC}{\textnormal{\textsc{BroadcastBackupQC}}}
\newcommand{\NextLeaderReceivesQC}{\textnormal{\textsc{CreateProposalBasedOnQC}}}
\newcommand{\NextLeaderReceivesTC}{\textnormal{\textsc{CreateProposalBasedOnTC}}}

\newcommand{\NonLeaderReceivesQC}{\textnormal{\textsc{ForwardBackupQCToNextLeader}}}

\newcommand{\ProcessProposalRequest}{\textnormal{\textsc{ProcessProposalRequest}}}
\newcommand{\ProcessNERequest}{\textnormal{\textsc{ProcessNERequest}}}

\newcommand{\CommitSpecCommit}{\textnormal{\textsc{CommitAndSpecCommit}}}

\newcommand{\mss}[1]{\ifmmode\mathsf{#1}\else\textsf{#1}\fi}
\newcommand{\localhighQC}{\textup{\mss{last\_qc}}\xspace}
\newcommand{\curView}{\mathsf{cur\_view}\xspace}

\newcommand{\proposal}{\mss{p}\xspace}
\newcommand{\lastTC}{\mss{last\_tc}\xspace}
\newcommand{\localhightip}{\mss{local\_tip}\xspace}

\newcommand{\highestvotedview}{\mss{highest\_voted\_view}\xspace}

\newcommand{\pmh}[1]{{\color{blue}{#1}}}

\newcommand{\false}{\mathsf{false}}

\newcommand{\Require}{{\bf require}\xspace}

\newcommand{\deq}{\stackrel{\text{def}}{=}}

\newcommand{\Recover}{\textup{\textsc{Recover}}}

\newcommand{\malor}{view responsiveness}

\newcommand{\msg}[1]{{\langle #1 \rangle}}

\usepackage{xcolor}
\definecolor{darkgray}{gray}{0.45}

\newcommand{\commentline}[1]{{\textcolor{darkgray}{\normalfont // \emph{#1}}}}
\newcommand{\commentinline}[1]{{\hfill \textcolor{darkgray}{\normalfont // \emph{#1}}}}

\ifcomments
    \newcommand{\todo}[1]{{\color{orange}{[ToDo: #1]}}}
    
    \newcommand{\new}[1]{{\color{purple}{#1}}}
    \newcommand{\old}[1]{{\color{cyan}\sout{#1}}}
\else
    \newcommand{\todo}[1]{}
    
    \newcommand{\new}[1]{}
    \newcommand{\old}[1]{}
\fi

\ifcomments
    \newcommand{\kushal}[1]{{\color{orange}{[KB: #1]}}}
    
    \newcommand{\mussadiq}[1]{{\color{purple}{[MJ: #1]}}}
    
    \newcommand{\yd}[1]{{\color{olive}{[YD: #1]}}}
    
    \newcommand{\tobias}[1]{{\color{red}{[TK: #1]}}}
    \newcommand{\tobiaschange}[1]{{\color{red!30!orange}{#1}}}
    \newcommand{\tobiasremove}[1]{{\color{red!30!orange}{\sout{#1}}}}
    \newcommand{\sd}[1]{{\color{magenta}{[SD: #1]}}}
    
    \newcommand{\jovan}[1]{{\color{teal}{[JV: #1]}}}
    
    \newcommand{\fatima}[1]{{\color{blue}{[FA: #1]}}}
    
\else
    \newcommand{\kushal}[1]{}
    
    \newcommand{\mussadiq}[1]{}
    
    \newcommand{\yd}[1]{}
    
    \newcommand{\tobias}[1]{}
    \newcommand{\tobiasremove}[1]{}
    \newcommand{\tobiaschange}[1]{{{#1}}}
    \newcommand{\sd}[1]{}
    
    \newcommand{\jovan}[1]{}
    
    \newcommand{\fatima}[1]{}
    
\fi

\newcommand{\authorsep}{\and}

\SetKwFunction{areEquivalent}{\sc areEquivalent}
\SetKwFunction{BQC}{GenerateQC}
\SetKwFunction{BroadcastTimeout}{\sc BroadcastTimeout}
\SetKwFunction{BuildBlock}{\sc BuildBlock}
\SetKwFunction{BuildQC}{\sc BuildQC}
\SetKwFunction{BuildTimeoutCertificate}{\sc BuildTimeoutCertificate}
\SetKwFunction{CheckQCsupersedesTips}{\sc CheckQCsupersedesTips}
\SetKwFunction{Commit}{\sc Commit}
\SetKwFunction{CommitIfNeeded}{\sc CommitIfNotCommitted}
\SetKwFunction{CreateNECertificate}{\sc CreateNECertificate}
\SetKwFunction{CreatePrepareMsg}{\sc BuildProposal}
\SetKwFunction{CreateVote}{\sc CreateVote}
\SetKwFunction{DA}{\sc DA}
\SetKwFunction{FExecuteBlock}{\sc ExecuteBlock}
\SetKwFunction{FindHighTip}{\sc FindHighTip}
\SetKwFunction{FReceiveStateHash}{\sc ReceiveStateHash}
\SetKwFunction{GetcurView}{\sc CurView}
\SetKwFunction{GetEligibleTips}{GetEligibleTips}
\SetKwFunction{GetLocalNE}{\sc DA\_GetLocalNE}
\SetKwFunction{HandleTimeout}{\sc HandleTimeout}
\SetKwFunction{Localtimeoutview}{\sc InitiateTimeout}
\SetKwFunction{HandleStaleTimeout}{\sc HandleStaleTimeout}
\SetKwFunction{HandleVote}{\sc HandleVote}
\SetKwFunction{HandleTimeoutVote}{\sc HandleTimeoutVote}
\SetKwFunction{Hash}{\sc Hash}
\SetKwFunction{IncrementView}{\sc IncrementView}
\SetKwFunction{Leader}{\sc Leader}
\SetKwFunction{LeaderRecovery}{\sc LeaderRecovery}
\SetKwFunction{MempoolTransactions}{\sc MempoolTransactions}
\SetKwFunction{ParentID}{\sc ParentID}
\SetKwFunction{PipelinedSafeBlock}{\sc SafetyCheck}
\SetKwFunction{ProcessMessage}{ProcessMessage}
\SetKwFunction{RecoveryValidity}{\sc RecoveryValidity}
\SetKwFunction{RetrieveBlock}{\sc RetrieveBlock}
\SetKwFunction{RetrieveProposal}{\sc RetrieveProposal}
\SetKwFunction{SafeProposal}{\sc SafeProposal}
\SetKwFunction{SendRecoveryRequest}{\sc SendRecoveryRequest}
\SetKwFunction{Sign}{\sc Sign}
\SetKwFunction{SpeculateAndECIfNeeded}{\sc SpeculateIfNotSpeculated}
\SetKwFunction{StartTimer}{\sc StartTimer}
\SetKwFunction{StopTimer}{\sc StopTimer}
\SetKwFunction{tip}{tip}
\SetKwFunction{UpdateQC}{\sc UpdateHighQC}
\SetKwFunction{GetTip}{\sc GetTip}
\SetKwFunction{ValidateChunk}{\sc ValidateChunk}
\SetKwFunction{ValidateHighTip}{\sc ValidateHighTip}
\SetKwFunction{ValidatorSync}{\sc ValidatorSync}
\SetKwFunction{IsFreshProposal}{\sc IsFreshProposal}
\SetKwFunction{ValidLeaderRoot}{\sc ValidLeaderRoot}
\SetKwFunction{verifyFreshProposal}{\sc VerifyFreshProposal}

\SetKwFunction{SendRecoveryRequest}{{\sc SendRecoveryRequest}}
\SetKwFunction{ValidatorSync}{{\sc ValidatorSync}}
\SetKwFunction{OnchunkReception}{{\sc OnchunkReception}}
\SetKwFunction{OnchunkNEReception}{{\sc OnchunkNEReception}}
\SetKwFunction{LeaderRecovery}{{\sc LeaderRecovery}}
\SetKwFunction{ProcessFailureProof}{{\sc ProcessFailureProof}}
\SetKwFunction{ChunkMatchesRequest}{{\sc ChunkMatchesRequest}}
\SetKwFunction{NoDataMatchesRequest}{{\sc NoDataMatchesRequest}}

\SetKwProg{Fn}{Function}{:}{}
\SetKwFor{Upon}{upon}{do}{end}

\SetKwFunction{DARequest}{{\sc DA\_LeaderRequest}}      
\SetKwFunction{OnPropReq}{{\sc DA\_OnProposalRequest}}         
\SetKwFunction{OnNEReq}{{\sc DA\_OnNoEndorsementRequest}}             
\SetKwFunction{ReportNE}{{\sc DA\_ReportLocalNE}}      
\SetKwFunction{GetLocalNE}{{\sc DA\_GetLocalNE}} 
\SetKwFunction{IsSafeNE}{{\sc DA\_IsSafeToNoEndorse}} 
\SetKwFunction{IsSafeVote}{{\sc Consensus\_IsSafeToVote}} 
\SetKwFunction{IsInvalidR}{\sc IsInvalidRequest}
\SetKwFunction{DAOnRelay}{{\sc DA\_OnRELAY}}

\SetKwComment{NoSlashComment}{}{}
\newcommand{\commentif}[1]{{\NoSlashComment*[h]{\textcolor{darkgray}{\normalfont // \emph{#1}}}}}

\title{\sysname : Fast, Responsive, Fork-Resistant Streamlined Consensus}
\author{
\setlength{\baselineskip}{1.1\baselineskip}
Mohammad Mussadiq Jalalzai \authorsep
        Kushal Babel \authorsep
        Jovan Komatovic \authorsep
        Tobias Klenze \authorsep
        Sourav Das \authorsep
        Fatima~Elsheimy \authorsep
        Mike Setrin \authorsep
        John Bergschneider \authorsep
        Babak Gilkalaye
 }
 \institute{Category Labs}

\begin{document}

\maketitle
\begin{abstract}

This paper introduces \sysname, a novel Byzantine Fault Tolerant (BFT) consensus protocol that advances both performance and robustness. \sysname is implemented as the consensus protocol in the Monad blockchain.
As a HotStuff-family protocol, \sysname~has \emph{linear message complexity} in the common case and is \emph{optimistically responsive}, operating as quickly as the network allows.
A central feature of \sysname\ is its \emph{tail-forking resistance}. In pipelined BFT protocols, when a leader goes offline, the previous proposal is abandoned. 
Malicious leaders can exploit this tail-forking behavior as a form of Maximal Extractable Value (MEV) attack by deliberately discarding their predecessor's block, depriving that proposer of rewards and enabling transaction reordering, censorship, or theft. \sysname\ prevents such tail-forking attacks, preserving both fairness and integrity in transaction execution.
Another related feature of \sysname~is its notion of \emph{speculative finality}, which enables parties to execute ordered transactions after a single round (i.e., a single view), with reverts occurring only in the rare case of provable leader equivocation. This mechanism reduces user-perceived latency.
Additionally, we introduce the \emph{leader fault isolation} property, which ensures that the protocol can quickly recover from a failure.
To our knowledge, no prior pipelined, leader-based BFT consensus protocol combines all of these properties in a single design.

\end{abstract}




\section{Introduction}
\label{sec:introduction}
Byzantine Fault Tolerant (BFT) consensus protocols allow participants in a network to reach consensus even when some participants may act arbitrarily.
With the rise of blockchains, however, BFT consensus protocols must address additional challenges: ensuring fairness, since participants are directly rewarded or penalized for their behavior, and achieving scalability to support hundreds of \emph{validators} coordinating within the protocol.
In the context of a high-performance L1 blockchain, these requirements give rise to several key design goals for \sysname.
While prior work addresses some of these goals in isolation, \sysname is, to the best of our knowledge, the first protocol to achieve all of them within a single framework.

%
\begin{itemize}[leftmargin=*, itemsep=3pt]
    
\item \emph{Rotating leaders, pipelining \& linear complexity:} 
Unlike traditional leader-based BFT protocols~\cite{Castro:1999:PBF:296806.296824} that have a steady leader driving consensus, blockchain BFT protocols employ rotating leaders~\cite{tendermint,HotStuff} to enhance fairness and robustness against slow or faulty leaders. Moreover, the protocol must pipeline commit decisions across successive leaders to achieve high throughput. Additionally, it should maintain linear communication complexity in the common case to lower validator overhead, promote wider participation, and reduce operational costs~\cite{HotStuff}.

\item \emph{Optimistic responsiveness:}
The \emph{optimistic responsiveness} property lets a protocol make progress at the actual network speed~$\delta$ rather than on the worst-case message delay~$\Delta \gg \delta$~~\cite{HotStuff,jalalzai2021fasthotstuff, gelashvili2021jolteon}. 
We require blockchain BFT protocols to have optimistic responsiveness, as it minimizes reliance on a carefully calibrated value of~$\Delta$, which is notoriously difficult to tune in practice and significantly affects both latency and throughput.

\item \emph{Tail-forking resistance:} 
Recent work~\cite{BEEGEES} has identified a fairness and liveness vulnerability in (pipelined) BFT protocols, known as \emph{tail-forking}. This vulnerability arises from a protocol design where leaders are responsible for disseminating certificates for the preceding leader’s block. 
In \Cref{fig:Forking-Attack} (a), we show how an offline leader causes the previous block to be abandoned, despite the block gathering the support of most validators. 
Furthermore, a malicious validator can purposefully abandon the prior block to steal, reorder, or insert transactions, capturing additional \emph{Maximal Extractable Value (MEV)}~\cite{flashboys, babel2023clockwork}—a practice that has already led to hundreds of millions of dollars in user losses.
Tail-forking thus amplifies MEV opportunities through cross-block transaction reordering and (short-term) censorship, giving leaders disproportionate influence over transaction inclusion.
Such manipulation increases centralization pressure and undermines both the fairness and reliability of the consensus protocol.
\begin{figure}[t!]
\centering
\begin{minipage}{0.7\textwidth}
\centering
\includegraphics[width=\textwidth]{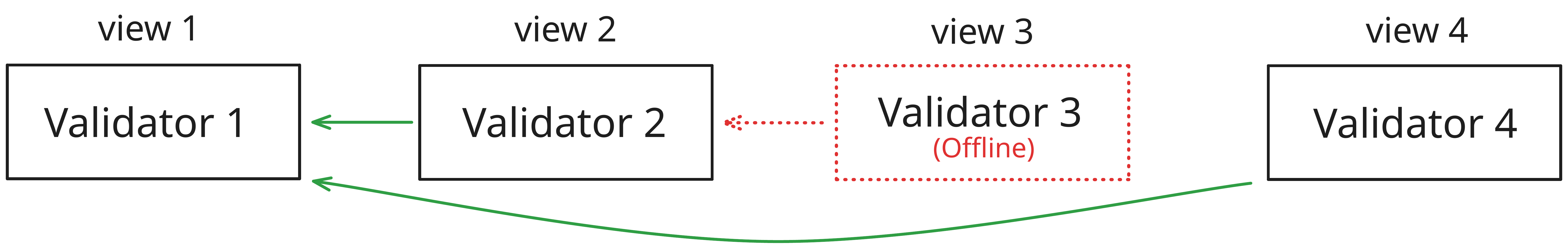}
\vspace{-15pt}
\caption*{(a) Example of tail forking in a pipelined protocol.}
\end{minipage}

\vspace{8pt}

\begin{minipage}{0.7\textwidth}
\centering
\includegraphics[width=\textwidth]{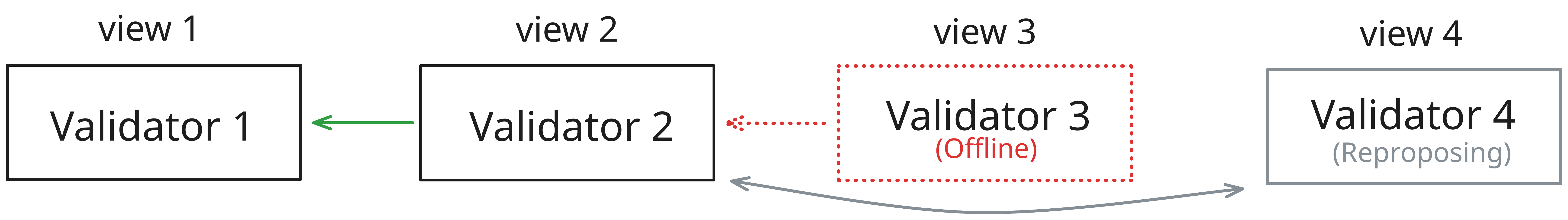}
\vspace{-15pt}
\caption*{(b) \sysname prevents tail-forking via reproposals (standard recovery).}
\end{minipage}

\vspace{8pt}

\begin{minipage}{0.7\textwidth}
\centering
\includegraphics[width=\textwidth]{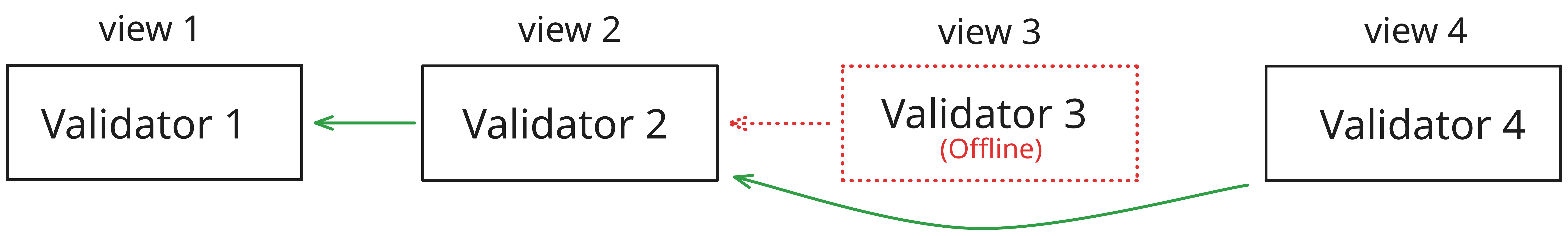}
\vspace{-15pt}
\caption*{(c) Fast recovery in \sysname, also tail-forking resistant.}
\end{minipage}

    \caption{
    Example of tail-forking (a) and how \sysname's standard (b) and fast (c) recovery prevent it.\\
    (a) In standard pipelined consensus, blocks can easily be abandoned. Here, Validator 2 proposes a block that gets abandoned because Validator 3 fails to collect votes for it.\\
    (b) \sysname guarantees that a correct validator's block is never abandoned. In this case, it is reproposed by Validator 4. \\
    (c) \sysname often enables fast recovery via backup QC, which skips the reproposal while still preventing tail-forking.
    }
    \label{fig:Forking-Attack}
\end{figure}

\item \emph{Speculative commit:} 
To provide low-latency confirmations for typical executions, the protocol should support \emph{speculative finality} after a single round (i.e., a single view), with revocation allowed only when the leader equivocates.
Because equivocation is provable and punishable in a blockchain setting, this mechanism enables fast confirmations in the common case.


\item \emph{Leader fault isolation:}
The protocol should recover from failures quickly. To this end, we introduce the \emph{leader fault isolation} property, which states that a single Byzantine leader only incurs a single timeout delay; all views led by correct leaders operate at network speed.

\item \emph{Low latency:}
The blockchain consensus protocol should have low commit latency, which leads to faster transaction confirmations and, consequently, a better user experience.

\end{itemize}


\subsection{\sysname: Technical Overview}  
We now provide a brief overview of the design and highlight the key ideas in \sysname.
%

\medskip
\mypara{Base protocol.}
\sysname operates in the partially synchronous model, and falls within the HotStuff family of protocols, operating in a steady state under honest leaders and a stable network—referred to as the \textit{happy path}—and transitioning to an \textit{unhappy path} when recovering from Byzantine leaders or network instabilities.

\sysname progresses through a sequence of \emph{views} (also known as \emph{rounds}), each coordinated by a designated \emph{leader}.  
Within a view, validators follow the standard pipelined consensus steps: they receive the leader's block proposal $p$, validate it, and send their vote to the leader of the next view. 
This next leader collects the issued votes and assembles a \emph{quorum certificate}~(QC) that points to $p$, confirming that more than $2/3$ of validators have voted for $p$, and then includes this QC in its own block proposal for the subsequent view.
Once validators receive the next leader’s proposal along with the QC for $p$, they enter the new view, forming a continuous, certified chain of block proposals. 
This describes {\sysname}’s operation under ideal conditions --- the happy path --- where leaders act honestly and the network is stable.  
On the happy path, every view produces a QC that is immediately incorporated by the following leader, allowing the blockchain to steadily grow without interruption.  

When failures occur, for instance due to an offline leader or long network delays, we call this the \emph{unhappy path}. A validator will notice a failure when it remains in a view $v$ for longer than a certain pre-specified duration. It then times out from the view and broadcasts (i.e., sends to all validators) a \emph{timeout message} for the current view $v$.
Timeout messages from more than $2/3$ validators indicate a failure of view $v$. The messages are combined into a \emph{timeout certificate}~(TC), which allows validators to advance to view $v+1$ without waiting for a QC from view $v$. In general, any validator in a view $\leq v$ will transition to view $v+1$ as soon as it receives either a QC or a TC for view $v$.

\medskip
\mypara{Tail-forking resistance via reproposals.}
Our design so far mirrors the pipelined HotStuff protocol, which lacks \emph{tail-forking resistance} --- the property that once a block proposed by a leader can gather a QC, it cannot be abandoned, unless the leader provably equivocated. 
To achieve this, \sysname introduces a \emph{reproposal mechanism}, where a leader recovering from a failed view must repropose the block from that view instead of proposing a new one.
\Cref{fig:Forking-Attack} shows \sysname's standard and fast recovery, which both prevent tail-forking.

\noindent \emph{Tips and high tip:}
A key challenge in the reproposal mechanism is to ensure that all validators agree on when and which block should be reproposed.
\sysname enforces this by having each validator track the latest proposal it has voted for, called its \emph{local tip}. Upon a view failure, validators include their local tip in their timeout messages.
The next leader collects these messages, identifies the \emph{high tip} --- the tip with the highest view among them --- and reproposes the corresponding block, unless it can construct a QC for that block.
Finally, the leader includes the reported local tips into the constructed TC, allowing everyone to verify the reproposal decision.
At a high level, reproposing the high tip ensures tail-forking resistance for the following reason.
For a block $B$ to obtain a QC, at least $f+1$ honest validators must have voted for it.
When recovering from a failed view, each of these honest validators reports either $B$ or a descendant of $B$ as its local tip.
Since the high tip is determined from the $2f+1$ local tips, at least one local tip from an honest validator that voted for $B$ is guaranteed to be included in the high-tip computation.
This ensures that $B$ (or its extension) is always reproposed, thereby providing tail-forking resistance.

\noindent \emph{No endorsement certificates:}
Note that under the high-tip reproposal rule, the leader must possess the block corresponding to the high tip in order to repropose it.
A naive but impractical solution would be to embed the entire block within the tip itself. Instead, in \sysname, a tip contains only the block’s identity (its hash) and a small amount of auxiliary metadata.
Additionally, \sysname\ allows a leader who does not possess the block corresponding to the high tip to retrieve the missing block by querying other validators. In rare cases where an honest leader fails to recover the block, validators jointly issue a cryptographic \emph{no-endorsement certificate} (NEC). The existence of an NEC guarantees that the high-tip block could have never obtained a QC, allowing the honest leader to safely bypass reproposal and propose a new block instead.

\medskip
\mypara{Fast recovery optimizations.}
In addition to the above, we introduce \emph{fast recovery} optimizations that help the protocol efficiently handle common failure scenarios.
\tobiaschange{
In particular, validators send \emph{backup votes} to the current leader, enabling it to form and broadcast a \emph{backup QC} for its own proposal, which can be used when the next leader fails to collect regular votes.
Using the recovery from a single Byzantine leader as an example, \Cref{fig:timeline-single-leader-failure-monadbft} illustrates how \sysname's fast recovery can restore the happy path faster than standard pipelined consensus protocols (\Cref{fig:timeline-single-leader-failure-pipelined-consensus}), while ensuring tail-forking resistance.
}
We elaborate on these optimizations in~\Cref{sec:Protocol_Description} after formalizing the related efficiency requirements in~\Cref{sub:consensus_properties}.

\begin{figure}[t]
    \centering
    \begin{minipage}[t]{\textwidth}
        \centering
        \includegraphics[height=5cm]{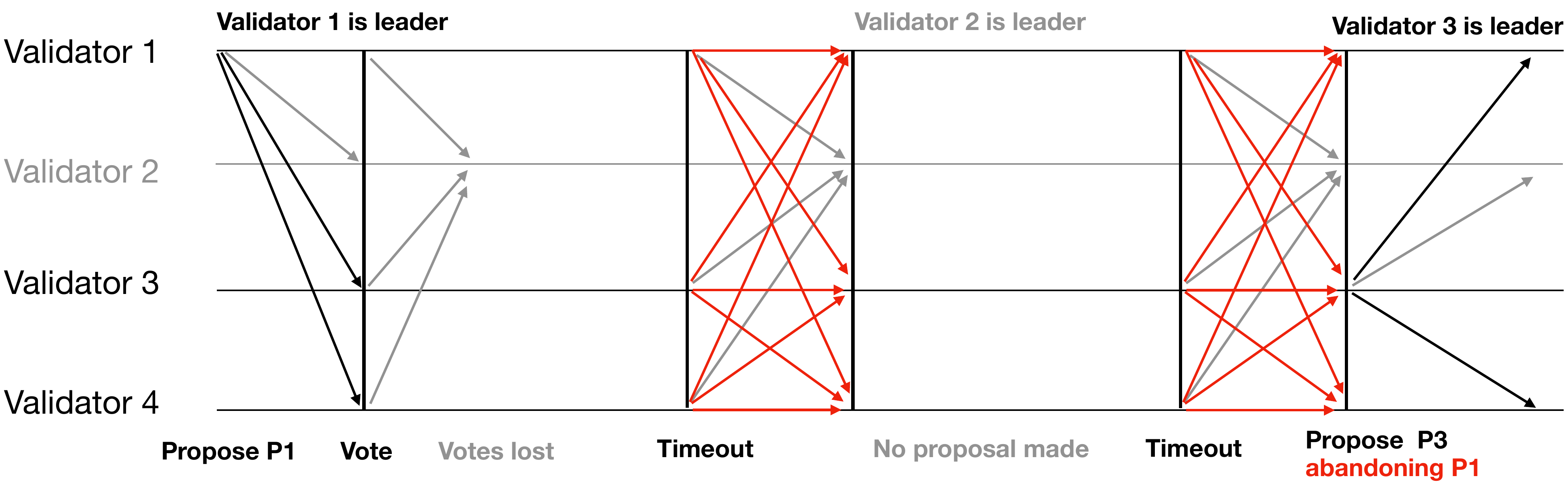}
        \caption{In \textbf{standard pipelined BFT}, a single offline leader (Validator 2) causes two views to time out, as Validator 2 not only fails to make a proposal, but also fails to collect votes for Validator 1's block proposal P1. Furthermore, Validator 1's block is lost when Validator 3 makes its own proposal P3.}
        \label{fig:timeline-single-leader-failure-pipelined-consensus}
    \end{minipage}

    \vspace{20pt}

    \begin{minipage}[t]{\textwidth}
        \centering
        \includegraphics[height=5cm]{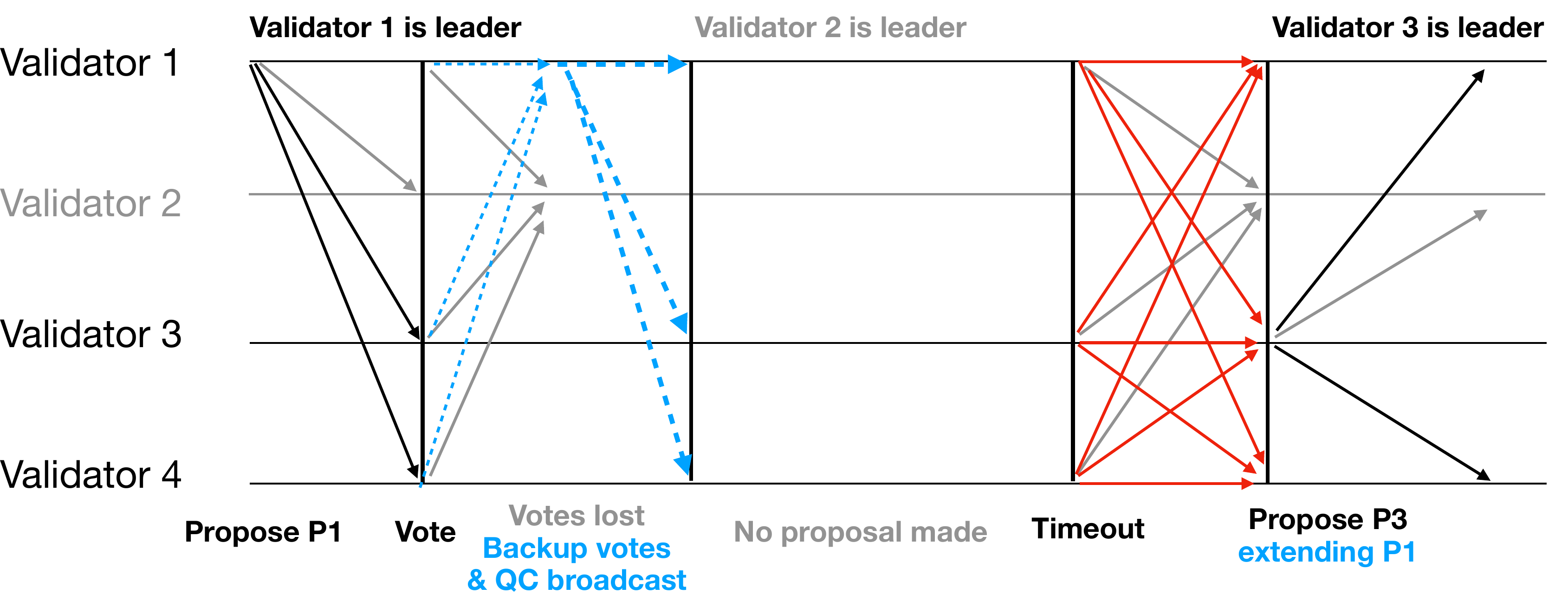}        \caption{In \textbf{\sysname}, 
        a single offline leader (Validator 2) causes only one view to time out, 
        as Validator~1 collects backup votes (thin dashed \textcolor{Cerulean}{blue}) and broadcasts a backup QC (thick dashed \textcolor{Cerulean}{blue}) for its own block proposal~P1, allowing the view to complete successfully. Validator 3, recovering from the failure, can build directly on the backup QC of P1, ensuring that P1 (Validator 1's block) is not abandoned.
        To rule out tail-forks, Validator 3 must prove that Validator~2 did not produce a block that received sufficient support.
        This figure illustrates \sysname's \emph{fast recovery} (introduced later in the paper), which avoids reproposals.
        }
        \label{fig:timeline-single-leader-failure-monadbft}
    \end{minipage}
\end{figure}

\subsection{Achieved Properties of \sysname}
\label{sub:properties}
\begin{table}[t!]
\caption{Comparison of chain-based rotating leader BFT SMR protocols. $\delta$ is the actual network delay (c.f.~\Cref{subsec:model}).}
\scriptsize
\label{tab:comparison}
\setlength{\tabcolsep}{4pt}
\renewcommand{\arraystretch}{1.2}

\centering
\noindent\makebox[\textwidth][c]{%
\begin{tabular}{@{}lccccccc@{}}
\hline
\rule{0pt}{8ex}
\textbf{Protocol} &
\shortstack{\textbf{Min.}\\\textbf{Commit}\\\textbf{Latency}} &
\shortstack{\textbf{Tail-}\\\textbf{forking}\\\textbf{Resilience}} &
\textbf{Pipelined} &
\shortstack{\textbf{Happy}\\\textbf{Path}\\\textbf{Complexity}} &
\shortstack{\textbf{Unhappy}\\\textbf{Path}\\\textbf{Complexity}} &
\shortstack{\textbf{Leader-based.}\\\textbf{Opt.}\\\textbf{Resp.}} \\
\hline
HotStuff \cite{HotStuff-2:2023/397} & $7\delta$ & $\times$ & $\checkmark$ & $O(n)$ & $O(n)$ & $\checkmark$ \\
Fast-HotStuff \cite{jalalzai2021fasthotstuff} & $5\delta$ & $\times$ & $\checkmark$ & $O(n)$ & $O(n^2)$ & $\checkmark$ \\
Jolteon \cite{gelashvili2021jolteon} & $5\delta$ & $\times$ & $\checkmark$ & $O(n)$ & $O(n^2)$ & $\checkmark$ \\
HotStuff-2 \cite{HotStuff-2:2023/397} & $5\delta$ & $\times$ & $\checkmark$ & $O(n)$ & $O(n)$ & $\times$ \\
PaLa \cite{Pala} & $4\delta$ & $\times$ & $\checkmark$ & $O(n)$ & $O(n^2)$ & $\times$ \\
Banyan \cite{albarello2023fastinternetcomputerconsensusICC} & $3\delta$ & $\times$ & $\times$ & $O(n^2)$ & $O(n^2)$ & $\times$ \\
BeeGees \cite{BEEGEES}$^{(1)}$ & $3\delta$ & $\checkmark$ & $\checkmark$ & $O(n)$ & $O(n^2)$ & $\times$ \\
VBFT~\cite{VBFT} & $2\delta$$^{(2)}$   & $\checkmark$ & $\checkmark$ & $O(n^2)$ & $O(n^2)$ & $\checkmark$ \\
Moonshot \cite{doidge2024moonshotoptimizingchainbasedrotating} & $3\delta$ & $\times$ & $\checkmark$ & $O(n^2)$ & $O(n^2)$ & $\checkmark$\\
HotStuff-1 (Pipelined) \cite{kang2024hotstuff1linearconsensusonephase} & $3\delta$$^{(2)}$ & $\times^{(3)}$ & $\checkmark$ & $O(n)$ & $O(n)$ & $\times$\\
\hline
\textbf{This work} & $3\delta$$^{(2)}$ & $\checkmark$ & $\checkmark$ & $O(n)$ & $O(n^2)$ & $\checkmark$ \\
\hline
\end{tabular}%
}

\vspace{2mm}
\begin{minipage}{\textwidth}
\scriptsize
\begin{flushleft}
(1) Although BeeGees achieves asymptotic complexities comparable to MonadBFT, certain deployment-oriented considerations are not reflected
in these bounds; we discuss these aspects in \Cref{para:beegees}.\\
(2) Latency is minimized due to speculation. \\
(3) The HotStuff-1 pipelined variant mitigates tail-forking with continuous block proposals by the same leader (``slotting'') but becomes vulnerable upon leader rotation, the original reason for the tail-forking vulnerability. \\
\end{flushleft}
\end{minipage}
\end{table}

\medskip
\mypara{Optimistic linear complexity.} 
As in prior work~\cite{jalalzai2021fasthotstuff,gelashvili2021jolteon}, \sysname achieves linear message complexity on the happy path and quadratic complexity on the unhappy path.
We adopt this design for two reasons.
First, we expect the protocol to operate predominantly on the happy path, since faulty leaders are rare in proof-of-stake blockchains due to financial penalties~\cite{he2023don, slashing-ethereum}.
Second, the quadratic communication pattern accelerates recovery after a view failure by leveraging the all-to-all message exchange. Specifically, \sysname relies on the all-to-all communication Bracha amplification approach~\cite{Bracha-Broadcast}~(hence quadratic cost) to rapidly disseminate view failure information and synchronize validators to the next view.

\medskip
\mypara{Optimistic responsiveness.}
On the happy path, \sysname progresses as soon as the next leader collects the QC from the previous view and includes it in its proposal, enabling the protocol to operate at network speed. 
Even on the unhappy path, progress remains largely independent of timeouts --- these are used only to detect a view failure and trigger the transition to the unhappy path. 
Specifically, validators issue timeout messages after waiting for $\Delta$ time units to detect such a failure. 
A TC, composed of these timeout messages, triggers the transition to the unhappy path; after that, the protocol operates at network speed.


\medskip
\mypara{Tail-forking resistance \& speculative commit.} 
Tail-forking resistance ensures that any QC-backed block, i.e., blocks that have been voted for by at least $f + 1$ honest validators, will appear in the chain at its proposed height, unless the leader equivocates.\footnote{In~\Cref{sub:consensus_properties}, we specify the precise conditions under which a QC-backed block can be reverted.}
We leverage this property and the ability to impose financial penalties for equivocation in proof-of-stake protocols to allow validators to \emph{speculatively commit} blocks, i.e., validators speculatively commit a block $B$ upon obtaining its QC.
If the leader proposing $B$ equivocates, validators may revoke the speculative commit, potentially incurring wasted effort. However, when coupled with economic disincentives for equivocation, such provable misbehavior is likely to be a rare occurrence. Therefore, we anticipate most users will adopt the speculatively finalized state, gaining the UX benefit of eliminating at least one round trip delay. In particular, latency-sensitive applications can act on the speculatively finalized state, including
trading and market-making updates, oracle and index maintenance, and low-value
interactive workloads (in-game asset updates, NFT listings/bids,
social interactions). In contrast, settlement-critical operations spanning multiple parties or domains (bridging, custodial exchange deposits/withdrawals, treasury movements, multi-party escrows) will likely defer action until full finality. Ultimately, the latency--safety  threshold is
application-specific: each application can choose  to utilize the speculative state  or wait for full finalization, depending on its tolerance for risk and need for speed.


\medskip
\mypara{Latency.}
An immediate consequence of the speculative commitment mechanism is that validators can speculatively execute the transactions contained in a block within $1.5$ round-trip time from its proposal: $1$ round-trip (one view)  to construct the QC and $0.5$ round-trip (half a view) to disseminate it. This optimization translates into a $60\%$ latency reduction over HotStuff~\cite{HotStuff}, which requires $3.5$ round trips, and a $40\%$ reduction over faster variants~\cite{jalalzai2021fasthotstuff, Diem-BFT}, which require $2.5$ round trips.
\tobiaschange{Additionally,
\sysname achieves leader fault isolation, which limits the delays of a single Byzantine leader to a single timeout, through the use of backup QCs, which allow each honest validator to complete their own view without relying on the next leader to collect votes.
}

\medskip 
We summarize the properties of \sysname\ and compare them with existing schemes in \Cref{tab:comparison}, while a more detailed discussion of related work is provided in \Cref{sec:related}.

\medskip
\mypara{\sysname's implementation in the Monad blockchain.}
\sysname is implemented as the consensus protocol of the Monad blockchain, where all messages are authenticated, and validators are weighted based on their stakes. Additionally, the implementation utilizes RaptorCast~\cite{raptorcast} for efficient, erasure-coded block propagation, BlockSync~\cite{blocksync} for retrieving missing blocks from peers, and StateSync~\cite{statesync} for synchronizing transaction execution states when validators fall behind. 
We refer to \url{https://docs.monad.xyz} for more details.

\medskip
\mypara{Paper overview.}
The rest of this paper is organized as follows.
In Section~\ref{sec:system-prelim}, we introduce the system model, describe the properties of consensus, and present the data structures that underpin our design.
In Section~\ref{sec:Protocol_Description}, we specify and describe the \sysname\ protocol.
In Section~\ref{sec:proofintuition}, we give an intuition of how and why \sysname\ satisfies the safety, liveness, and efficiency properties introduced above; full proofs are available in Appendix~\ref{sec:analysis}.
Finally, in Section~\ref{sec:related}, we discuss related work.

\section{System Model \& Preliminaries}
\label{sec:system-prelim}

\subsection{System Model} 
\label{subsec:model}
\medskip
\mypara{Participants.} The system comprises $n = 3f+1$ parties or validators, indexed from $1$ to $n$, out of which up to $f$ are Byzantine~\cite{lamport1982byzantine}, i.e., they can deviate from the protocol and behave arbitrarily.
(Our protocol can trivially be adapted to support $n > 3f + 1$.)
Depending on the type of Byzantine behavior, we also call Byzantine validators \emph{offline} or \emph{malicious}. All other validators (called \emph{correct} or \emph{honest}) follow the protocol.

\medskip
\mypara{Cryptographic tools.}
Our protocol relies on standard cryptographic primitives, summarized below:
\begin{itemize}
  \item \emph{Digital signatures:} Each validator has a public–private key pair and knows the public keys of all other validators. 
  Messages, along with their message type, are signed with the private key, while the corresponding public key is used to verify their authenticity.
  
  \item \emph{Signature aggregation:} 
  Our protocol additionally relies on aggregate signatures~\cite{Boneh:2003} to aggregate multiple validator's signatures (on potentially different messages) into a single signature. 
  In particular, our implementation uses the pairing-based aggregate signatures from~\cite{Boneh:2003}.
\end{itemize}


\medskip
\mypara{Communication network.}
Validators communicate over point-to-point authenticated channels.
We assume that this network is \emph{reliable}, meaning that any message sent by a correct validator to another correct validator is eventually delivered. 


We assume the network to be \emph{partially synchronous}~\cite{dwork1988consensus}, where the network behaves asynchronously up to some unknown point in time, 
after which it operates synchronously. 
Formally, this model assumes the existence of an unknown Global Stabilization Time (GST) and a known delay bound $\Delta > 0$.  
Any message sent at real time $\tau$ is guaranteed to be delivered by time $\max(\tau, \GST) + \Delta$. 
We use $\delta \leq \Delta$ to denote the \emph{actual} network delay.
%

\subsection{Consensus Properties} 
\label{sub:consensus_properties}


\mypara{Block-chaining paradigm.}
The primary objective of a blockchain consensus protocol is to ensure that client requests~(transactions), organized into \emph{blocks}, are \emph{committed} (or \emph{finalized}) onto an ever-growing linear sequence of blocks.  
Concretely, each validator maintains its own \emph{local log}, but the protocol must guarantee that the local logs of all correct validators remain consistent.    
This guarantees that all correct validators maintain a consistent view of the log, equivalent to that of a single fault-free server.  
Before formalizing our protocol, we introduce the abstract notion of a block used to specify the properties of our protocol.  
The concrete block data structure and all related components will be detailed later in the paper; here, we rely only on the abstract formulation needed for the problem definition.

Our work follows the \emph{block-chaining paradigm}, in which each leader proposes a value in the form of a block that explicitly extends a previously proposed block $B$ by including the pointer of $B$, referred to as its \emph{parent}.  
This creates a unique chain over all preceding blocks in the log.  
The first block in this chain is the \emph{genesis block}, and the distance from the genesis block to a block $B$ is its \emph{height}.  
A block $B_k$ at height $k$ has the form
\[
B_k \gets \big( b_k,\, P(B_{k-1}) \big),
\]
where $b_k$ is the proposed payload at height $k$, $B_{k-1}$ is the block at height $k-1$ and $P(B_{k - 1})$ is a pointer to block $B_{k - 1}$.
(In our protocol, pointers are realized through quorum certificates; further details appear in \Cref{sub:data_structures}.)
A block $B_k$ is said to \emph{extend} a block $B_h$ if $B_k = B_h$ or $B_k$ is a descendant of $B_h$.  
All blocks on the path from the genesis block to $B_k$ are the \emph{ancestors} of $B_k$.  
In this paradigm, whenever a block $B_k$ is committed, all of its ancestors are committed as well.

\medskip
\mypara{Security properties.}
\sysname\ satisfies the following fundamental security properties:
\begin{itemize}
    \item \emph{Safety:} No two correct validators commit different blocks at the same log position (height) in their local logs.
    
    
    \item \emph{Liveness:} After GST, every block proposed by a correct validator is eventually committed by every correct validator.
\end{itemize}
Intuitively, safety ensures that all correct validators have a consistent view of the log (see Theorem~\ref{theorem:jovan_safety}), while liveness ensures that, once the network stabilizes, blocks issued by correct validators are eventually committed (see Theorem~\ref{theorem:liveness}).


Another fundamental property of leader-based pipelined blockchain consensus protocols is \emph{tail-forking resistance}. 
In its classical form, as defined in~\cite{BEEGEES}, this property guarantees that every block proposed by a correct validator after GST is eventually committed (a guarantee we refer to as \emph{liveness} in this work).
In non-pipelined protocols, this guarantee is relatively easy to achieve: an honest leader performs all the work required to drive its own block to commitment and can therefore ensure that its block is not abandoned. 
In pipelined protocols, however, this reasoning breaks down because the commitment of a block proposed in view $v$ requires the leader of view $v+1$ to collect votes for the proposal; thus, even if the leader of view $v$ behaves correctly, it cannot unilaterally guarantee that its block will eventually be committed. 
This interdependence across views makes tail-forking resistance significantly more subtle to ensure in the pipelined setting.

In this work, we adopt a strictly stronger form of the property. 
Our variant requires that any block that receives votes from at least $f + 1$ correct validators must eventually be committed, unless the leader equivocates. 
\sysname\ provides this guarantee for \emph{fresh} proposals, i.e., proposals containing blocks that have not been issued before (see~\Cref{sub:data_structures} for formal definition of fresh proposals).
Since progress is already guaranteed by the above liveness property, we formalize tail-forking resistance as the following safety property (proven in \Cref{theorem:jovan_tail_forking}):
\begin{itemize}
    \item \emph{Tail-forking resistance:} If a leader proposes a block $B$ as part of a fresh proposal and at least $f + 1$ correct validators vote for it, then—unless the leader equivocates by issuing multiple proposals—any block $B'$ committed by a correct validator must satisfy that either $B$ extends $B'$ or $B'$ extends $B$.
\end{itemize}

\sysname also incorporates \emph{speculative commitments}, and we require explicit guarantees on when such commitments may be revoked. 
A block is speculatively committed once a quorum certificate (QC) is formed for a fresh proposal containing that block. 
Since a QC aggregates $2f + 1$ votes, this implies that at least $f + 1$ correct validators voted for the proposal.
Our stronger tail-forking resistance property then \emph{directly implies} that a speculatively committed block cannot be reversed unless the issuing leader equivocates, and equivocation is a provable and punishable fault that is expected to be rare in practice. 
Thus, the stronger definition of the tail-forking resistance property simultaneously ensures that every block proposed by a correct leader after GST is eventually committed (a security property for honest leaders) and that speculative commitments can only be revoked in the presence of equivocation (a security property for applications operating on the speculatively finalized state). 
The latter property, following directly from our strengthened tail-forking resistance property, can be formalized as (and proven in \Cref{corollary:jovan_reversion}):

\begin{itemize}
    \item \emph{Reversion of speculatively committed blocks:} Suppose a block $B$ contained in a fresh proposal is speculatively committed by a correct validator, and another block $B'$ is committed by a correct validator such that neither $B$ extends $B'$ nor $B'$ extends $B$. 
    Then, the leader that proposed $B$ equivocated.
\end{itemize}

\medskip
\mypara{Efficiency properties.}
Next, we discuss the efficiency properties satisfied by \sysname. We assume that local computation has negligible time overhead.
Some of these properties depend on the high-level design of the protocol, which we will outline as needed.
We begin by formalizing the properties that characterize the protocol’s behavior in the common-case operation—that is, when the network is synchronous and all (or most) validators behave honestly.
In this setting, we require \sysname\ to operate as efficiently as the network permits, and achieve both low latency and linear communication overhead.
We define these properties formally below.
\begin{itemize}[itemsep=2pt]
    \item \emph{Leader-based optimistic responsiveness~\cite{doidge2024moonshotoptimizingchainbasedrotating}:} After GST, a correct leader of view $v$ needs to wait only until $n - f = 2f + 1$ messages of view $v-1$ of the same type (e.g., votes) have been received to guarantee that it can create a proposal that will be voted upon by all correct validators.\footnote{Note that by our liveness guarantee, the block of this proposal is guaranteed to be committed.}

    \item \emph{Optimistic responsiveness:} If all validators are correct, then during any post-GST time interval $[\tau_1, \tau_2]$, each correct validator commits $\Omega\big((\tau_2 - \tau_1) / \delta\big)$ blocks.
    Here, the interval $[\tau_1, \tau_2]$ is post-GST if $\tau_1 \geq \GST$, and $\delta$ denotes the actual network delay.
    (See \Cref{thm:optimistic_responsiveness}.)

    \item \emph{Linear communication:} When all validators are correct, after GST, each view requires only linear message complexity.
    
    \item \emph{Low latency:} After GST, a block proposed by an honest leader is speculatively finalized in $1.5$ round trips and fully finalized in $2.5$ round trips,\footnote{Full finalization in \sysname\ requires two consecutive honest leaders.} measured from the time the leader issues the proposal.
\end{itemize}
Note that the properties discussed above, such as optimistic responsiveness and linear communication, are stated under the assumption that all (or most) validators behave correctly.
However, \sysname provides strictly stronger guarantees.
Concretely, in any post-GST period that contains views led by correct leaders, new blocks are committed within $O(\delta)$ time, and only a linear number of messages are exchanged ---without requiring all validators to be correct. \sd{It is okay for now, but this part is a bit confusing. If you can achieve a stronger property, then why do you formalize a weaker one?} \jovan{Because the weaker property is easier to define. Take a look at the proofs and it will become clear that the current formulation of the stronger one is ugly.}

Next, we capture the behavior of \sysname in the presence of a faulty leader.
To characterize the behavior in faulty views, our protocol defines $\viewduration \in O(\Delta)$ as the local timeout duration after which a view is deemed failed. It depends on the fixed upper bound $\Delta$ of the network delay, whereas $\delta$, the actual network delay, may be substantially lower. 

We now introduce the \emph{leader fault isolation} property.
It ensures that, after GST, a single Byzantine leader only causes delays of a single view timeout.
This holds as an honest leader can propose and obtain a backup QC for its proposal within the same view without depending on the next leader, thereby making progress without incurring timeout delays. 
Recall from~\Cref{sec:introduction} that \sysname\ proceeds in sequential \emph{views}, each with a designated leader.
\begin{itemize}[itemsep=2pt]

 
    \item \emph{Leader fault isolation}: After GST, 
    the \emph{duration} of a view $v+1$ by a Byzantine leader, following an honest leader of view $v$ and followed by an honest leader in view $v+2$, is bound by $\viewduration + O(\delta)$. The duration of a view $v$ is the time between the first honest validator entering $v$, and the last honest validator entering a view $>v$. 
    (See \Cref{thm:leader_fault_isolation}.)
\end{itemize}

Note that the duration of a view $v+1$ by an honest leader following another honest leader in view $v$ is bound by $O(\delta)$, i.e., the view progresses at the actual network speed and does not depend on a predetermined timeout period. This is due to the backup QC mechanism, which allows the leader of view $v+1$ to build a QC independently from the leader of view $v+2$.

\subsection{Data Structures} \label{sub:data_structures}
%
In this subsection, we present the data structures we use in \sysname.

\medskip
\mypara{Views.} 
A \emph{view} is a positive integer representing a logical round of the consensus algorithm, each associated with a distinct \emph{leader} (or \emph{proposer}) responsible for issuing a block. In \sysname, leadership rotates in a round-robin fashion, ensuring no two consecutive views share the same leader.

\medskip
\mypara{Blocks and block headers.}  
A block $B$ is formally defined as
\[
    \iblock = \langle \originalview, \payload, \payloadhash, \qc, \blockhash \rangle,
\]
with the components defined as follows:
\begin{itemize}
    \item $\originalview$ is the view in which the block is first proposed.
    \item $\payload$ is the block’s content.
    \item $\payloadhash$ is the cryptographic hash of $\payload$.
    \item $\qc$ is a quorum certificate for the parent block (i.e., the block that this block extends), a cryptographic proof that $2f + 1$ validators have endorsed that block.
    \item $\blockhash$ is the cryptographic hash of $\langle \originalview, \payloadhash, \qc \rangle$ and is used to identify the block. 
\end{itemize}

A block header $\blockheader$ is identical to a block, except with $\payload$ omitted.

\medskip
\mypara{Votes \& quorum certificates.}
A vote $\ivote$ is the basic unit of validator endorsement, and is defined as:
\[
 \ivote = \langle \view,\blockhash, \blockid, \sigma \rangle,
\]
where \blockhash is defined as above, 
with the components specified as follows:

\begin{itemize}
    \item $\view$ is the view in which the vote is cast.

   \item $\blockhash$ is the hash of the block being endorsed.
    \item $\blockid$ is the identifier of the block proposal being endorsed (we will formally define it in the ``Proposals'' paragraph).

    \item $\sigma$ is the validator’s signature over all the above fields.
\end{itemize}

A quorum certificate $\iqc$ aggregates votes on a proposal from a quorum of validators, providing a collective attestation of support.  
Formally,
\[
    \iqc = \langle \view, \blockhash, \blockid, \Sigma \rangle,
\]
where $\view$, \blockhash and $\blockid$ are as defined above, and:
\begin{itemize}
    \item $\Sigma$ is an aggregated signature over $\langle \view, \blockhash, \blockid \rangle$, combining $2f+1$ distinct validators' signatures to certify that the proposal has been endorsed by a quorum.
\end{itemize}

Only votes cast in the same view and for the same proposal can be aggregated. 
We say that a QC $\iqc$ \emph{points to} a proposal $\iprop$ if and only if $\iqc.\blockid = \iprop.\blockid$. 
Each validator in \sysname\ also maintains the most recent QC that enabled it to advance to a later view.
This QC, referred to as the validator’s \emph{local high QC} and stored in $\localhighQC$, is updated whenever the validator enters a new view due to receiving a QC (see \Cref{subsection:Pacemaker}).

\medskip
\mypara{Proposals.}
A block proposal, or simply a proposal $\iprop$, encapsulates a block and is defined as
\[
    \iprop = \langle \view, \blockid, \block, \sigma, \tc, \nec \rangle,
\]
with components defined as follows:
\begin{itemize}
    \item $\view$ denotes the view in which the proposal is made.
    
    \item $\blockid$ is the hash of $\langle \block.\blockhash, \view \rangle$.
    \item $\block$ is the block being proposed; we say that ``$\Proposal$ contains $\Proposal.\block$''.
    \item $\sigma$ is the leader’s signature over $\blockid$.
    \item $\tc$ is a timeout certificate, providing a cryptographic proof that view $\view - 1$ has failed (formally defined later in the ``Timeout certificates'' paragraph). It can also be $\bot$.
    \item $\nec$ is a no-endorsement certificate, closely related to timeout certificates (formally defined below). It can also be $\bot$.
\end{itemize}

A \emph{fresh proposal} contains a new block. 
A \emph{reproposal}, in contrast, is a proposal that embeds a previous block. It corresponds to a previously issued proposal that failed to gather sufficient support during the initial issuance. Reproposals play a crucial role in ensuring tail-forking resistance. For a fresh proposal $\iprop$, $\iprop.\view = \iprop.\block.\originalview$ holds, whereas for a reproposal, $\iprop.\view > \iprop.\block.\originalview$ holds.

\medskip
\mypara{Tips.} 
The tip $\itip$ of a proposal is the proposal without the payload of the contained block. Formally,
\[
    \itip = \langle \view, \blockid, \blockheader, \sigma, \tc, \nec \rangle.
\]
Each validator maintains a \emph{local tip} in its $\localhightip$ variable, which is crucial for enforcing \sysname's tail-forking resistance.
The local tip is updated whenever the validator processes a new proposal; further details on updating the local tip appear in \Cref{sec:Protocol_Description}.


\medskip
\mypara{High tips.}  
Given a set of tips $\{\itip_1, \itip_2, \dots, \itip_j\}$, the \emph{high tip} of the set is defined as the $\itip_i$ with the largest view number, i.e., the one maximizing $\itip_i.\view$ for $1 \leq i \leq j$.
If multiple distinct tips attain the same maximal view, we break ties by the largest
$\itip.\blockheader.\qc.\view$. This corresponds to the  highest-QC rule in pipelined consensus algorithms. If there are still tied tips, the leader is free to choose one of them.

\medskip
\mypara{Timeout messages.} 
Validators expect regular progress by moving to higher views. A validator in view $v$ issues a timeout message $\itimeout$ for $v$ if a pre-specified amount of time has passed since entering $v$.
A timeout message carries either a local tip with its corresponding tip vote, or a QC in the $\qc$ field, but not both. If the validator holds a QC $\iqc$ with $\iqc.\view \geq \tip.\view$, it includes only the QC, since it is more relevant. The message also includes a certificate (QC or TC) from the previous view, which allows validators in a lower view to advance to the same view.
Formally, a timeout message is defined as
\[
 \itimeout = \langle \view, \tip, \hightipvote, \qc, \lastcer, \sigma \rangle,
\]
with components defined as follows:
\begin{itemize}
    \item $\view$ denotes the view for which the timeout message is issued.
    \item $\tip$ is the tip of the latest \emph{fresh} proposal (by view number) that the validator has voted for, i.e., the validator's local tip; can be $\bot$.
    \item $\hightipvote$ represents the validator’s vote for the included tip; can be $\bot$.
    \item $\qc$ is the last QC (by view number) observed; can be $\bot$.
    \item $\lastcer$ is the last certificate (by view number) observed by the validator, i.e., a certificate from view $\view - 1$.
    \item $\sigma$ is a digital signature on the tuple $\langle \view, \tipview, \qcview \rangle$, where if the local tip's view is lower or equal to the view of the latest QC, then $\qcview$ is the view of the latest QC and $\tipview = \bot$. Otherwise, $\qcview= \tip.\blockheader.\qc.\view$ and $\tipview=\tip.\view$.
\end{itemize}


\medskip
\mypara{Timeout certificates.}
A \emph{timeout certificate} $\itc$ is created from a quorum of $2f+1$ timeout messages for the same view. It represents agreement on a timeout condition and determines whether the leader must repropose a previous block, and if so, which one.  
To determine this, only the highest tip or QC among the timeout messages is relevant, so the TC embeds just one: whichever has the higher view (ties favor the QC).
To prove this is indeed the highest tip (resp. highest QC) among all $2f+1$ messages, the TC contains the view numbers of the local tip (resp. local high QC) from each contributing timeout message via the \tipsviews and \qcsviews vectors, allowing validators to verify the correctness of the selection.
Formally, a TC is defined as
\[
    \itc = \langle \view, \tipsviews, \hightip, \qcsviews, \highQC, \Sigma \rangle,
\]
with components defined as follows:
\begin{itemize}
    \item $\view$ is the view for which the TC is issued.
    \item $\tipsviews$ is the vector of validator-view tuples from each of the $2f+1$ timeout messages. Specifically, $\tipsviews$ is a vector of $(i, v)$ tuples for each validator $i$ whose timeout message with local tip from view $v$ contributes to the TC.
   
    \item $\hightip$ is the tip with the highest view among those received in the timeout messages; can be $\bot$. 
    \item $\qcsviews$ is the vector of validator-view pairs from each of the $2f+1$ timeout messages. 
    It is a vector of $(i, v)$ tuples for each validator $i$ that submitted a timeout message $\itimeout$, where $v$ is defined as 
    \begin{itemize}
        \item $\itimeout.\qc.\view$ if ${\itimeout.\qc} \neq \bot$, or
        \item $\itimeout.\tip.\blockheader.\qc.\view$ otherwise.
    \end{itemize}
    \item $\highQC$ is the most recent QC among those received in the timeout messages; can be $\bot$.
    \item $\Sigma$ is the aggregated signature over the timeout messages, certifying the correctness of the $\view$, $\tipsviews$ and $\qcsviews$ fields.
\end{itemize}

\medskip
\mypara{Block recovery \& no-endorsement certificates.}
If a leader needs to repropose a high tip, but does not have the full block proposal corresponding to that tip, it initiates \emph{block recovery}. \sysname guarantees that the leader is able to either obtain the requested proposal, or construct a \emph{no-endorsement certificate} (NEC), which is a message signed by $2f+1$ validators attesting that they  did not vote for the proposal. Crucially, in the latter case, no QC could have been formed for it. Consequently, the block can be safely discarded, and the leader may issue a fresh proposal extending the QC of the requested high tip, since tail-forking resistance is only relevant for blocks supported by at least $f+1$ honest validators.
Formally, we define an NEC $\inec$ as:
\[
    \inec = \langle \view, \hightipQCview, \Sigma \rangle,
\]
with components defined as follows: 
\begin{itemize}
    \item $\view$ is the view for which the NEC is issued.
    \item $\Sigma$ is the aggregate signature of $2f+1$ validators on $\langle \view, \hightipQCview \rangle$
    \item $\hightipQCview$ is the view of the QC of the requested high tip
\end{itemize}
We introduce the data structures and messages that are used in block recovery in \Cref{sub:recovery}.


\medskip
\mypara{Terminology.}
A leader \emph{equivocates} if it issues more than one distinct proposal in a view; these are  \emph{equivocated} proposals. 
Moreover, we say that a block $B_1$ is the \emph{parent} of a block $B_2$ if and only if there exists a proposal $p$ such that
\[
p.\block = B_1 \quad \text{and} \quad p.\blockid = B_2.\qc.\blockid.
\]
Since the $\blockid$ field of a proposal is determined solely by the view and the block of the proposal, any two proposals with the same $\blockid$ must contain the same block. Consequently, each block has exactly one parent.
We further say that a block $B_2$ \emph{extends} a block $B_1$ if and only if $B_1 = B_2$ or $B_2$ is a descendant of $B_1$ under the parent relation.
A block $B_2$ \emph{strictly extends} $B_1$ if and only if $B_2$ extends $B_1$ and $B_2 \neq B_1$. 

We now extend this terminology to proposals.
Concretely, we say that a proposal $p_1$ is a \emph{parent} of a proposal $p_2$ if and only if $p_1.\block$ is the parent of $p_2.\block$.
Since multiple proposals may contain the same block, a proposal $p_2$ may have more than one parent; recall that a proposal carries additional fields beyond $\blockid$.
Then, we say that a proposal $p_2$ \emph{extends} a proposal $p_1$ if and only if block $p_2.\block$ extends block $p_1.\block$.
Finally, a proposal $p_2$ \emph{strictly extends} a proposal $p_1$ if and only if block $p_2.\block$ strictly extends block $p_1.\block$.

\section{\sysname} 
\label{sec:Protocol_Description}
We begin with an overview of \sysname, and then present its two logical modes of operation, the \emph{happy} and \emph{unhappy} path in detail in \Cref{sub:happy_path,sub:unhappy_path}. 
We formally describe the protocol in \Cref{Algorithm:Consensus-Execution_1,Algorithm:Consensus-Execution_2} with helper functions in \Cref{Algorithm: Consensus-Execution-2,Algorithm:ValidationPredicates,Algorithm:Utilities}. 
In~\Cref{subsection:Pacemaker} and \Cref{Algorithm:Pacemaker}, we describe the protocol's pacemaker, a tool critical for liveness. 
In~\Cref{sub:recovery} and \Cref{Algorithm:Recovery} we describe block recovery.
Our algorithms use the keyword {\bf require} \verb|<condition>| as a shortcut for \emph{if not} \verb|<condition>| \emph{then return;} (or \emph{return $\false$} when a return value is required).

\medskip
\mypara{Protocol overview.} 
\sysname protocol is a leader-based blockchain (aka state machine replication) consensus protocol.
On the happy path, the protocol operates normally, without failures or network instabilities, with a new block being added to the chain in every view.
This process encompasses a leader collecting votes and constructing a quorum certificate for the previous proposal, creating and disseminating a new proposal that extends the previous one, and validators voting to endorse the newly issued proposal.

Validators expect regular progress. When no valid proposal arrives within the timeout period, for instance due to network instabilities or a Byzantine leader, validators suspect a failure and broadcast timeout messages to initiate the unhappy path.
Timeout messages are then aggregated by each validator into a timeout certificate, which serves two purposes: First, it proves that the view has timed out and it is safe to advance to the next view. Second, the information contained in timeout messages and aggregated in TCs prescribe a specific recovery path.
We distinguish two recovery types:
\begin{itemize}
    \item \emph{Standard recovery} involves the leader reproposing the highest tip (also called \emph{high tip}) reported in the collected timeout messages. If the leader already has the block $B$ corresponding to the tip, reproposal is straightforward. Otherwise, the leader first queries for the block $B$ using a sub-protocol that we call \emph{block recovery}.
    This mechanism either (i) recovers~$B$, allowing the leader to repropose it, or (ii) returns a proof of its unavailability in the form of an NEC, certifying that a majority of correct validators could not obtain $B$. In the latter case, the leader proposes a fresh block with the above NEC, and extends the block referenced by the QC included in the high tip. Looking ahead, upon successful recovery of $B$, reproposing $B$ is critical to ensure tail-forking resistance.
    \item \emph{Fast recovery} involves the leader obtaining or constructing a QC either from the previous view (allowing a happy-path proposal directly extending the QC), or such that the QC proves that the high tip proposal does not need to be reproposed. The leader can directly make a fresh proposal extending the block that obtained a QC.
\end{itemize}
Fast recovery streamlines recovery in the most common failure scenarios, but it is merely an optimization and not strictly necessary for \sysname's security properties.

\subsection{Happy Path} 
\label{sub:happy_path}
\begin{figure}[t!]
    \centering
    \includegraphics[width=0.8\textwidth]{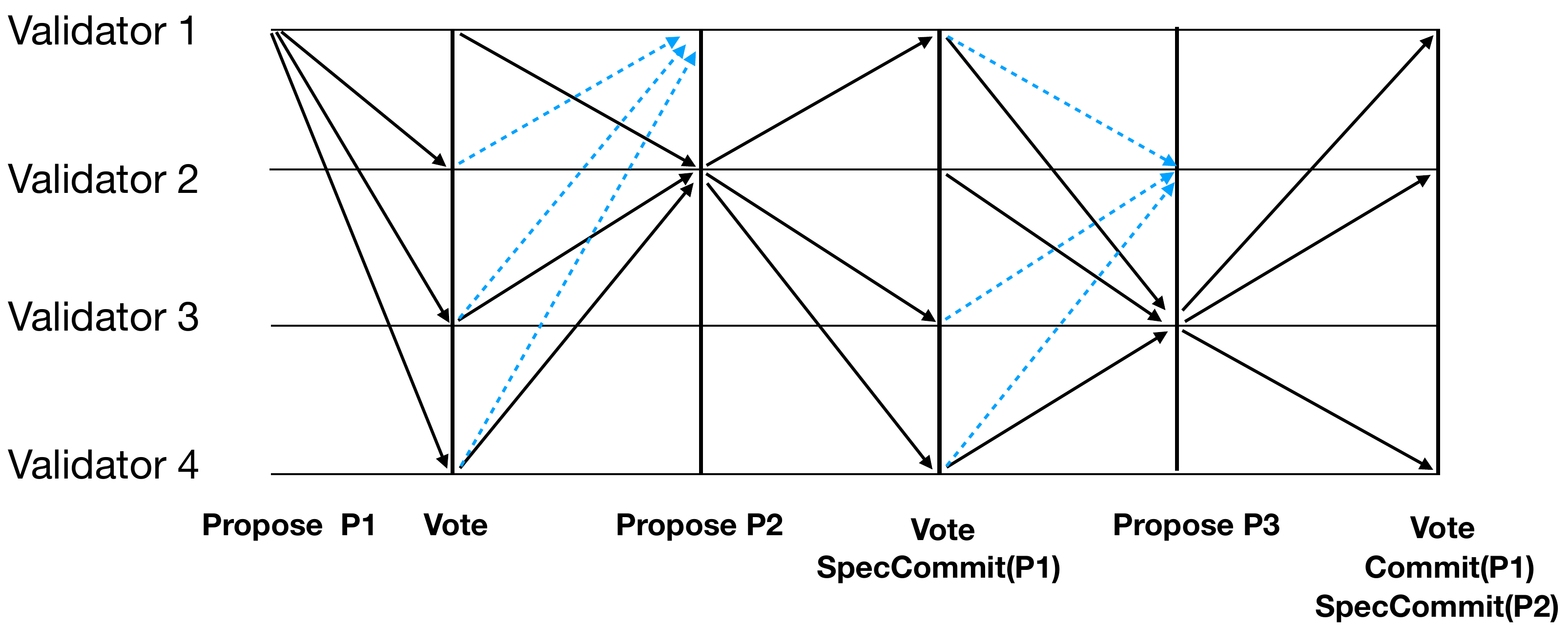}
    \caption{Happy path workflow in \sysname. Note that validator sends votes to both the current (dashed, \textcolor{Cerulean}{blue}) and the next leader (solid, \textcolor{black}{black}). Each validator speculatively commits the parent block of a proposal after observing a QC from the previous round for the parent, and commits the grandparent block if the parent and grandparent have QCs from consecutive views. For simplicity, we omit the dissemination of the backup QC, as it is not required in this example. }
    \label{fig:MonadBFT-HappyPath}
\end{figure}

We illustrate the happy path workflow of \sysname\ in~\Cref{fig:MonadBFT-HappyPath}.
All proposals in the happy path are \emph{fresh}, and directly extend the block from the previous view. Therefore, for any such proposal $p$, we have
\[
    p.\view = p.\block.\originalview = p.\block.\qc.\view + 1;~~ 
    p.\tc = p.\nec = \bot.
\]

\mypara{Description of the happy path.}
Consider a view $v$, led by leader $L_v$. $L_v$ builds or receives a QC from view $v-1$, forms a proposal $p_v$ using this QC, and broadcasts it to all validators.

Upon receiving proposal $p_v$, each validator $i$ verifies its validity using the $\SafetyCheck()$ function in ~\Cref{Algorithm:ValidationPredicates}.
Specifically, validator $i$ checks that the proposal $p_v$ was issued by the designated leader $L_v$ of view $v$. If the QC included in $p_v$ originates from the immediately preceding view $v-1$, then it is a happy path proposal, and there should not be a TC. 
Validator $i$ continues to process the proposal $p_v$ by retrieving the $\mathit{parent}$ (the proposal referenced by $p_v.\block.\qc$) and the $\mathit{grandparent}$ (the proposal referenced by $\mathit{parent}.\block.\qc$) of proposal $p_v$, and then:
\begin{enumerate}
    \item It transitions to view $v$.
    \item It (irrevocably) finalizes the  $\mathit{grandparent}$ block
    if the QCs pointing to $\mathit{grandparent}$ and $\mathit{parent}$ originate from consecutive views, i.e., $p_v.\block.\qc.\view = \mathit{parent}.\block.\qc.\view + 1$. 
    
    \item It speculatively finalizes the $\mathit{parent}$ if $\mathit{parent}$ is a fresh proposal (since tail-forking resistance only holds for fresh proposals).
    
    \item Next, validator $i$ casts its vote for proposal $p_v$ and sends the vote to \emph{both} $L_v$ and to the next leader (i.e., the leader of view $v + 1$). Additionally, validator $i$ updates its local tip to the tip of $p_v$.
\end{enumerate}
Upon constructing a QC from view $v$, the next leader $L_{v+1}$ continues along the happy path.

\medskip
\mypara{Genesis.}
The blockchain is initialized with a \emph{genesis tip} and a \emph{genesis QC} for view 0, which are valid without a signature, and are used to initialize the $\localhightip$ and $\localhighQC$ variables in each validator. 
When the protocol starts, all validators call $\Mref{Algorithm:Pacemaker}.\IncrementView(\localhighQC)$ to enter the first view, and the leader of that view additionally invokes $\Mref{Algorithm: Consensus-Execution-2}.\NextLeaderReceivesQC(\localhighQC)$ to issue its proposal.
For brevity, we omit the details on genesis and initialization of variables from the algorithms.

\begin{algorithm}[p]
\caption{\sysname: Pseudocode for validator $i$ [part 1 of 2]}
 \label{Algorithm:Consensus-Execution_1}
\newcounter{lastAlgoLine}           
\small
    \LinesNumbered

\commentline{we highlight the steps needed for view synchronization (the Pacemaker) in \pmh{blue}} \\
\textbf{Global variables:} \localhightip, \highestvotedview, \lastTC, \localhighQC, $\itimeouts$, $\ivotes$, $\itimeoutvotes$, $\pmh{\curView}$\\


\one
\Upon{\textup{receiving a proposal $\iprop$ from validator $\ileader$ with $\iprop.\view \geq \pmh{\curView}$}} {
    \label{Alg:Consensus-Execution-Receipt-of-Proposal} 
    
        \Require $\ileader$ is the leader of view $\iprop.\view$ and $\Mref{Algorithm:ValidationPredicates}.\SafetyCheck(\iprop)$; \commentinline{if proposal is invalid, reject it} \\

        \textbf{if} $\iprop.\tc \neq \bot$ \textbf{then} \pmh{$\Mref{Algorithm:Pacemaker}.\IncrementView(\iprop.\tc)$}; \textbf{else} \pmh{$\Mref{Algorithm:Pacemaker}.\IncrementView(\iprop.\block.\qc)$}; \\ \label{line:increment_view_proposal} 

        \If{$\iprop.\block.\qc.\view = \iprop.\view - 1$}{ 
            \commentline{proposal contains the QC from the preceding view} \\
            \If{\textup{$i$ is the leader of view $\iprop.\block.\qc.\view$}} {
                $\Mref{Algorithm: Consensus-Execution-2}.\CurrentLeaderReceivesQC(\iprop.\block.\qc)$; \commentinline{disseminate the backup QC} \label{line:disseminate_backup_qc_proposal} 
            }
        
            \textbf{send} $\iprop.\block.\qc$ to the leader of view $\iprop.\block.\qc.\view$, if not already sent; \label{line:forward_qc_proposal_current_leader}
        }
        
        $\Mref{Algorithm: Consensus-Execution-2}.\CommitSpecCommit(\iprop.\block.\qc)$; \commentinline{attempt commit and speculatively commit based on QC} \label{line:commit_spec_commit_proposal}

        \one
        
        \If{\textup{$\iprop.\view > \mathsf{highest\_voted\_view}$}}{ \label{Alg:Algorithm: Consensus-Execution:safety_check}

            $\localhightip \gets \Mref{Algorithm:Utilities}.\GetTip(\iprop)$;\label{Alg:update_tip} \commentinline{update the local tip} \\
            
            $\ivote \gets \Mref{Algorithm:Utilities}.\CreateVote(\curView, \localhightip.\blockheader)$; \commentinline{create a vote for the proposal} \\ \label{line:create_vote_normal}
            
            \textbf{send} $\ivote$ to the leaders of views $\pmh{\curView}$ and $\pmh{\curView} +1$; \label{line:send_vote_normal}

            $\highestvotedview \gets \pmh{\curView}$; \commentinline{update the highest view in which a vote has been issued} \label{line:highest_voted_view_update_proposal}
        }
    }



\Upon{\textup{receiving a vote $\ivote$ with $\ivote.\view \geq \pmh{\curView}$}}{ \label{Alg:consensus:Leader_happy_path}
    \Require $\Mref{Algorithm:ValidationPredicates}.\ValidVote(\ivote)$;

    \Require validator $i$ is the leader of view $\ivote.\view$ or view $\ivote.\view + 1$; \commentinline{only leaders handle votes}
    
    $\iqc \gets \Mref{Algorithm: Consensus-Execution-2}.\HandleVote(\ivote)$; \commentinline{try to form a QC from $2f{+}1$ votes}\\ \label{line:form_qc}
  
    \Require $\iqc \ne \bot$; \commentinline{do nothing if not collected enough votes yet}

    $\pmh{\Mref{Algorithm:Pacemaker}.\IncrementView(\iqc)}$; \\ \label{line:increment_view_form_qc}
    $\Mref{Algorithm: Consensus-Execution-2}.\CommitSpecCommit(\iqc)$; \\
  
    \If{\textup{validator $i$ is the leader of view $\iqc.\view$}}{
        $\Mref{Algorithm: Consensus-Execution-2}.\CurrentLeaderReceivesQC(\iqc)$; \label{line:disseminate_backup_qc}
    } \Else {
        $\Mref{Algorithm: Consensus-Execution-2}.\NextLeaderReceivesQC(\iqc)$; \label{line:propose_form_qc}
    }
}

\Upon{\textup{receiving a QC $\iqc$ from a validator $j$ with $\iqc.\view \geq \pmh{\curView}$}} {
\label{line:receive_backup_qc}
        \Require $\Mref{Algorithm:ValidationPredicates}.\ValidQC(\iqc)$; \\

        \If{\textup{$j$ is the leader of view $\iqc.\view$}} {
        $\pmh{\Mref{Algorithm:Pacemaker}.\IncrementView(\iqc)}$; \commentinline{enter the next view} \\ \label{line:increment_view_receive_qc}

        $\Mref{Algorithm: Consensus-Execution-2}.\CommitSpecCommit(\iqc)$; \\

        $\Mref{Algorithm: Consensus-Execution-2}.\NonLeaderReceivesQC(\iqc)$; \commentinline{process a backup QC} \label{line:handle_backup_qc_receive_qc}
        }

        \If{\textup{$i$ is the leader of view $\iqc.\view$}} {
                \Mref{Algorithm: Consensus-Execution-2}.$\CurrentLeaderReceivesQC(\iqc)$; \commentinline{spread the QC} \label{line:disseminate_backup_qc_2}
        }

        \If{\textup{$i$ is the leader of view $\iqc.\view + 1$}} {
                $\pmh{\Mref{Algorithm:Pacemaker}.\IncrementView(\iqc)}$; \commentinline{enter the next view} \\ \label{line:increment_view_receive_qc_next_leader}

                    $\Mref{Algorithm: Consensus-Execution-2}.\CommitSpecCommit(\iqc)$; \\

                \Mref{Algorithm: Consensus-Execution-2}.$\NextLeaderReceivesQC(\iqc)$; \commentinline{utilize the received QC to issue a new proposal} \label{line:propose_receive_qc}
        }
    }
\setcounter{lastAlgoLine}{\value{AlgoLine}}

\Upon{\textup{\pmh{timing out in view $\curView$}}}{ 
\label{line:timeout_from_view}


   $\highestvotedview \gets \pmh{\curView}$; \commentinline{this prevents voting for proposals of the current view} \label{line:update_highest_voted_view_timeout}

   $\itimeout \gets \Mref{Algorithm:Utilities}.\CreateTimeoutMsg(\pmh{\curView})$; \\
    \textbf{broadcast} $\itimeout$;  \label{line:broadcast_timeout_message}
}
\end{algorithm}

\begin{algorithm}[htbp]
\caption{\sysname: Pseudocode for validator $i$ [part 2 of 2]}
\label{Algorithm:Consensus-Execution_2}
\small
    \LinesNumbered




    
  
            

        

\Upon{\textup{receiving a timeout message $\itimeout$ with $\itimeout.\view \geq \pmh{\curView}$}} {
\label{line:receive_timeout_message}
        \Require $\Mref{Algorithm:ValidationPredicates}.\ValidTimeoutMessage(\itimeout)$; \\

            \If{\textup{$\itimeout.\lastcer$ is a TC \textbf{and} $i$ has not yet broadcast a timeout message or TC for $\itimeout.\lastcer.\view$}} {
                \textbf{broadcast} $\itimeout.\lastcer$; \label{line:broadcast_tc_timeout}
            }

        $\pmh{\Mref{Algorithm:Pacemaker}.\IncrementView(\itimeout.\lastcer)}$; \commentinline{set $\curView$ to $\itimeout.\lastcer.\view + 1$} \\ \label{line:increment_view_receive_qc_timeout_message}
        \If{\textup{$i$ is the leader of view $\itimeout.\lastcer.\view + 1$ \textbf{and} no proposal has been issued for that view 
        }} {
            \If{\textup{$\itimeout.\lastcer$ is a QC}} {
                $\Mref{Algorithm: Consensus-Execution-2}.\NextLeaderReceivesQC(\itimeout.\lastcer)$; \commentinline{utilize the QC to issue a new proposal} \label{line:propose_receive_qc_timeout_message}
            } \Else {
                $\Mref{Algorithm: Consensus-Execution-2}.\NextLeaderReceivesTC(\itimeout.\lastcer)$; \commentinline{utilize the TC to issue a new proposal}
            }
        }

        \If{\textup{$\itimeout.\lastcer$ is a QC}} {
            \If{\textup{$i$ is the leader of view $\itimeout.\lastcer.\view$}} {
                $\Mref{Algorithm: Consensus-Execution-2}.\CurrentLeaderReceivesQC(\itimeout.\lastcer)$; \label{line:disseminate_backup_qc_timeout}
            }
        
            \textbf{send} $\itimeout.\lastcer$ to the leader of view $\itimeout.\lastcer.\view$; \\ \label{line:forward_qc_timeout_message_current_leader}
            \textbf{send} $\itimeout.\lastcer$ to the leader of view $\itimeout.\lastcer.\view + 1$; \label{line:forward_qc_timeout_message}

            $\Mref{Algorithm: Consensus-Execution-2}.\CommitSpecCommit(\itimeout.\lastcer)$; \\

        }

        \one
        $\itc, \iqc \gets \Mref{Algorithm:Pacemaker}.\pmh{\HandleTimeout(\itimeout)}$; \commentinline{returns a TC or a QC, or neither}\\ \label{line:form_tc_qc}
        

        \If{$\iqc \neq \bot$} {
            $\pmh{\Mref{Algorithm:Pacemaker}.\IncrementView(\iqc)}$; \commentinline{further timeout messages from same view will be ignored after this} \\ \label{line:increment_view_form_qc_timeout_message}

            $\Mref{Algorithm: Consensus-Execution-2}.\CommitSpecCommit(\iqc)$; \\

            \If{\textup{$i$ is the leader of view $\iqc.\view + 1$}} {
                $\Mref{Algorithm: Consensus-Execution-2}.\NextLeaderReceivesQC(\iqc)$; \commentinline{utilize the constructed QC to issue a new proposal} \label{line:propose_form_qc_timeout_message}
            }
        } \If{$\itc \neq \bot$} {
            $\pmh{\Mref{Algorithm:Pacemaker}.\IncrementView(\itc)}$; \commentinline{further timeout messages from same view will be ignored after this} \\ \label{line:increment_view_form_tc}
            \If{\textup{$i$ is the leader of view $\itc.\view + 1$}} {
                $\Mref{Algorithm: Consensus-Execution-2}.\NextLeaderReceivesTC(\itc)$; \commentinline{utilize the constructed TC to issue a new proposal}
            }
        }
    }

\one 
\Upon{\textup{receiving a TC $\itc$ with $\itc.\view \geq \pmh{\curView}$}} {
\label{line:receive_tc}
    \Require $\Mref{Algorithm:ValidationPredicates}.\ValidTC(\itc)$; \\
    $\pmh{\Mref{Algorithm:Pacemaker}.\IncrementView(\itc)}$; \\ \label{line:increment_view_receive_tc}
    \If{\textup{no TC and no timeout message broadcast for view $\itc.\view$}} {
        \textbf{broadcast} $\itc$; \label{line:broadcast_tc_receive_tc}
    }

    \If{\textup{$i$ is the leader of view $\itc.\view + 1$}} {
        $\Mref{Algorithm: Consensus-Execution-2}.\NextLeaderReceivesTC(\itc)$; \commentinline{utilize the received TC to issue a new proposal}
    }
}

\one 

\Upon{\textup{receiving a proposal request $\msg{\ProposalRequest, \itc}$ from validator $L$ with $\itc.\view+1 \geq \pmh{\curView}$}} {
\label{line:receive_proposal_request}
    \Require $\Mref{Algorithm:ValidationPredicates}.\ValidTC(\itc)$ \textbf{and} $\itc.\hightip \neq \bot$ \textbf{and} $L$ is the leader of view $\itc.\view+1$; \\
    $\pmh{\Mref{Algorithm:Pacemaker}.\IncrementView(\itc)}$; \\ \label{line:increment_view_proposal_request}
    \Mref{Algorithm:Recovery}.$\ProcessProposalRequest(\itc)$;
}

\Upon{\textup{receiving an NE request $\msg{\NERequest, \itc}$ from validator $L$ with $\itc.\view+1 \geq \pmh{\curView}$}} {
\label{line:receive_ne_request}
    \Require $\Mref{Algorithm:ValidationPredicates}.\ValidTC(\itc)$ \textbf{and} $\itc.\hightip \neq \bot$ \textbf{and} $L$ is the leader of view $\itc.\view+1$; \\
    $\pmh{\Mref{Algorithm:Pacemaker}.\IncrementView(\itc)}$; \\ \label{line:increment_view_ne_request}
    \Mref{Algorithm:Recovery}.$\ProcessNERequest(\itc)$;
}


\setcounter{lastAlgoLine}{\value{AlgoLine}}
\end{algorithm}

\begin{algorithm} [htbp]
    \DontPrintSemicolon
    \caption{\sysname: Pacemaker}
    \label{Algorithm:Pacemaker}

    \SetKwFunction{UpdateQC}{\sc UpdateQC}

    \Fn{$\IncrementView(\icer)$}{ \label{alg:Increment_view}
        \Require $\icer.\view \geq \curView$; \\
        stop local timer for $\curView$;\\
        $\curView \gets \icer.\view + 1 $;\\ \label{line:enter_view}
        start local timer for $\curView$;

        \textbf{if} $\icer$ is a QC \textbf{then} $\localhighQC \gets \icer$; \textbf{else} $\lastTC \gets \icer$;

    }

\Fn(\commentif{precondition guaranteed by caller: $\textup{\Mref{Algorithm:ValidationPredicates}}.\ValidTimeoutMessage(\itimeout)$}){$\HandleTimeout(\itimeout)$}{ 


    $\iview \gets \itimeout.\view$; \\
    \If{$\itimeout.\tip \neq \bot$}{
        $\iqc \gets \Mref{Algorithm: Consensus-Execution-2}.\HandleTimeoutVote(\itimeout.\hightipvote)$; \commentinline{attempt fast recovery by building QC from tip votes carried by timeout messages} \\
        \If{$\iqc \neq \bot$} {
            \Return $\bot, \iqc$; \commentinline{QC is from view $\itimeout.\view$}
        }
    }

    $\itimeouts[\iview] \gets \itimeouts[\iview] \cup \{ \itimeout \}$; \\
    \If{\textup{$|\itimeouts[\iview] | = f + 1$}}{
        \textbf{trigger} the timeout event for $\itimeout.\view$ $(= \curView)$, if not already triggered; \commentinline{Bracha mechanism} \\ \label{line:premature_timeout}
    }

    \If{$|\itimeouts[\iview]| < 2f + 1$} {
        \Return $\bot, \bot$; \commentinline{not enough timeouts yet to build a TC}
    }

    \commentline{build a TC using timeout messages; only done once per view, as \IncrementView is called on }\\
    \commentline{the resulting TC, and this function is only called for messages from $\curView$}\\
    $\itips, \iqcs \gets \text{maps from validators to views, initialized to map all validators to $0$}$; \\
    \For{\textup{\textbf{each}} $\imsg \in \itimeouts[\iview]$}{
        \If{$\imsg.\tip \neq \bot$}{
            $\itips[\text{the sender of } \imsg] \gets \imsg.\tip.\view$; \\
            $\iqcs[\text{the sender of } \imsg] \gets \imsg.\tip.\blockheader.\qc.\view$; 
        } \Else {
            $\iqcs[\text{the sender of } \imsg] \gets \imsg.\qc.\view$;
        }
    }
    \commentline{compare maximum reported tip view and maximum reported qc view}\\
    \If(\commentif{TC created with high tip}){$\max(\itips) > \max(\iqcs)$}{
        $\iindex \gets$ a validator such that 
$\itips[\iindex] = \max(\itips)$ 
and no validator $j$ satisfies 
$\itips[j] = \max(\itips)$ 
and \hphantom{$\iindex \gets$}$\iqcs[j] > \iqcs[\iindex]$; 
\commentinline{select validator with highest tip and highest QC among ties} \\

        $\imsg \gets \text{the timeout message from validator $\iindex$ in $\itimeouts[\iview]$}$; \\
        \Return\  \Mref{Algorithm:Utilities}.\CreateTC$(\iview, \itips, \imsg.\tip, \iqcs, \bot, \itimeouts[\iview]), \bot$;

    } \Else(\commentif{TC created with high QC}) {
        $\iindex \gets$ a validator such that 
$\iqcs[\iindex] = \max(\iqcs)$; \\
        $\imsg \gets \text{the timeout message from validator $\iindex$ in $\itimeouts[\iview]$}$; \\
        \Return\  \Mref{Algorithm:Utilities}.\CreateTC$(\iview, \itips, \bot, \iqcs, \imsg.\qc, \itimeouts[\iview]), \bot$;
    }

}
   

    
\end{algorithm}

\begin{algorithm}[p]
\caption{\sysname: Helper Functions [part 1 of 3]}
 \label{Algorithm: Consensus-Execution-2}
\LinesNumbered

\Fn{$\IsFreshProposal(\iprop)$}{
    \Return $\iprop.\block.\qc.\view = \iprop.\view - 1$ \textbf{or} $\iprop.\nec \neq \bot$ \textbf{or} ($\iprop.\tc \neq \bot$ \textbf{and} $\iprop.\tc.\highQC \neq \bot$);

}

    \Fn(\commentif{precondition guaranteed by caller: $\textup{\Mref{Algorithm:ValidationPredicates}}.\ValidQC(\iqc)$}){$\CommitSpecCommit(\iqc)$}{
        $\iparent \gets \text{a proposal to which $\iqc$ points}$; \commentinline{$\iparent.\blockid = \iqc.\blockid$} \\
        \If{$\IsFreshProposal(\iparent)$}{
            \textbf{speculatively commit} $\iparent.\block$; 
        }

        $\igparent \gets \text{a proposal to which $\iparent.\block.\qc$ points}$; \\
        \If{$\iqc.\view = \iparent.\block.\qc.\view + 1$} {
            \commentline{if $\iparent$ and $\igparent$ are from consecutive views, commit $\igparent$} \\
            \textbf{commit} $\igparent.\block$;
            \label{line:commit_grandparent_1}
        }
    }

\Fn(\commentif{precondition guaranteed by caller: $\textup{\Mref{Algorithm:ValidationPredicates}}.\ValidQC(\iqc)$}){$\NextLeaderReceivesQC(\iqc)$\label{Func:NextLeaderReceivesQC}}{
    $\iprop \gets \Mref{Algorithm:Utilities}.\CreatePrepareMsg(\curView, \bot, \iqc, \bot, \bot)$; \commentinline{produce a fresh proposal}\\
    
    \textbf{broadcast} $\iprop$; 
}

\Fn(\commentif{precondition guaranteed by caller: $\textup{\Mref{Algorithm:ValidationPredicates}}.\ValidTC(\itc)$}){$\NextLeaderReceivesTC(\itc)$}{    
    
        \If{\textup{$\itc.\highQC \neq \bot$}} { \label{alg:consensus:createproposalfromhighqc}
            \commentline{produce a fresh proposal via the high QC}
            $\iprop \gets \Mref{Algorithm:Utilities}.\CreatePrepareMsg(\tt \curView, \bot, \itc.\highQC, \itc,\bot)$;
        }
        
        \ElseIf{\textup{the block $\iblock$ with $\iblock.\blockhash = \itc.\hightip.\blockheader.\blockhash$ is available}} { \label{alg:consensus:payload_present}
            $\iprop \gets \Mref{Algorithm:Utilities}.\CreatePrepareMsg(\tt \curView, \iblock, \bot, \itc,\bot)$; \commentinline{produce a reproposal of the high tip} \\
                
        }\Else{ 
            $(\iprop_\textit{old}, \inec) \gets$ $\Mref{Algorithm:Recovery}.\Recover(\itc)$;\label{alg:consensus:payloadli_request} \commentinline{send request to recover the missing high tip block}\label{alg:consensus:sendrequest}\\
            
            \If{$\inec \neq \bot$}{
                \label{alg:consensus:NEC-Formed}

                \commentline{the block could not be retrieved — produce a fresh proposal}
                
                $\iprop \gets \Mref{Algorithm:Utilities}.\CreatePrepareMsg(\tt \curView, \bot, \itc.\hightip.\blockheader.\qc, \itc, \mathsf{nec})$; \\
        
            } \Else {
                \commentline{the block could be retrieved — repropose it}
    
                $\iprop \gets \Mref{Algorithm:Utilities}.\CreatePrepareMsg(\tt \curView, \iprop_{\textit{old}}.\block, \bot, \itc, \bot)$; \label{alg:consensus:Reproposal} \commentinline{produce a reproposal of the high tip} \\
            }
    }
    \textbf{broadcast} $\iprop$; 
}
\Fn(\commentif{precondition guaranteed by caller: $\textup{\Mref{Algorithm:ValidationPredicates}}.\ValidVote(\ivote)$}){$\HandleVote(\ivote)$}{
    
    $\iview \gets \ivote.\view$; $\iproposalid \gets \ivote.\blockid$; \\
    $\ivotes[\iproposalid] \gets \ivotes[\iproposalid] \cup \{ \ivote \}$; \commentinline{individual votes}\\

    \If{$|\ivotes[\iproposalid]| = 2f + 1$} {
        $\iqc \gets \Mref{Algorithm:Utilities}.\CreateQC(\ivotes[\iproposalid])$; \\
        \Return $\iqc$;
    }
    \Return $\bot$;
    }

\Fn(\commentif{precondition guaranteed by caller: $\textup{\Mref{Algorithm:ValidationPredicates}}.\ValidVote(\ivote)$}){$\HandleTimeoutVote(\ivote)$}{
    \commentline{identical to \HandleVote, but exclusively for votes embedded in timeout messages}\\
    $\iview \gets \ivote.\view$; $\iproposalid \gets \ivote.\blockid$; \\
    $\itimeoutvotes[\iproposalid] \gets \itimeoutvotes[\iproposalid] \cup \{ \ivote \}$; \commentinline{votes embedded in timeout messages}\\

    \If{$|\itimeoutvotes[\iproposalid]| = 2f + 1$} {
        $\iqc \gets \Mref{Algorithm:Utilities}.\CreateQC(\itimeoutvotes[\iproposalid])$; \\
        \Return $\iqc$;
    }
    \Return $\bot$;
    }

    \Fn{$\CurrentLeaderReceivesQC(\iqc)$}{
    \label{Func:DisseminateBackupQC}
    \textbf{broadcast} $\iqc$, if not already broadcast;
}

    \Fn{$\NonLeaderReceivesQC(\iqc)$}{
        \textbf{send} $\iqc$ to the leader of view $\iqc.\view + 1$, if not already sent;
        
        

       
    }

\end{algorithm}

\begin{algorithm}[p]
\DontPrintSemicolon
\caption{\sysname: Helper Functions [part 2 of 3, continues on next page]}
 \label{Algorithm:ValidationPredicates}


\Fn{$\ValidBlock(\iblock)$} {
    \Return $\iblock.\payloadhash = \Hash(\iblock.\payload)$ \textbf{and} $\ValidBlockHeader(\text{the block header of $\iblock$})$; \\
}

\one

\Fn{$\ValidBlockHeader(\iblockheader)$} {
    \Return $\ValidQC(\iblockheader.\qc)$ \textbf{and} $\iblockheader.\blockhash = \Hash(\iblockheader.\originalview, \iblockheader.\payloadhash, \iblockheader.\qc)$;
}

\one

\Fn{$\ValidVote(\ivote)$}{
\Require $\ivote.\blockid = \Hash(\ivote.\blockhash, \ivote.\view)$; \\
\Return $\ivote.\sigma$ \textup{is a valid signature on $\tupled{\ivote.\view,\ivote.\blockhash,\ivote.\blockid}$};
}

\one

\Fn{\ValidQC$(\iqc)$}{
\Require $\iqc.\blockid = \Hash(\iqc.\blockhash, \iqc.\view)$; \\
\Return $\iqc.\Sigma$ \textup{aggregates $2f+1$ valid signatures on $\tupled{\iqc.\view,\iqc.\blockhash,\iqc.\blockid}$};
}

\one

\Fn{$\SafetyCheck(\iprop)$}{ \label{line:safety_check}
    \Require $\ValidBlock(\iprop.\block)$ \textbf{and} $\iprop.\blockid = \Hash(\iprop.\block.\blockhash, \iprop.\view)$; \\
    \Require $\iprop.\sigma$ is a valid signature on $\iprop.\blockid$ of the leader of view $\iprop.\view$;


    \Require $\iprop.\view \geq \iprop.\block.\originalview$ \textup{\textbf{and}} $p.\view > p.\block.\qc.\view$;
    \label{line:jovan_safety_check_proposal}

    \If{\textup{\Mref{Algorithm: Consensus-Execution-2}.$\IsFreshProposal(\iprop)$}} {
        
        \Return $\ValidTip(\text{the tip of } \iprop)$; \commentinline{verify the validity of the tip} \label{line:jovan_valid_tip_fresh_proposal}
    }

        \commentline{not fresh; verify that $\iprop.\tc.\hightip$ has been reproposed} \\
    \Require $\iprop.\view > \iprop.\block.\originalview$; \\ \label{line:jovan_check_block_view_reproposal}    
    \Require \textup{$\iprop.\tc \ne \bot$ \textbf{and} $\ValidTC(\iprop.\tc)$ \textbf{and} $\iprop.\view = \iprop.\tc.\view + 1$ \textbf{and} $\iprop.\tc.\hightip \neq \bot$};
    \label{line:valid_proposal_tc_check}

         \Return the block header of $\iprop.\block$ is equal to $\iprop.\tc.\hightip.\blockheader$; \label{alg:safetycheck:block_hash_check}

    }

\one


\Fn{$\textsc{\ValidTip}(\itip)$}{
    \Require $\ValidBlockHeader(\itip.\blockheader)$; \\
    \Require $\itip.\blockid = \Hash(\itip.\blockheader.\blockhash, \itip.\view)$; \\
    \Require $\itip.\sigma$ is a valid signature on $\itip.\blockid$ of the leader of view $\itip.\view$;



    \Require $\itip.\blockheader.\originalview = \itip.\view$; \label{line:jovan_require_block_view_fresh}

    \Require$\itip.\view > \itip.\blockheader.\qc.\view$;
    \label{line:jovan_safety_check_tip_2}

\If{\textup{$\itip.\view = \itip.\blockheader.\qc.\view+1$}}{\label{alg:line:safetycheck_extends_QC_1}
                \Return $\itip.\tc = \bot$ \textbf{and} $\itip.\nec = \bot$; \commentinline{tip constructed on the happy path}

    }


    \Require $\itip.\tc \ne \bot$ \textup{\textbf{and}} $\ValidTC(\itip.\tc)$ \textbf{and} $\itip.\view = \itip.\tc.\view + 1$;\\
        \If{\textup{$\itip.\nec = \bot$}}{ \label{line:check_tc_safety_check}
         \label{alg:line:safetycheck_extends_QC_2}
         \commentline{tip constructed using the high QC from the TC}

    \Return$\itip.\tc.\highQC \neq \bot$ \textbf{and} $\itip.\blockheader.\qc = \itip.\tc.\highQC$; \\

    }

        \Require \ValidNEC($\itip.\nec$); \commentinline{otherwise, the fresh tip is constructed with an NEC} \\
        \Return $\itip.\view = \itip.\nec.\view$ \textbf{and} $\itip.\blockheader.\qc.\view = \itip.\nec.\hightipQCview$;


 
}


\one
\Fn{$\ValidTimeoutMessage(\itimeout)$}{
    \textup{\textbf{if} $\itimeout.\lastcer$ is a QC \textbf{then require} $\ValidQC(\itimeout.\lastcer)$; \textbf{else require} $\ValidTC(\itimeout.\lastcer)$;} \\
    \Require $\itimeout.\lastcer.\view = \itimeout.\view - 1$; \\

    \Require $(\itimeout.\tip \neq \bot \textup{\textbf{ and }} \itimeout.\qc = \bot) \textup{\textbf{ or }} (\itimeout.\tip = \bot \textup{\textbf{ and }} \itimeout.\qc \neq \bot)$; \commentinline{must have exactly one}

    \If{$\itimeout.\tip \neq \bot$}{
        \Require $\itimeout.\hightipvote \neq \bot$ \textbf{and} $\ValidVote(\itimeout.\hightipvote)$ \textbf{and} $\itimeout.\hightipvote.\view = \itimeout.\view$; \\
        \Require $\ValidTip(\itimeout.\tip)$ \textbf{and} $\itimeout.\tip.\view \leq \itimeout.\view$; \\
        \Return $\itimeout.\sigma$ is a valid signature on $\tupled{\itimeout.\view, \itimeout.\tip.\view, \itimeout.\tip.\blockheader.\qc.\view}$;
    }

    \Require $\ValidQC(\itimeout.\qc) \textup{\textbf{ and }} \itimeout.\qc.\view < \itimeout.\view$; \\
    \Return $\itimeout.\sigma$ is a valid signature on $\tupled{\itimeout.\view, \bot, \itimeout.\qc.\view}$;
}

    


        






\end{algorithm}

\setcounter{algocf}{4}

\begin{algorithm}[t]
\DontPrintSemicolon
\caption{\sysname: Helper Functions [continuation of part 2 of 3]}

\setcounter{AlgoLine}{44}

\one
\Fn{$\ValidTC(\itc)$}{
    \Require $\itc.\hightip \neq \bot$ \textbf{or} $\itc.\highQC \neq \bot$; \\
    \Require $\itc.\Sigma$ aggregates $2f + 1$ valid signatures on $\itc.\view$, $\itc.\tipsviews$ and $\itc.\qcsviews$; \\

    \If{$\itc.\highQC \neq \bot$} {
        \Require $\itc.\hightip = \bot$ \textbf{and} $\ValidQC(\itc.\highQC)$ \textbf{and} $\itc.\highQC.\view < \itc.\view$; \\
        \Return $\max(\itc.\tipsviews) \leq \itc.\highQC.\view$ 
        \textup{\textbf{and}} $\itc.\highQC.\view = \max(\itc.\qcsviews)$;
    }

    \Require $\ValidTip(\itc.\hightip)$ \textbf{and} $\itc.\hightip.\view \leq \itc.\view$;
    
          $\itips \gets \tc.\tipsviews$; $\iqcs \gets \itc.\qcsviews$; \\

        for every validator $i$ such that $\itips[i] \neq \bot$, \Require $\iqcs[i] < \itips[i]$; \\

        \Require \textup{\textbf{not} exists $i$ such that $\itips[i] > \itc.\hightip.\view$};\commentinline{no more recent tip than $\itc.\hightip$}

        \commentline{ties must be broken by view of QC}
        
        \Require \textup{\textbf{not} exists $i$ such that $\itips[i] = \itc.\hightip.\view$ and $\iqcs[i] > \itc.\hightip.\blockheader.\qc.\view$};

        \Return $\itc.\hightip.\view > \max(\iqcs)$;
        

}



\Fn{$\ValidNEC(\inec)$} {
    \Require $\inec.\hightipQCview < \inec.\view - 1$; \\
    \Return $\inec.\Sigma$ aggregates $2f+1$ valid signatures on $\tupled{\inec. \view, \inec.\hightipQCview}$;
}

\end{algorithm}

\begin{algorithm}[p]
\DontPrintSemicolon
\caption{\sysname: Helper Functions [part 3 of 3]}
 \label{Algorithm:Utilities}
    \Fn{\CreatePrepareMsg{$\iview, \iblock, \iqc, \itc, \inec$}}{
    
        \If{\textup{$B = \bot$}} {
            \commentline{fresh proposal request; build a new block} \\
            
            $\iblockview \gets \iview$; \\
            
            $\ipayload \gets $ fill with mempool transactions;\\
            $\ipayloadhash \gets \Hash(\ipayload)$; \\   
            
            $\iblockhash \gets \Hash(\iblockview, \ipayloadhash, \iqc)$; \\
            
            $B \gets \langle \iblockview, \ipayload, \ipayloadhash, \iqc, \iblockhash \rangle $;
        }
        
        $\iproposalid \gets \Hash(\iblock.\blockhash, \iview)$;\\
        $\sigma \gets \Sign(\iproposalid)$;\\
        \Return $\langle \iview, \iproposalid, \iblock, \sigma, \itc, \inec \rangle $; 
    }
    
    \one
    \Fn{$\GetTip(\iprop)$}{
    \label{line:function_get_tip}
        \If{$\textup{\Mref{Algorithm: Consensus-Execution-2}}.\IsFreshProposal(\iprop)$} {
            \Return the tip of $\iprop$;
        }
        \Return $\iprop.\tc.\hightip$;  
    }

    \one
    \Fn{\CreateVote{$\iview, \iblockheader$}}{
    \label{line:function_create_vote}
        $\iproposalid \gets \Hash(\iblockheader.\blockhash, \iview)$;\\
        $\sigma \gets \Sign(\iview,  \iblockheader.\blockhash, \iproposalid)$;\\
    
        \Return $\langle \iview,  \iblockheader.\blockhash, \iproposalid, \sigma \rangle$; 
    }

    \one
    \Fn{$\CreateQC(\ivotesproposalview)$}{
        $\Sigma \gets$ aggregate signatures from all votes in $\ivotesproposalview$; \\
        $\ivote \gets$ any vote in $\ivotesproposalview$; \\
        \Return $\langle \ivote.\view, \ivote.\blockhash, \ivote.\blockid, \Sigma \rangle$;  
    }

    \one
    \Fn{$\CreateTC(\iview, \itips, \ihightip, \iqcs, \ihighqc, \itimeouts)$}{
        \Return $\langle \iview, \itips, \ihightip, \iqcs, \ihighqc, \text{aggregate signatures from all messages in $\itimeouts$} \rangle$;

    }

\one
\Fn{\CreateTimeoutMsg($\iview$)}{
\label{line:view_certificate_create_timeout_message}
    \If{$\localhighQC.\view = \iview - 1$} {
        $\iviewcer \gets \localhighQC$;
    } \Else {
        $\iviewcer \gets \lastTC$; 
    }



    \If{$\localhightip.\view \leq \localhighQC.\view$}{
        $\sigma \gets \textsc{Sign}(\iview, \bot, \localhighQC.\view)$;

        \Return $\langle
            \iview,
            \bot,
            \bot,
            \localhighQC,
             \iviewcer,
            \sigma
        \rangle$;
    }\Else{
        $\ivote \gets \CreateVote(\iview,\; \localhightip.\blockheader)$; \label{line:create_vote_timeout}

        $\sigma \gets \textsc{Sign}(\iview,\; \localhightip.\view,\; \localhightip.\blockheader.\qc.\view)$;

        \Return $\langle
            \iview,\;
            \localhightip,\;
            \ivote,\;
            \bot,\;
            \iviewcer,\;
            \sigma
        \rangle$;
    }
}

\one
    \Fn{$\CreateNEC(\iview, \iqcview, \mathit{sigs})$}{
        \Return $\langle \iview, \iqcview, \text{aggreagate signatures from $\mathit{sigs}$} \rangle$;
    }
\end{algorithm}

  

%
\begin{algorithm}[t!]
\DontPrintSemicolon
\caption{\sysname: Block Recovery}
\label{Algorithm:Recovery}
    \one 
    \Fn{$\Recover(\itc)$}{
        \label{Algoritm:Recover}
        $X \gets$ any set of $\batch$ validators; \commentinline{optimization: prefer validators who reported $\itc.\hightip$} \\ 
        {\bf send} $\msg{\ProposalRequest, \itc}$ to all validators in $X$; \\
        $\NEset \gets \emptyset$; $\requested \gets X$; \\
        {\bf broadcast} $\msg{\NERequest, \itc}$ to all validators; \\

        \one
        \commentline{function proceeds asynchronously until a proposal or NEC is returned, or $\Recover$ is called with}\\
        \commentline{a TC from a higher view}\\

        \Upon{\textup{receiving $\msg{\ProposalResponse, p}$}} {
            \label{Alg:Basic_DA_blockid_match} 
            \Require $p.\blockid = \itc.\hightip.\blockid$; \commentinline{if false, drop message, but don't return} \\
            \Return $(p, \bot)$;
        }

        \one
        \Upon{\textup {first receipt of $\ine_i = \msg{\NE, \sigma_i}$ from validator $j$}}{
            \commentline{if condition is false, drop message, but don't return} \\
            \Require $\sigma_i$ is a valid signature on $\tupled{\itc.\view+1, \itc.\hightip.\blockheader.\qc.\view}$ of validator $j$; \\ 
            $\NEset \gets \NEset \cup \{ \ine_i\}$; \\
            \If{$|\NEset| = 2f + 1$}{
                \Return $(\bot, \Mref{Algorithm:Utilities}.\CreateNEC(\itc.\view + 1, \itc.\hightip.\blockheader.\qc.\view, \NEset))$;
            }
            
        }

        \one
        \Upon{\textup {at regular intervals of duration $\intervalduration$}}{
            $X \gets \text{any $\batch$-sized subset of $($the set of all validators} \setminus  \requested)$; \\
            {\bf send} $\msg{\ProposalRequest, \itc}$ to validators in $X$; \\
            $\requested \gets \requested \cup X$; 
        }
    }
        
    \Fn{$\ProcessProposalRequest(\itc)$}{
        \label{Algorithm:ProcessProposalRequest}
        \If{\textup{$i$ has cached a proposal $p$ with $p.\blockid = \itc.\hightip.\blockid$}\label{Alg:Has_voted_for_hightip}}{
            {\bf send} $\msg{\ProposalResponse, p}$ to the leader of view $\itc.\view + 1$;
        }
    }

    \Fn{$\ProcessNERequest(\itc)$}{
        \label{Algorithm:ProcessNERequest}
        \If{\textup {$i$ has not voted for a proposal $p$ with $p.\blockid = \itc.\hightip.\blockid$ \textbf{and} has not sent an $\NE$ message in view $\itc.\view + 1$
        } \label{Has_not_voted_for_higthip}}{
            $\sigma \gets \Sign (\itc.\view+1, \itc.\hightip.\blockheader.\qc.\view)$; \\ 
            {\bf send} $\msg{\NE, \sigma}$ to the leader of view $\itc.\view + 1$;
        }
    }
\end{algorithm}

\begin{figure}[t!]
    \centering
    \includegraphics[width=0.72\linewidth]{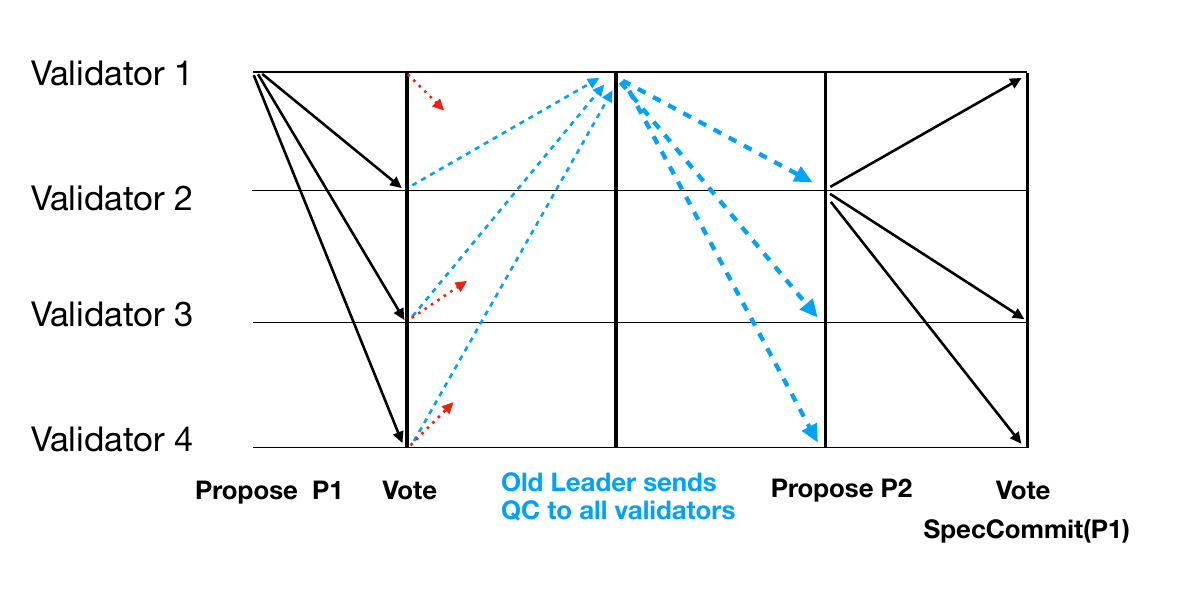}  
    \caption{Happy path workflow in \sysname using backup QC. In this scenario, votes sent to the current leader, Validator 2, are arbitrarily delayed (dotted, \textcolor{red}{red}), preventing the formation of a QC. In this case, the old leader, Validator 1, will construct a QC using the backup votes (dashed, \textcolor{Cerulean}{blue}) and broadcast it to all validators, including the next leader (bold dashed, \textcolor{Cerulean}{blue}). This enables the next leader to extend the proposal. Although some messages are lost, this example belongs to the \emph{happy path}, as a QC is successfully constructed and a timeout is avoided. Validators 3 and 4 also forward the backup QC received from Validator 1 to Validator 2. For simplicity, we omit this, as it is not required in this example.}
    \label{fig:MonadBFT-Unhappy-Backup-QC}
\end{figure}
\medskip
\mypara{Achieving leader fault isolation.}
Since the backup QC mechanism that is crucial for leader fault isolation already plays a role on the happy path, we describe it here.
%
Recall that $v$ denotes the current view with leader $L_v$, and $p_v$ is the proposal by $L_v$. Let $L_{v+1}$ denote the leader of the next view $v{+}1$.

To achieve the \emph{leader fault isolation} property (see~\Cref{sub:consensus_properties}), we introduce the notion of a \emph{backup QC}. A backup QC is a quorum certificate that a leader obtains for the block it proposes. To enable $L_v$ to form such a backup QC, we augment \sysname\ with the following modifications:
\begin{itemize}
    \item[(1)] Each validator, upon receiving $p_v$ from $L_v$, also sends its vote for $p_v$ back to $L_v$.  
    Once $L_v$ collects $2f{+}1$ votes, it constructs a \emph{backup QC} and broadcasts it to all validators.
\end{itemize}
Intuitively, this change prevents a malicious $L_{v+1}$ from forcing validators to timeout in view $v$ by withholding its proposal. This also provides an alternative way for a QC to be formed and disseminated as illustrated in \Cref{fig:MonadBFT-Unhappy-Backup-QC}.

However, sending votes to $L_v$ alone does not always ensure that $L_v$ can construct a backup QC. For instance, a malicious $L_{v+1}$ might prematurely trigger view changes by sending its proposal $p_{v+1}$ to a subset of honest validators, causing them to advance to view $v{+}1$ before voting for $p_v$. Consequently, $L_v$ would fail to obtain a backup QC due to insufficient votes. To address this, we make the following additional change:
\begin{itemize}
    \item[(2)] Each validator, upon receiving a proposal $p_{v+1}$ from $L_{v+1}$ that extends $p_v$ via a QC from view $v$, forwards the QC $p_{v+1}.\block.\qc$ to $L_v$.  
    Upon receiving this QC, $L_v$ treats it as its backup QC and broadcasts it to all validators.
x\end{itemize}

Lastly, we add the following changes to ensure the leader fault isolation property when $L_v$ is Byzantine but $L_{v+1}$ is honest:
\begin{itemize}
    \item[(3)] Each validator, upon receiving a backup QC for $p_v$, increments its view, and forwards the QC   to $L_{v+1}$ unless it has already received a valid proposal from $L_{v+1}$.
\end{itemize}
This rule prevents a malicious $L_v$ from triggering a view change at all validators (except the leader) to view $v{+}1$ much earlier than $L_{v+1}$, which would effectively shorten $L_{v+1}$'s timeout window (view duration) and could lead to unnecessary timeouts in view $v{+}1$.

\subsection{Unhappy Path} \label{sub:unhappy_path}
The unhappy path handles recovery when a view fails to progress—typically due to a faulty leader or prolonged network delays. Once $2f+1$ validators declare the view failed, \sysname\ transitions to the unhappy path\footnote{We use the terms \emph{happy} and \emph{unhappy path} only for explanatory purposes; the actual protocol does not have such an unambiguous global state, as there may be divergent local perspectives by different validators---for instance, when one validator receives a proposal, but another does not and times out.}, enabling an honest leader to recover while ensuring tail-forking resistance.

\medskip
\mypara{Timeout messages.}
Let $v$ be the current view. If a validator does not enter the next view $v+1$ within a predetermined timeout period after entering view $v$, it broadcasts a timeout message, signaling that view~$v$ has failed from its perspective.
Recall that a timeout message includes: (1) the current view, (2) the local tip, (3) a vote for the local tip, (4) the most recent observed QC (i.e., the local high QC), (5) the QC or TC that caused the validator to enter view $v$, and (6) a signature. 

To create the timeout message, each validator $i$ maintains its local tip $T_i$ and the most recent QC $\iqc_i$ it has seen, and checks whether $\iqc_i.\view \geq T_i.\view$. If so, the validator sets $(\qc, \tip,\hightipvote) \gets (\iqc_i, \bot, \bot)$. Otherwise, the validator sets the $(\qc, \tip) \gets (\bot, T_i)$. In the latter case, the validator also populates the $\hightipvote$ field with a vote for $T_i$.



\medskip
\mypara{Fast recovery: Building a QC.}
Recall that in the happy path, the leader $L_{v+1}$ of the view $v+1$ directly collects votes and builds a QC for the proposal of view $v$. Even when this fails, in two cases a QC for the same proposal can be constructed, allowing the next honest leader to quickly recover from the failure and make a regular fresh proposal:
\begin{itemize}
    \item \textbf{\underline{Case 1:}} The previous leader $L_v$ constructs a QC for its own proposal, broadcasts it, and it is received by the next honest leader.
    \item \textbf{\underline{Case 2:}} Enough timeout messages from view $v$ contain \emph{tip votes}, and they can be used to construct a QC for it.
\end{itemize}

\subsubsection{Processing Timeout Messages}
\label{subsub:process-timeout-msg}

\begin{figure}[t]
    \centering
    \includegraphics[width=0.8\textwidth]{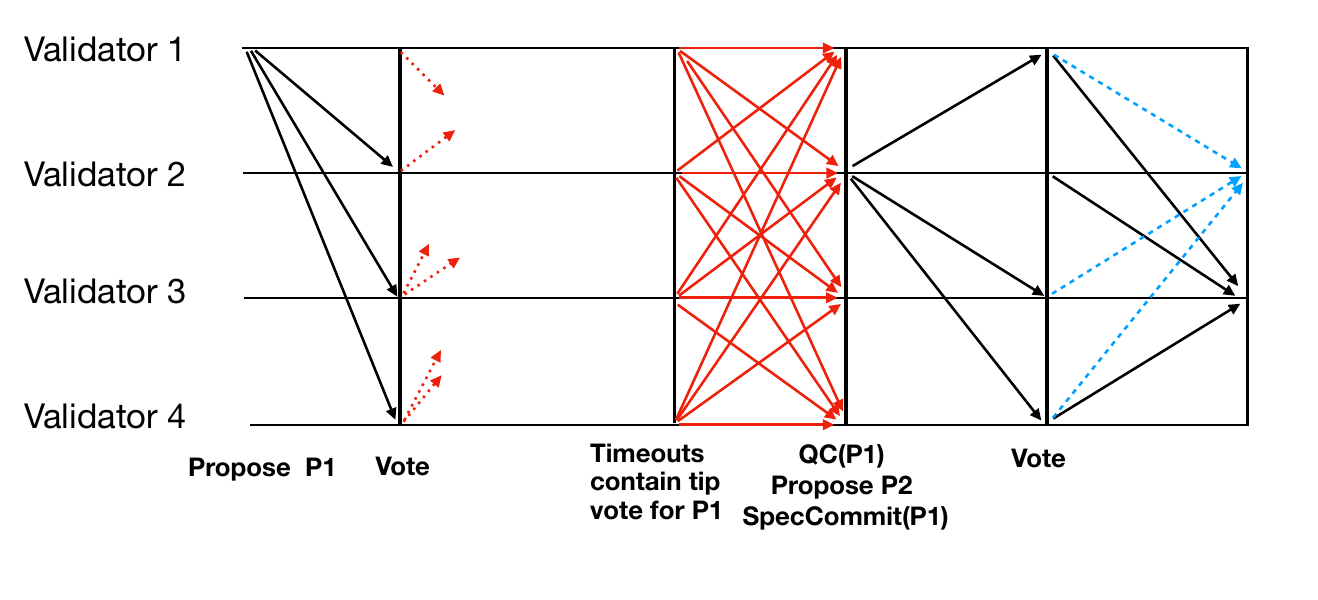}
    \caption{Unhappy path, but fast recovery workflow in \sysname, illustrating formation of a QC from tip votes. In this scenario, votes for proposal $P_1$, sent to both the current and next leader (dotted, \textcolor{red}{red}), are delayed. 
    Consequently, no QC can be build for the proposal and the view times out. Upon timing out, validators broadcast timeout messages (solid, \textcolor{red}{red}), which also include their votes for $P_1$. The next leader (validator~2) then aggregates these votes to form a QC for $P_1$ and extends it with a new proposal $P_2$. This example also illustrates that a validator can cast both a vote and a timeout message in the same view and that both a QC and a TC can be formed in the same view.}
    \label{fig:MonadBFT-Unhappy-Tip-votes}
\end{figure}

\begin{figure}[t]
    \centering
    \includegraphics[width=0.8\textwidth]{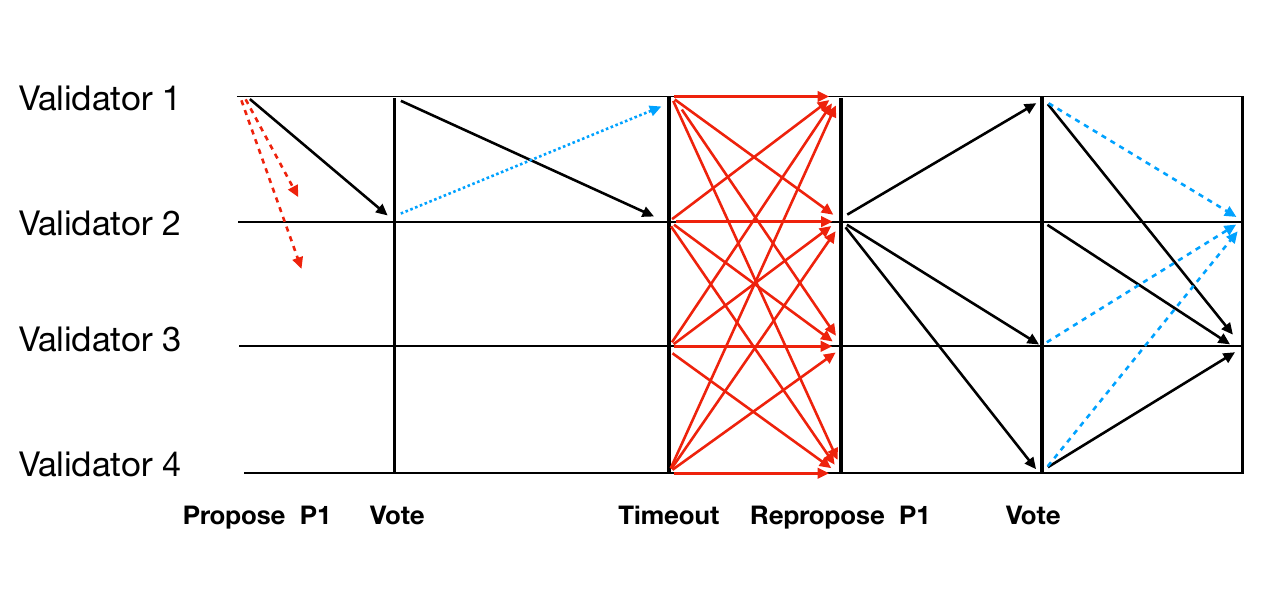}
    \caption{Unhappy path workflow in MonadBFT using standard recovery. In this example, the proposal P1 is not received by sufficiently many validators to form a QC, neither via regular (solid, black), nor backup (dashed, blue), nor tip votes (solid, \textcolor{red}{red}). P1 appears as the high tip of the timeout certificate constructed. It is then reproposed by the next leader.}
    \label{fig:MonadBFT-Unhappy-Workflow}
\end{figure}



%
If constructing or obtaining a QC from view $v$ is not possible, the leader $L_{v+1}$ waits until $2f+1$ timeout messages are received and then builds a TC.
Let $\tipsviews$ and $\qcsviews$ be the set of views of local tips and local high QCs reported in the timeout messages, respectively. $L_{v+1}$ creates a TC $\itc$ as: 
\[
    \itc.\view \gets v; \quad \itc.\tipsviews \gets \tipsviews; \quad \itc.\qcsviews \gets \qcsviews. 
\]
While creating $\itc$, $L_{v+1}$ uses the signatures $\sigma$ included in the timeout messages to compute $\itc.\Sigma$, and justify $\itc.\tipsviews$ and $\itc.\qcsviews$.

Let $T$ and $\iqc$ be the highest tip and local high QC (ordered by view number) reported in these timeout messages.
If multiple such tips have the same maximum view, the tip that contains the highest QC must be chosen (as validated by $\ValidTC$ in \Cref{Algorithm:ValidationPredicates}). 
If $\iqc.\view \geq T.\view$, $L_{v+1}$ sets $(\tc.\hightip, \itc.\highQC) \gets (\bot, \iqc)$, otherwise, $L_{v+1}$ sets $(\tc.\hightip, \itc.\highQC) \gets (T, \bot)$.

\medskip
Now, depending upon the TC, the following cases can occur.

\medskip
\mypara{\underline{Case 3:}}
If $\iqc.\view \geq T.\view$, $L_{v+1}$ proposes a fresh proposal $p$ with a new block $B$ with $B.\qc \gets \iqc$ and:
\[
    p.\view \gets v+1; \quad  p.\block \gets B; \quad p.\tc \gets \itc; \quad p.\nec \gets \bot.
\]
This is a fast recovery path since the TC itself directly implies that no reproposal is needed. We now turn to the two standard recovery cases.

\medskip
\mypara{\underline{Case 4:}}
If $T.\view > \iqc.\view$, $L_{v+1}$ then checks if it possesses the block $B$ corresponding to the tip $T$. If it does, $L_{v+1}$ reproposes the block $B$ in a proposal $p$ where:
\[
    p.\view \gets v+1; \quad 
    p.\block\gets B; \quad  
    p.\tc \gets \itc; \quad 
    p.\nec \gets \bot.
\]

\medskip
\mypara{\underline{Case 5:}}
Otherwise, i.e., if $L_{v+1}$ does not have the block $B$ corresponding to the tip $T$, the leader invokes the block recovery part of the consensus protocol. As we discuss in~\Cref{sub:recovery}, this recovery either returns the block $B$ or a no-endorsement certificate~(NEC), certifying that the block $B$ is not recoverable. 

If the block recovery returns $B$, $L_{v+1}$ reproposes $B$ exactly as in case 4. Otherwise, let $\inec$ be the NEC returned by the block recovery.
$L_{v+1}$ then issues a fresh proposal $p$ with a new block $B'$ with $B'.\qc \gets T.\qc$, and: 
\[
    p.\view \gets v+1; \quad 
    p.\block\gets B'; \quad  
    p.\tc \gets \itc; \quad 
    p.\nec \gets \inec.
\]


\subsubsection{Processing Proposals in the Unhappy Path}
\label{subsub:process-proposals}
This section will describe how the validators process the proposals they receive in the unhappy path.

A validator processes QCs and TCs to advance views regardless of how they are obtained -- whether constructed locally or received directly or indirectly as part of a proposal, timeout message, or tip. Moreover, upon receiving any valid message, each validator executes the pacemaker logic first (\Cref{subsection:Pacemaker}), before using the message for other protocol actions.

\medskip
\mypara{\underline{Case 1 \& 2:}}
During view $v+1$, if a validator receives a proposal $p$ from $L_{v+1}$ with $p.\view = v+1$ and $p.\block.\qc.\view = v$, the validator processes $p$ exactly as in the happy path~(see~\Cref{sub:happy_path}).

\medskip
\mypara{\underline{Case 3:}}
Alternatively, upon receiving a proposal $p$ with:
\[
    p.\tc \ne \bot;  \quad 
    p.\tc.\hightip = \bot; \quad 
    p.\tc.\highQC \ne \bot; \quad 
    p.\nec = \bot,
\]
each validator checks that the proposal $p$ is consistent with $p.\tc$ using the \SafetyCheck() function in Algorithm~\ref{Algorithm:ValidationPredicates}. Intuitively, these checks ensure that $L_{v+1}$ followed the procedure in case 3 of~\Cref{subsub:process-timeout-msg} to create~$p$.
Additionally, the validator checks that:
(i) $p.\tc.\highQC = p.\block.\qc$, i.e.,
the proposed block extends the block corresponding to $p.\tc.\highQC$; and  
(ii) $p.\tc.\highQC.\view < p.\view - 1$. 
Check (ii) here prevents $L_{v+1}$ from triggering case 3, when $L_{v+1}$ could have created $p$ as per case 1 and 2 in~\Cref{subsub:process-timeout-msg}.

If all these checks are successful, as in the happy path, the validator: 
(i) sends its vote for $p$ to both $L_{v+1}$ and the leader of view $v+2$; and
(ii) updates its local tip to the tip of proposal $p$.

\medskip
\mypara{\underline{Case 4:}}
Alternatively, upon receiving a proposal $p$ with:
\[
    p.\tc \ne \bot;  \quad 
    p.\tc.\hightip \ne \bot; \quad 
    p.\tc.\highQC = \bot; \quad 
    p.\nec = \bot,
\]
as in the previous case, each validator checks that the proposal $p$ is consistent with $p.\tc$ to ensure that $L_{v+1}$ followed the procedure in case 4 of~\Cref{subsub:process-timeout-msg} to create $p$.
Moreover, if all these checks are successful, the validator casts votes.
This is also the only scenario in which the local tip is \emph{not} updated to the tip of proposal $p$.
Since $p$ is a reproposal, the local tip is instead set to the high tip carried in the included TC.
Note that $p$ and the high tip refer to the same block, as $p$ is a valid reproposal of that high tip.

\medskip
\mypara{\underline{Case 5:}}
There are two possibilities here, depending upon whether $p.\nec$ is set to $\bot$ or not. 
If $p.\nec = \bot$, the validator processes this proposal exactly as in case 4 above. 
Otherwise, the validator checks that:
(i) $p.\view = p.\nec.\view = v+1$, i.e., the NEC is from the current view,
(ii) $p.\block.\qc.\view = p.\nec.\hightipQCview$, i.e., the block extends the parent of the high tip that was requested by the leader, and, 
(iii) $p.\itc \neq \bot$ and valid.
If all these checks are valid, the validator casts votes and updates its local tip to the tip of proposal $p$.

\begin{figure}[h!]
    \centering
    \begin{minipage}[b]{0.4\textwidth}
        \centering
        \includegraphics[width=\textwidth]{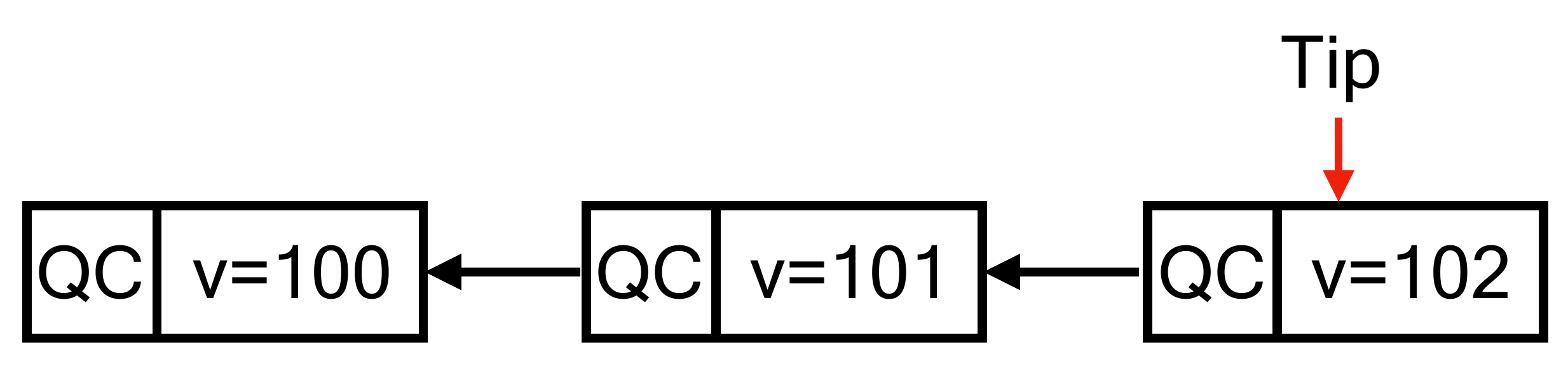}
        \caption{The tip points to the latest proposal.}
        \label{fig:unhappy-1}
    \end{minipage}
    \hfill
    \begin{minipage}[b]{0.48\textwidth}
        \centering
        \includegraphics[width=\textwidth]{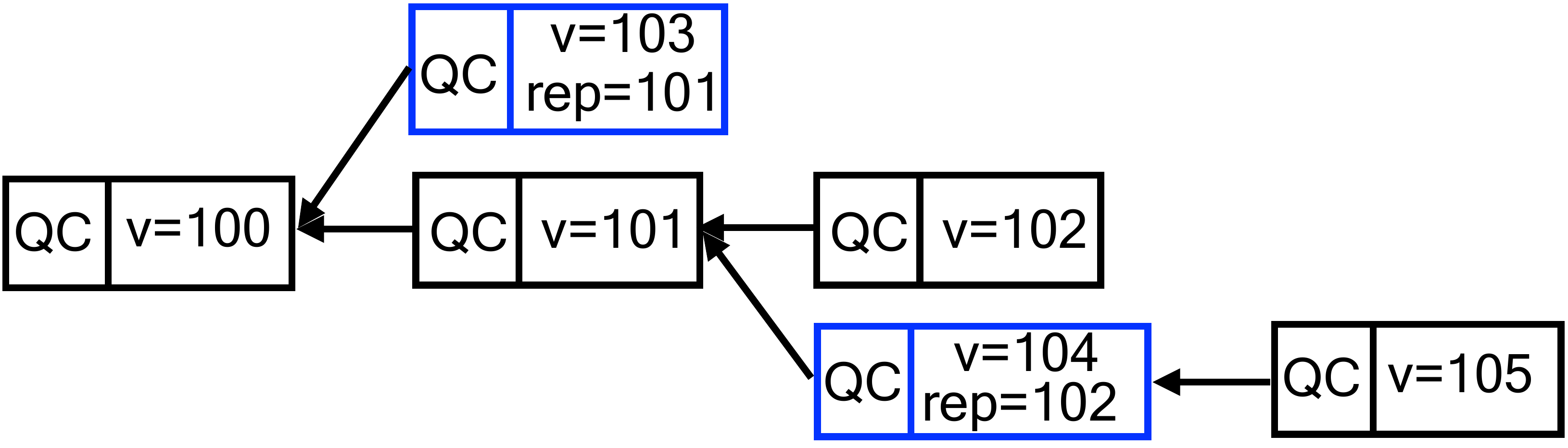}
        \caption{Fresh proposals and reproposals.}
        \label{fig:unhappy-2}
    \end{minipage}
    \hfill
    
    \vspace{20pt}
    
    \begin{minipage}[b]{0.49\textwidth}
        \centering
        \includegraphics[width=\textwidth]{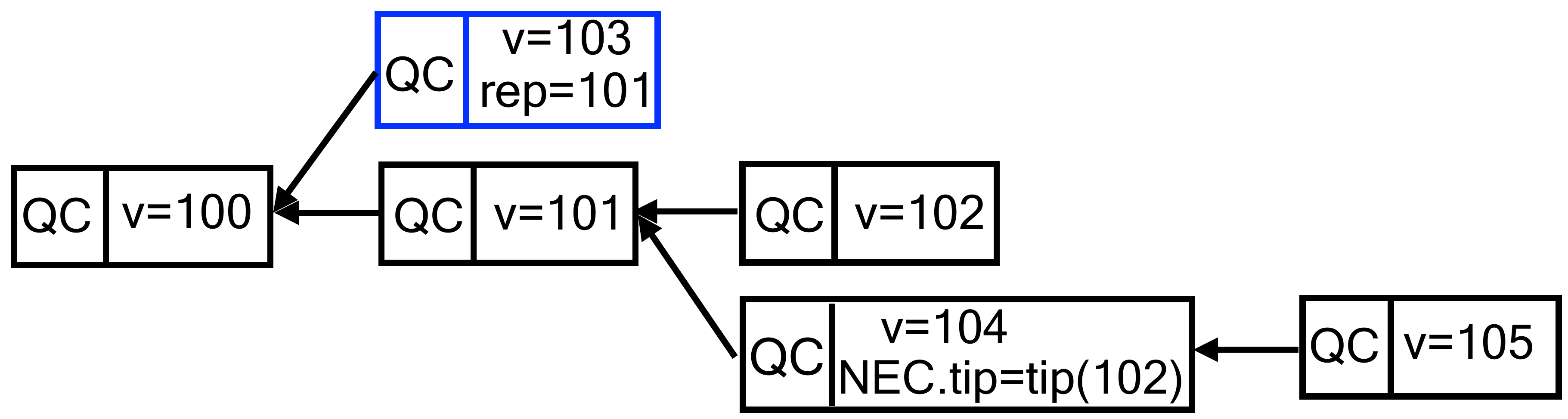}
        \caption{\NEC allowing a leader to make a fresh proposal} 
        \label{fig:unhappy-3}
    \end{minipage}
    \hfill
    \begin{minipage}[b]{0.48\textwidth}
        \centering 
         
        \textbf{The figures} illustrate the selection of the \hightip and the subsequent proposal processes during view failures. Blue rectangles indicate reproposals, while black rectangles indicate fresh proposals.
        
    \end{minipage}
\end{figure}

\medskip
\mypara{Remark on the recursive structure of tips.}
At first glance, the definition of tips appears recursive: a tip may contain a TC, which in turn may contain a tip, and so on. In practice, this recursion can be trivially avoided. 
Consider any tip $\itip$ with $\itip.\tc \neq \bot$. We distinguish two cases:
\begin{itemize}
    \item If $\itip.\nec \neq \bot$, then the TC can be safely discarded (by setting $\itip.\tc = \bot$).
    Since a valid NEC already provides all the information needed to validate the tip, the embedded TC becomes redundant and the recursive reference can be safely removed.
    
    \item If $\itip.\nec = \bot$, then $\itip.\tc$ necessarily contains a nonempty $\highQC$ (i.e., $\itip.\tc.\highQC \neq \bot$), which in turn implies $\itip.\tc.\hightip = \bot$. Thus, the TC does not contain another tip, so no recursion occurs.
\end{itemize}
Hence, although the data type permits a recursive structure, valid tips that arise during the protocol can always be represented in a non-recursive form.
For simplicity, we leave the recursive definition in the presentation, but in the actual \sysname implementation these recursions are eliminated.

\medskip
\mypara{Examples of tip handling \& reproposal dynamics.}
Figure~\ref{fig:unhappy-1} illustrates the local tip held by a validator.
Figure~\ref{fig:unhappy-2} depicts multiple view failures: first, views $100$ and $101$ succeed, i.e., QC is formed for the blocks proposed in these views. Next, view $102$ fails, and the tip from view $101$ is selected as the high tip and reproposed in view $103$.
The proposal from view $102$ did not propagate to a sufficient number of validators and, consequently, was not selected as the high tip in view $103$.
Following the failure of view $103$, the tip from view $102$ resurfaces as part of the subsequent TC and is selected as the high tip for view $104$.
Note that the tip from view $103$ is not considered a candidate for the high tip, as only fresh proposals are eligible.
The tip from view $102$ is then reproposed in view $104$, and view $104$ succeeds.

Figure~\ref{fig:unhappy-3} depicts an alternative scenario for view $104$.  
After the failure of view $103$, the leader does not hold the block corresponding to the high tip (i.e., the tip associated with the fresh proposal of view $102$) and therefore requests it using the block recovery protocol.  
The leader of view $104$ subsequently receives an NEC for the high tip from view $102$, enabling it to issue a fresh proposal in view $104$ instead of reproposing the block from view $102$.  
The NEC includes the view of the highest observed QC, view $101$, which the fresh proposal must extend; in other words, the parent of the high tip from view $102$ becomes the parent of the newly issued fresh proposal.

\begin{figure}[t]
    \centering
    \includegraphics[width=0.9\textwidth]{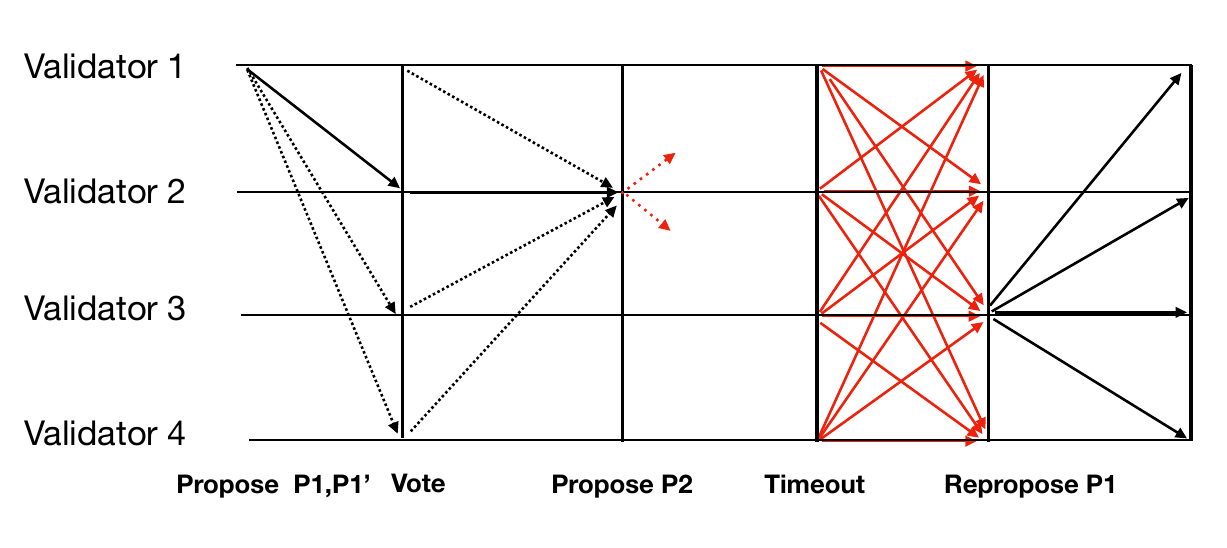} 
    \caption{Example to illustrate that tail-forking cannot be ruled out under leader equivocation. Validator 1, a malicious leader, makes two equivocating fresh proposals, sending P1 (solid, black) to validator 2 and a conflicting proposal P1' (dashed, black) to validators 3 and 4 (f+1 honest validators). Each validator votes for the proposal they received and sends their vote to the next leader, validator 2 (the figure omits backup votes). Validator 2 collects enough votes to form a QC for P1', but it fails to send proposal P2, causing validators to send timeout messages containing their highest observed tips. Validators 3 and 4 send timeout message including the tip of P1' as the highest tip. Validator 2 and the Byzantine validator 1 send timeout messages including the tip of P1 as the highest tip. Validator 3 builds a TC containing both P1 and P1' as high tips. Since both are proposals building on the same QC, validator 3 can freely choose to repropose one of them. In this example, validator 3 reproposes P1, leading P1' to be abandoned despite having obtained a QC. 
    }
    \label{fig:MonadBFT-Unhappy-tailforking}
\end{figure}

\subsubsection{Tail-forking Resistance and Equivocation}
Recall that tail-forking resistance is guaranteed only for blocks that receive at least $f+1$ votes from honest validators and whose issuing leader has not equivocated (see Section~\ref{sub:consensus_properties}).  
In this section, we analyze how equivocation interacts with tail-forking resistance, in particular how equivocation can cause a block that already carries a QC to be reverted. We illustrate this in~\Cref{fig:MonadBFT-Unhappy-tailforking}.

First, observe that equivocation by itself does not compromise the voting process: in particular, it cannot lead to two conflicting QCs for the same view.  
This is because each correct validator casts at most \emph{one} vote per view---that is, it votes for at most one proposal in any view.
Once a validator has voted in a given view, all subsequent proposals for that view are rejected (enforced by the $\mathsf{highest\_voted\_view}$ variable; see the check in Algorithm~\ref{Algorithm:Consensus-Execution_1}, line~\ref{Alg:Algorithm: Consensus-Execution:safety_check}).  

The critical point arises when considering how equivocation can interact with the protocol’s recovery mechanisms whose goal is to prevent tail-forking attacks.  
Suppose the leader of some view $v$ equivocates by issuing multiple proposals 
$p_{set} = \{p_1, p_2, \ldots, p_k\}$ ($k > 1$) for the same height in the blockchain.  
As argued above, at most one proposal $p^{\star} \in p_{set}$ can gather a QC.  
We now distinguish two cases:

\begin{itemize}
    \item \emph{Case 1: Sufficient QC dissemination.}  
    If, by the ``backup voting'' mechanism, at least $f+1$ honest validators receive the QC for $p^{\star}$ while still in view $v$, then all of them advance to view $v+1$ without ever issuing a timeout message for view $v$.
    Consequently, no TC for view $v$ can be formed.
    This forces every proposal in view $v+1$ to carry a QC from view $v$, which in turn ensures that it must extend $p^{\star}$, thereby preserving tail-forking resistance.
    

    \item \emph{Case 2: Insufficient QC dissemination.}  
    If fewer than $f+1$ honest validators receive the QC for $p^{\star}$, then a TC for view $v$ may be formed.
    In that case, although $p^{\star}$'s tip must still be present among the $2f+1$ timeout messages, the subsequent leader may also observe some other conflicting proposal from $p_{set}$ and choose to repropose it. 
    This effectively abandons $p^{\star}$, even though a QC for it exists, thereby breaking tail-forking resistance. 
    For this reason, \sysname guarantees tail-forking resistance only when the issuing leader has not equivocated.  
\end{itemize}

\subsection{Pacemaker}
\label{subsection:Pacemaker}
The pacemaker is a coordination protocol that validators use to advance views once the network becomes synchronous. Specifically, it guarantees that after GST, if a correct validator enters view~$v$, then all correct validators will enter view~$v$ or higher within a bounded time.
In \sysname, we implement the pacemaker using quorum certificates~(QCs), timeout certificates~(TCs), and Bracha-style amplification based on timeout messages~\cite{Bracha-Broadcast,Bracha_based_pacemaker_bounded_memory}. We summarize the pacemaker protocol in~\Cref{Algorithm:Pacemaker} and describe it below.

\smallskip
Consider validator $i$ whose current view is $v$. Then, validator $i$ performs the following actions:
\begin{enumerate}[itemsep=2pt]
    \item Upon receiving a QC~$\iqc$ with $\iqc.\view \geq v$, 
    (i) forwards $\iqc$ to the leader of view~$\iqc.\view{+}1$, (ii) updates its local view to $\iqc.\view{+}1$, and (iii) updates its local high QC to $\iqc$.

    \item Upon receiving a TC~$\itc$ with $\itc.\view \geq v$,   
    (i) broadcasts a timeout message for view~$\itc.\view$, unless it has already done so, and (ii) updates its local view to $\itc.\view{+}1$. 

    \item Upon receiving $f{+}1$ timeout messages for some view~$v' \ge v$, broadcasts its own timeout message for view~$v'$, unless it has already done so. Validator $i$ uses a QC and TC included in any of these $f+1$ timeout message to create its timeout message. This step implements the Bracha amplification mechanism~\cite{Bracha-Broadcast,Bracha_based_pacemaker_bounded_memory}. 

    \item Upon receiving $2f{+}1$ timeout messages for some view~$v' \ge v$, constructs a TC~$\itc$ and processes it as in step~2 above.
\end{enumerate}

\subsection{Block Recovery}
\label{sub:recovery}
Recall from~\Cref{subsub:process-timeout-msg}, during the unhappy path, a leader $L$ in view $v$ collects a TC $\itc$ and in certain scenarios, reproposes the block corresponding to the high-tip in $\itc$.
Let $p$ be the proposal corresponding to $\itc.\hightip$, and say $L$ seeks to repropose the block $p.\block$. Now, if $L$ does not possess the $p.\block$, it invokes $\Recover(\itc)$ to recover $p$ and consequently the block $p.\block$ or a NEC certifying that $p.\block$ need not be reproposed. We specify the recovery protocol in~\Cref{Algorithm:Recovery}, and describe it next.
\begin{itemize}[itemsep=3pt]
    \item The leader $L$ sends the $\msg{\ProposalRequest, \itc}$ message to batches of $\batch$ validators, prioritizing those who included tip $\itc.\hightip$ in their timeout messages. $L$ continues to periodically send $\msg{\ProposalRequest, \itc}$ to additional validators, until it recovers a proposal or an NEC. Upon receiving $\msg{\ProposalRequest,\itc}$, each validator validates $\itc$ and checks that $L$ is the designated leader for $\itc.\view$. Upon successful validation, the validator increments its view, and if it possesses the proposal $p$, sends $p$ back to $L$.

    \item Simultaneously, $L$ broadcasts the $\msg{\NERequest,\itc}$. Upon receiving $\msg{\NERequest,\itc}$, each validator validates the $\NERequest$ similarly to $\ProposalRequest$. Upon successful validation, the validator increments its view, and responds with a no endorsement~(NE) message $\msg{\NE,\sigma}$ unless it has previously cast a vote for $p$. Here, $\sigma$ is a digital signature on the tuple $\tupled{\itc.\view+1, \itc.\hightip.\blockheader.\qc.\view}$. 

    \item The leader, upon receiving $2f+1$ valid NE messages, aggregates the signatures to create an NEC.

    \item The block recovery protocol terminates when either $L$ receives the proposal $p$ or constructs an NEC.
\end{itemize}

\begin{remark}
Note that it is possible to simultaneously recover proposal $p$ and construct an NEC if fewer than $f+1$ honest validators hold the proposal. The presence of an NEC does not imply that the block is unrecoverable; rather, it indicates that fewer than~$f+1$ honest validators obtained it during its propagation. Consequently, the block can be safely abandoned without compromising the tail-forking resistance property.
\end{remark}

\begin{remark}
Note that, in the worst case, the recovery protocol we describe incurs $O(n/\batch)$ latency. This can happen when only $O(1)$ correct validators have the block, and $f$ malicious validators who claim to have the proposal in their timeout messages do not assist the leader to recover $p$. 
Although it is possible to improve the protocol using erasure codes, we opt for this simpler protocol as we envision that such situations will occur rarely in practice.
\end{remark}



\section{Analysis}
\label{sec:proofintuition}

This subsection provides an overview and intuition of how and why \sysname satisfies its properties. 
Our goal is to develop intuition for the core mechanisms of \sysname and how they guarantee these properties, and we therefore include several detailed explanations.
The full proofs can be found in \Cref{sec:analysis}.

\medskip
\mypara{Safety.}
Safety requires that no two correct validators commit conflicting blocks, that is, different blocks at the same height in their respective local logs.
Informally, the commitment rule in \sysname states that a block is committed once a block has obtained a ``QC of a QC'' and both QCs are from consecutive views.
More precisely, suppose a block $B$ appears in a proposal $p$ from view $v$.
A correct validator commits $B$ upon observing (i) a QC for $p$ from view $v$, and (ii) a QC for a proposal $p'$ from view $v + 1$ such that $p'$ directly extends $p$ (by carrying a QC for $p$).
We now explain why, after this point, no conflicting block $B'$ can ever be committed.

First, consider QCs formed in view $v$.
Since $2f + 1$ validators (including at least $f + 1$ correct validators) vote for proposal $p$ in view $v$, and correct validators never vote for two different proposals within the same view, it is impossible to collect $2f + 1$ votes in view $v$ for any other proposal.
Hence, the only QC that can be formed in view $v$ is the QC for $p$.
By the same quorum-intersection reasoning, the only proposal that can obtain $2f + 1$ votes in view $v + 1$ is the proposal $p'$.
Thus, the only QC that can be formed in view $v + 1$ is the QC for $p'$, which strictly extends proposal $p$.
Therefore, in the two consecutive views responsible for committing $B$, the protocol permits only the QC for $p$ in view $v$ and only the QC for $p'$ in view $v + 1$.

Next, consider any valid proposal $p^*$ formed in view $v + 2$.
We show that such a proposal must necessarily strictly extend $p$, which already suffices to protect the commitment of $B$.
There are two possibilities for $p^*$:

\begin{itemize}
    \item \emph{Case 1:} Proposal $p^*$ is a fresh proposal carrying a QC from view $v + 1$.
    In view $v + 1$, the only QC that can form is the QC for $p'$, and $p'$ strictly extends $p$.
    Thus, proposal $p^*$ carrying a QC from view $v + 1$ must strictly extend $p$.

    \item \emph{Case 2:} Proposal $p^*$ is either (i) a fresh proposal carrying an NEC from view $v + 2$ or carrying a TC from view $v + 1$ whose $\highQC$ field is non-$\bot$, or (ii) a reproposal carrying a TC from view $v + 1$.
    All of these possibilities imply the existence of a TC formed in view $v + 1$.
    

    We first observe that any TC $\itc$ formed in view $v + 1$ must satisfy two properties:  
    (1) its $\hightip$ field must be non-$\bot$ (and consequently its $\highQC$ field must be $\bot$), and  
    (2) this high tip must come from view $v + 1$ and must carry a QC from view $v$.
    Indeed, at least $f + 1$ correct validators (namely, those that voted for proposal $p'$) update their local tip to the tip of $p'$ in view $v + 1$; recall that $p'$ is a fresh proposal because it carries a QC from view $v$.
    Since TCs aggregate timeout messages from $2f + 1$ validators, at least one of the updated tips from these correct validators must appear in any valid TC, and therefore in $\itc$.
    Moreover, because the highest-view tip is selected as the high tip (breaking ties by choosing the tip that carries the most recent QC), the high tip of $\itc$ must be the tip from view $v + 1$ that carries a QC from view $v$.
    Finally, because no conflicting QC can be formed in view $v$, the QC carried by this tip must be the QC for $p$, which implies that the high tip of $\itc$ strictly extends proposal $p$. 

    With this structure in place, any valid proposal $p^*$ derived from a TC formed in view $v + 1$ must also strictly extend $p$:
    \begin{itemize}
        \item If $p^*$ is a fresh proposal carrying an NEC from view $v + 2$, observe that forming an NEC requires at least $f + 1$ correct validators to issue $\NE$ messages upon receiving a TC from view $v + 1$.
        Each $\NE$ message identifies the view whose QC the leader must incorporate into the new proposal.
        Since every TC formed in view $v + 1$ has a high tip carrying a QC from view $v$, the NEC forces $p^*$ to include that QC, which in turn points to $p$.
        Thus, $p^*$ strictly extends $p$.

        \item Note that $p^*$ cannot carry a TC from view $v + 1$ whose $\highQC$ field is non-$\bot$, since we have established that every TC formed in view $v + 1$ must use its high tip (and therefore has its $\highQC$ field set to $\bot$).

        \item If $p^*$ is a reproposal carrying a TC from view $v + 1$, it reproposes the high tip of that TC.
        As shown above, that high tip strictly extends $p$, so $p^*$ strictly extends $p$ as well.
    \end{itemize}
\end{itemize}
Thus, every valid proposal in view $v + 2$ strictly extends $p$, thereby preserving the commitment of $p$.
The same reasoning applies inductively to all views beyond $v + 2$.
Therefore, once $B$ is committed, no future proposal can revert or conflict with $p$, and no conflicting block $B'$ can ever be committed.

\medskip
\mypara{Tail-forking resistance.}
We now give an intuitive explanation of why \sysname satisfies tail-forking resistance.
Assume a leader proposes a block $B$ as part of a fresh proposal $p$ in some view $v$, that at least $f + 1$ correct validators vote for $p$, and that the leader does not equivocate in view $v$.
Tail-forking resistance states that proposal $p$ (equivalently, block $B$ carried by $p$) can never be abandoned.
As before, we illustrate the argument for view $v + 1$; the reasoning then applies inductively to all later views.

Consider any valid proposal $p^*$ issued in view $v + 1$.
We now show that $p^*$ must extend $p$, meaning that $p^*$ is either a reproposal of $p$ or builds directly on top of $p$.
There are several possibilities for the form of $p^*$:
\begin{itemize}
    \item \emph{Case 1:} Proposal $p^*$ is a fresh proposal carrying a QC from view $v$.
    Since $f + 1$ correct validators voted for $p$ in view $v$ and never vote for two different proposals in the same view, any QC formed in view $v$ must be the QC for $p$.
    Therefore, the parent of proposal $p^*$ must be $p$, which shows that $p$ is not abandoned by $p^*$.
    
    \item \emph{Case 2:} Proposal $p^*$ is either (i) a fresh proposal carrying an NEC from view $v + 1$ or carrying a TC from view $v$ with a non-$\bot$ $\highQC$ field, or (ii) a reproposal carrying a TC from view $v$.
    All of these possibilities rely on the existence of a TC formed in view $v$.

    We first show that any TC $\itc$ formed in view $v$ must satisfy two conditions:
    (1) its $\hightip$ field must be non-$\bot$ (implying that its $\highQC$ field is $\bot$), and  
    (2) this high tip must be the tip of proposal $p$.
    Indeed, at least $f + 1$ correct validators (precisely those who voted for $p$) update their local tip to the tip of $p$ in view $v$.
    Since a TC aggregates timeout messages from $2f + 1$ validators, at least one of these updated tips must appear in any valid TC and therefore in $\itc$.
    Furthermore, because the highest-view tip is chosen as the high tip and the leader does not equivocate in view $v$, the high tip of $\itc$ must be the tip of $p$.

    This implies that any valid proposal $p^*$ derived from a TC formed in view $v$ must extend $p$:
    \begin{itemize}
        \item If $p^*$ is a fresh proposal carrying an NEC from view $v + 1$, observe that forming an NEC requires at least $f + 1$ correct validators to issue $\NE$ messages upon receiving a TC from view $v$.
        Since this TC has the tip of $p$ as its high tip and at least $f + 1$ correct validators supported $p$ in view $v$, none of these validators issue an $\NE$ message.
        Consequently, such an NEC cannot be constructed, making this case impossible.

        \item Proposal $p^*$ cannot carry a TC from view $v$ whose $\highQC$ field is non-$\bot$, since we have already shown that every TC from view $v$ must use its high tip and therefore has its $\highQC$ field set to $\bot$.

        \item If $p^*$ is a reproposal carrying a TC from view $v$, it must repropose the high tip of that TC.
        As argued above, that high tip is exactly the tip of proposal $p$, so $p^*$ reproposes $p$ (its block $B$).
    \end{itemize}
\end{itemize}

\medskip
\mypara{Liveness.}
We now explain how \sysname guarantees that every block proposed by a correct leader after GST is eventually committed.
Consider a proposal $p$ (fresh or a reproposal) issued by a correct leader in some post-GST view $v$.
Since the network is stable after GST, the leader disseminates $p$ to all validators sufficiently quickly so that all correct validators issue their votes for $p$.
Because votes are sent both to the current leader and to the next leader, the leader of view $v$ receives at least $2f + 1$ votes for $p$ and constructs a QC for it.
This QC is then broadcast to all validators, and every correct validator receives a QC before timing out from view $v$.
Consequently, no TC can be formed in view $v$.
This fact is essential: it implies that any valid proposal in any later view must strictly extend $p$, meaning that $p$ remains on the chain of all future valid proposals.

Eventually, there must be three consecutive correct leaders in views strictly greater than $v$.
Let these views be $v^*$, $v^* + 1$, and $v^* + 2$, all greater than $v$.
Since any valid proposal after view $v$ must extend $p$, the proposal issued in view $v^*$ necessarily extends $p$.
Because we are post-GST, the leader of view $v^*$ successfully disseminates its proposal to all validators, who all vote for it.
Thus, the leader of view $v^* + 1$ collects these votes and forms a QC for the proposal from view $v^*$.
The leader of view $v^* + 1$ then issues a proposal extending the view-$v^*$ proposal via this QC; clearly, this proposal from view $v^* + 1$ also strictly extends $p$.
All correct validators vote for this view-($v^* + 1$) proposal as well, ensuring that the leader of view $v^* + 2$ forms a QC for the proposal from view $v^* + 1$ and issues a new proposal carrying that QC.
When a correct validator receives the proposal from view $v^* + 2$, it observes two consecutive QCs for the view-$v^*$ proposal: one from view $v^*$ and one from view $v^* + 1$.
The commitment rule causes the validator to commit the view-$v^*$ proposal.
Since the view-$v^*$ proposal strictly extends $p$, and since all ancestors of a committed proposal are also committed, proposal $p$ (and thus its block) is committed at this moment.

\medskip
\mypara{Optimistic responsiveness.}
The (non-leader-based) optimistic responsiveness property states that, when all validators are correct, \sysname commits $\Omega(d / \delta)$ blocks during any post-GST time interval of duration $d$, where $\delta$ denotes the actual network delay (as opposed to the pessimistic bound $\Delta$).
This follows from the fact that, after GST, every view successfully extends the previous one by forming a QC.
Since each view is led by a correct leader, the leader reliably disseminates its proposal to all validators, all correct validators vote for it, and the next correct leader collects these votes to form a QC and issues a proposal that carries it.
As a consequence, a block proposed in view~$v$ is committed by the end of view~$v+2$: 
view~$v$ forms a QC for the proposal, view~$v+1$ forms the ``child'' QC extending it, and view~$v+2$ disseminates these two consecutive QCs, completing the commit. 
Since each view takes $O(\delta)$ time when all validators are correct, a block proposed in view~$v$ is committed within $O(\delta)$ time.
Over a duration of $d$, the protocol executes $\Omega(d / \delta)$ views, and all but the last two of these views produce committed blocks.
(The final two views also lead to commitments, but their corresponding commits occur just after the end of the time interval.)
Lastly, since each correct leader can issue its proposal as soon as it collects either $2f + 1$ votes (sufficient to form a QC) or $2f + 1$ timeout messages (sufficient to form a TC), \sysname satisfies the leader-based optimistic responsiveness property.

\medskip
\mypara{Leader fault isolation.}
This property states that, after GST, if a view $v + 1$ led by a Byzantine leader is ``sandwiched'' between two views $v$ and $v + 2$ led by correct leaders, then the duration of view $v + 1$ is at most $\viewduration + O(\delta)$, where $\viewduration$ is a single timeout duration.
In other words, a single Byzantine leader cannot delay progress for more than one $\viewduration$ period.

The intuition is as follows.
Since view $v$ is led by a correct leader and we are after GST, the leader of view $v$ successfully constructs a QC for its proposal and disseminates it to all validators.
Upon receiving this QC, all correct validators move to view $v + 1$, and they do so within $O(\delta)$ time of one another.
Let $\tau(v + 1)$ denote the time when the first correct validator enters view $v + 1$.
From this moment, it takes at most $\viewduration + O(\delta)$ time for all correct validators to issue their timeout messages for view $v + 1$.
Indeed, every correct validator issues a timeout message after spending $\viewduration$ time in view $v + 1$ (unless it has already progressed to later views), and all correct validators enter view $v + 1$ within $O(\delta)$ time of $\tau(v + 1)$.
Thus, by time $\tau(v + 1) + \viewduration + O(\delta)$, every correct validator has produced its timeout message for view $v + 1$.
Consequently, after another $\delta$ time, all correct validators can construct a TC for view $v + 1$ and transition to view $v + 2$, ensuring that the Byzantine leader’s influence is confined to a single view-duration period.

\medskip
\mypara{Linearity and low latency on the happy path.}
On the happy path (after GST, with leaders behaving correctly), each view incurs only linear communication.
This is because each view follows the ``leader-to-all, all-to-next-leader'' pattern, which incurs linear number of messages.
Importantly, no timeout messages are issued on the happy path, and therefore no TCs are formed.
Since quadratic communication arises only from TC construction, all quadratic costs are confined to the unhappy path (i.e., during failures or before GST).

Regarding latency, forming and disseminating a QC for a proposal $p$ requires exactly three communication steps on the happy path:
step~1 disseminates $p$,
step~2 disseminates votes for $p$,
and step~3 disseminates the next proposal carrying the QC for $p$.
Because \sysname permits speculative commitment upon observing a single QC, proposal $p$ can be speculatively committed within these three communication rounds.
(Recall that such speculative commitments can only be reverted in the presence of equivocation.)

For definitive, irrevocable commitment, \sysname requires one additional round-trip to form the second QC:
one more step of vote dissemination for the proposal that carries a QC for $p$, followed by the dissemination of the next proposal that carries this second QC.
Thus, irrevocable commitment is achieved with only two additional communication steps beyond speculative commitment, for a total of five communication steps.

\section{Related Work} 
\label{sec:related}
\sysname belongs to the HotStuff family of consensus protocols; therefore, we primarily discuss this class of protocols. We refer the reader to~\Cref{tab:comparison} for a detailed comparison with prior works.

\medskip
\mypara{BeeGees.}\label{para:beegees} Among existing protocols, BeeGees~\cite{BEEGEES} is the most closely related to MonadBFT. It introduces the any-honest-leader (AHL) property, which implies tail-forking resistance. Since both BeeGees and MonadBFT are HotStufff~\cite{HotStuff} derivatives, they share  linear complexity on the happy path, as well as pipelining as properties. However, the two protocols differ in several respects.
\begin{itemize}
    \item First, BeeGees's Slow View Change requires the recovery leader to receive 
full block data from prior failed views: at the $k$-th consecutive failure, 
each of the $n$ validators sends $\mathcal{O}(k \cdot b)$ data, resulting 
in $\mathcal{O}(k \cdot b \cdot n)$ inbound data at the recovery leader, 
where $b$ is the block size. 
MonadBFT, having been designed with practical deployment in mind, must address this and it does so at the protocol level: the recovery leader either recovers the pending block or obtains a no-endorsement certificate proving the block is unavailable, without requiring full block data to be embedded in view change messages.\footnote{We note that BeeGees could in principle be adapted to circumvent this issue; however, given its primarily theoretical focus, this does not appear to have been a design priority.}

\item Second, MonadBFT introduces fast recovery via backup QCs, which ensures that a single Byzantine leader causes only one timeout delay (leader fault isolation). 

\item Third, when sufficient votes for implicit QC materialization are unavailable, BeeGees' recovery leader must wait $O(\Delta)$ before proposing, whereas MonadBFT maintains leader-based optimistic responsiveness on the unhappy path once timeout messages have been collected.

\item Finally, MonadBFT introduces the notion of speculative commitment with precise reversion guarantees: a speculatively committed block can only be reverted if the proposing leader equivocated. While BeeGees' AHL property and reproposal mechanism may implicitly support a form of speculative finality, BeeGees does not define or discuss speculative commitment. MonadBFT makes this a first-class protocol property, with formal guarantees and concrete latency benefits for applications.
\end{itemize}

\medskip
\mypara{Concurrent work HotStuff-1\cite{kang2024hotstuff1linearconsensusonephase}.}
The concurrent work HotStuff-1 is a leader-based protocol with linear communication complexity in both the happy and unhappy paths.
HotStuff-1 achieves a weaker speculative finalization property than \sysname, where validators can speculative finalize a block within 1.5 round trips only if the parent block is already committed.
Moreover, HotStuff-1 does not prevent tail-forking attacks. To mitigate them, it assigns each validator a long view, subdivided into slots, where the leader proposes one block per slot.
This design introduces several drawbacks. First, the view duration must be sufficiently long to allow an honest leader to finalize its block, which increases the risk of prolonged censorship and causes higher delays if the leader fails.
In fact, a rational validator trying to maximize its MEV would censor MEV-related transactions until the last block, to maximize MEV~\cite{alpturer2025timing}.
Furthermore, HotStuff-1 lacks optimistic responsiveness (in the sense of \malor\ property we define in~\Cref{sub:consensus_properties}) when a leader fails during view switch. 
Specifically, after detecting the failure of the previous view, leader in HotStuff-1 is still required to wait $O(\Delta)$ before proposing.
Finally, tail-forking remains possible for the block proposed before the leader is rotated out, which is the root cause of the tail-forking problem.

\medskip
\mypara{Other protocols in the HotStuff family.}  
Other members of the HotStuff family such as HotStuff~\cite{HotStuff-2:2023/397}, Fast-HotStuff~\cite{jalalzai2021fasthotstuff}, Jolteon~\cite{gelashvili2021jolteon}, HotStuff-2~\cite{HotStuff-2:2023/397}, and PaLa~\cite{Pala},  have higher commit latencies and provide only limited protection against tail-forking. 

\medskip
\mypara{Protocols in the PBFT family.}
Within the family of Practical Byzantine Fault Tolerance (PBFT) protocols~\cite{castro01practicalPhD,doidge2024moonshotoptimizingchainbasedrotating,VBFT}, PBFT with tentative execution (PBFT-TE)~\cite{castro01practicalPhD} introduces a fast path that reduces latency to three message delays, while reverting to the traditional PBFT protocol if the fast path fails. 
PBFT-TE incurs $O(n^2)$ message complexity (even in the optimistic case), is not pipelined, and prevents tail-forking only if $f+1$ correct validators retain the QC for the block. 
The Moonshot protocol~\cite{doidge2024moonshotoptimizingchainbasedrotating} also achieves three-message-delay latency but  requires $O(n^2)$ communication on the happy path, and only provides weak tail-forking resistance.
Finally, VBFT~\cite{VBFT} provides tail-forking resistance, has speculative commit latency of two message delays, but incurs $O(n^2)$ message complexity in the optimistic case.

\section*{Acknowledgments}
We thank Justin Jacob and Robert Kiel for helpful feedback on this work. We are especially grateful to Karolos Antoniadis, Manuel Bravo, Yacine Dolivet, and Denis Kolegov for their extensive comments and suggestions.

\iffc
    \bibliographystyle{splncs04}
    \bibliography{References}
\fi
\appendix

\clearpage

\section{Proofs}
\label{sec:analysis}

\subsection{Safety}
\label{sub:safety}

Before proceeding with the proof of safety, we refine our ``extension'' terminology to also apply to tips.

\medskip
\mypara{Tips and proposals.}
We say that a proposal $p$ is a \emph{parent} of a tip $T$ of a fresh proposal $p'$ if and only if $p$ is a parent of $p'$. 
Since $p'$ might have multiple parents, tip $T$ can also have multiple parents.
We further say that a tip $T$ of a fresh proposal $p'$ \emph{extends} a proposal $p$ if and only if $p'$ extends $p$.
Finally, a tip $T$ of a fresh proposal $p'$ \emph{strictly extends} $p$ if and only if $p'$ strictly extends $p$.




\medskip
\mypara{Proof.}
Recall from~\Cref{sub:consensus_properties} that the safety property requires that no two correct validators commit distinct blocks at the same log position in their local logs.
To establish this, we first show that no correct validator votes for two distinct proposals within the same view.

\begin{lemma} 
\label{lemma:jovan_same_vote}
Let $\mathit{vote}_1$ and $\mathit{vote}_2$ be any two votes cast by a correct validator $i$ such that $\mathit{vote}_1.\view = \mathit{vote}_2.\view$. 
Then,
\[
\mathit{vote}_1.\blockid = \mathit{vote}_2.\blockid.
\]
\end{lemma}
\begin{proof}
For the sake of contradiction, assume that the statement of the lemma does not hold.
Without loss of generality, suppose validator $i$ issues vote $\mathit{vote}_1$ before issuing vote $\mathit{vote}_2$.
It is impossible for validator $i$ to send two different votes at line~\ref{line:send_vote_normal} of \Cref{Algorithm:Consensus-Execution_1}, as it updates the $\mathsf{highest\_voted\_view}$ variable immediately after casting the first vote (line~\ref{line:highest_voted_view_update_proposal} of \Cref{Algorithm:Consensus-Execution_1}).
Likewise, validator $i$ cannot send two different votes at line~\ref{line:broadcast_timeout_message} of \Cref{Algorithm:Consensus-Execution_1}, as it can time out from a given view $v = \mathit{vote}_1.\view = \mathit{vote}_2.\view$ at most once.
Furthermore, $\mathit{vote}_1$ cannot be sent at line~\ref{line:broadcast_timeout_message} of \Cref{Algorithm:Consensus-Execution_1} as, upon timing out from view $v$, validator $i$ updates its $\mathsf{highest\_voted\_view}$ variable to $v$ (line~\ref{line:update_highest_voted_view_timeout} of \Cref{Algorithm:Consensus-Execution_1}), thereby preventing it from sending a subsequent vote at line~\ref{line:send_vote_normal} of \Cref{Algorithm:Consensus-Execution_1}.
Hence, we conclude that $\mathit{vote}_1$ must be sent at line~\ref{line:send_vote_normal} of \Cref{Algorithm:Consensus-Execution_1}, and $\mathit{vote}_2$ must be sent at line~\ref{line:broadcast_timeout_message} of \Cref{Algorithm:Consensus-Execution_1}.

When validator $i$ creates $\mathit{vote}_1$, it invokes the $\CreateVote$ function with parameters $\curView$ and $\localhightip.\blockheader$ (line~\ref{line:create_vote_normal} of \Cref{Algorithm:Consensus-Execution_1}).  
When issuing $\mathit{vote}_2$, validator $i$’s $\localhightip$ and $\curView$ variables remain unchanged, since neither is modified after $\mathit{vote}_1$ is cast.  
(If either variable were modified, validator $i$ would have entered a view greater than $v$, which would prevent it from issuing $\mathit{vote}_2$ for view $v$.)  
Therefore, when $\mathit{vote}_2$ is created, validator $i$ invokes the $\CreateVote$ function with parameters $\curView$ and $\localhightip.\blockheader$ (line~\ref{line:create_vote_timeout} of \Cref{Algorithm:Utilities}).  
Since these parameters are identical to those used for $\mathit{vote}_1$, the statement of the lemma directly follows.
\qed
\end{proof}

Next, we show that two distinct proposals cannot both obtain QCs in the same view.

\begin{lemma} 
\label{lemma:jovan_qc_intersection}
Let $\mathit{qc}_1$ and $\mathit{qc}_2$ be any two QCs such that $\mathit{qc}_1.\view = \mathit{qc}_2.\view$.
Then, 
\[
\mathit{qc}_1.\blockid      = \mathit{qc}_2.\blockid.
\]
\end{lemma}
\begin{proof}
This follows from quorum intersection: since each QC has $2f + 1$ votes and $n =3f + 1$, at least one honest validator must have voted for both. 
As honest validators never vote for two distinct proposals in the same view (by \Cref{lemma:jovan_same_vote}), we have $\mathit{qc}_1.\blockid = \mathit{qc}_2.\blockid$.
\qed
\end{proof}

The lemma below shows that whenever a correct validator votes for a fresh proposal, its local tip is aligned with that proposal’s block.


\begin{lemma} \label{lemma:jovan_local_tip}
Let $\mathit{vote}$ be a vote cast by a correct validator $i$ in view $v$ for a fresh proposal $p$:
\[
\mathit{vote}.\blockid = p.\blockid,
\quad 
\mathit{vote}.\view = v,
\quad \text{and} \quad
p.\block.\originalview = v.
\]
Then, at the moment $\mathit{v ote}$ is cast, the following is satisfied at validator $i$:
(1) the block header stored in $\localhightip$ matches the block header of $p.\block$, and 
(2) $\localhightip.\view = v$.
\end{lemma}
\begin{proof}
Validator $i$ constructs $\mathit{vote}$ by invoking $\CreateVote(\curView, \localhightip.\blockheader)$ (see line~\ref{line:create_vote_normal} of \Cref{Algorithm:Consensus-Execution_1} and line~\ref{line:create_vote_timeout} of \Cref{Algorithm:Utilities}).  
Recall that
\[
p.\blockid = \Hash(p.\block.\blockhash, p.\view = v)
\]
and
\[
\mathit{vote}.\blockid = \Hash(\localhightip.\blockheader.\blockhash, \curView).
\]
Since $\mathit{vote}.\blockid = p.\blockid$, it follows that
\[
\localhightip.\blockheader.\blockhash = p.\block.\blockhash
\qquad\text{and}\qquad
\curView = p.\view = v.
\]
(Otherwise, $p.\blockid$ would differ from $\mathit{vote}.\blockid$.)
Note additionally that 
\[
p.\block.\blockhash = \localhightip.\blockheader.\blockhash = \Hash(p.\block.\originalview = v, \cdot, \cdot).
\]

Let $p'$ be the proposal whose reception most recently updated $\localhightip$ to its current value, which is now used in $\mathit{vote}$.
Since validator $i$ only processes valid proposals, $p'$ must be valid.
Moreover, because $\localhightip$ is used in a vote for view $v$, we must have $p'.\view \leq v$.
We now consider the two cases for $p'$:

\begin{itemize}
    \item \emph{Case 1:} $p'.\view < v$.
    
    Since $p'$ is valid, it follows from line~\ref{line:jovan_safety_check_proposal} of \Cref{Algorithm:ValidationPredicates} that $p'.\block.\originalview \leq p'.\view < v$.
    We further distinguish:
    \begin{itemize}
        \item If $\IsFreshProposal(p') = \text{true}$, then $\localhightip$ is updated to the tip of proposal $p'$ (see the $\GetTip(p')$ function), which itself contains the block header of $p'.\block$.
        However, as $p.\block.\blockhash = \Hash(v, \cdot, \cdot)$ and $p'.\block.\blockhash = \Hash(< v, \cdot, \cdot)$, we must have 
        \[
        \localhightip.\blockheader.\blockhash \neq p.\block.\blockhash,
        \] 
        contradicting the equality above.
        Therefore, this case is impossible.

        \item If $\IsFreshProposal(p') = \text{false}$, then, according to the check at line~\ref{alg:safetycheck:block_hash_check} of \Cref{Algorithm:ValidationPredicates}, we have that
\[
p'.\tc.\hightip.\blockheader.\blockhash = p'.\block.\blockhash = \Hash(< v, \cdot, \cdot).
\]
In this case, $\localhightip$ is updated to $p'.\tc.\hightip$ (see the $\GetTip(p')$ function).
However, since $p.\block.\blockhash = \Hash(v, \cdot, \cdot)$ while $p'.\tc.\hightip.\blockheader.\blockhash = \Hash(< v, \cdot, \cdot)$, it follows that 
\[
        \localhightip.\blockheader.\blockhash \neq p.\block.\blockhash.
        \]
Thus, this case is also impossible.
    \end{itemize}

    \item \emph{Case 2:} $p'.\view = v$.
    
    Since $p'$ is valid, line~\ref{line:jovan_safety_check_proposal} of \Cref{Algorithm:ValidationPredicates} ensures that $p'.\block.\originalview \leq v$.
    We distinguish:
    \begin{itemize}
        \item If $\IsFreshProposal(p') = \text{true}$, then $p'.\block.\originalview = v$ (because of the check at line~\ref{line:jovan_require_block_view_fresh} of \Cref{Algorithm:ValidationPredicates}, called at line~\ref{line:jovan_valid_tip_fresh_proposal}) and $\localhightip$ is updated to the tip of $p'$ (see the $\GetTip(p')$ function).
        Furthermore, since $\localhightip.\blockheader.\blockhash = p.\block.\blockhash = p'.\block.\blockhash$, it follows that $p.\block = p'.\block$ and thus $\localhightip.\blockheader$ indeed matches the block header of $p.\block$.
        Furthermore, $\localhightip.\view = v$ (as $\localhightip$ is updated to the tip of $p'$ with $p'.\view = v$), proving the lemma in this case.

        \item If $\IsFreshProposal(p') = \text{false}$, then line~\ref{line:jovan_check_block_view_reproposal} of \Cref{Algorithm:ValidationPredicates} implies $p'.\block.\originalview < v$.
        Moreover, due to the check at line~\ref{alg:safetycheck:block_hash_check} of \Cref{Algorithm:ValidationPredicates}, we have that
        \[
        p'.\tc.\hightip.\blockheader.\blockhash = p'.\block.\blockhash = \Hash(< v, \cdot, \cdot).
        \]
        In this case, $\localhightip$ is updated to $p'.\tc.\hightip$ (see the $\GetTip(p')$ function).
        However, since $p.\block.\blockhash = \Hash(v, \cdot, \cdot)$ while $p'.\tc.\hightip.\blockheader.\blockhash = \Hash(< v, \cdot, \cdot)$, it follows that 
        \[
        \localhightip.\blockheader.\blockhash \neq p.\block.\blockhash,
        \]
        which renders this case impossible.
    \end{itemize}
\end{itemize}
The only feasible scenario is the subcase where $p'.\view = v$ and $\IsFreshProposal(p') = \text{true}$, under which the desired statement holds.
This concludes the proof.
\qed
\end{proof}




    

The following lemma is critical for the safety property of the \sysname protocol.
In essence, this shows that once $f + 1$ correct validators vote for the ``second-level'' QC of a proposal, 
every valid proposal in subsequent views must strictly extend that proposal.

\begin{lemma}\label{lemma:jovan_safety_crucial}
Let $p$ be any proposal with $p.\view = v$.  
Suppose that at least $f + 1$ correct validators vote for a proposal $p'$ satisfying
\[
p'.\view = v + 1, 
\quad 
p'.\block.\qc.\view = v, 
\quad \text{and} \quad 
p'.\block.\qc.\blockid = p.\blockid.
\]
Then, any valid proposal $p^*$ issued in a view $v^* > v + 1$ (i.e., $p^*.\view = v^* > v + 1$) must strictly extend $p$.
\end{lemma}
\begin{proof}
Let $S$ denote the set of correct validators, with $|S| \geq f + 1$, that voted for fresh proposal $p'$ in view $v + 1$.  
For every correct validator $i \in S$, the following holds.  
At the moment $i$ votes for proposal $p'$ in view $v + 1$, \Cref{lemma:jovan_local_tip} implies that  
(1) $\localhightip$'s block header matches the block header of $p'.\block$, and  
(2) $\localhightip.\view = v + 1$.  
Hence, since validator $i$ cannot update its $\localhightip$ variable to a tip with view $v + 1$ in any view smaller than $v + 1$, it follows that $i$ must have updated its local tip in view $v + 1$ to a tip whose block header matches the block header of $p'.\block$.

We prove the lemma using induction on the view number $v'$.
Let us define the following five induction hypotheses, $X(v')$, $Y(v')$, $Q(v')$, $W(v')$, and $Z(v')$:
\begin{itemize}\itemsep2pt
    \item $X(v')$: Any (valid) proposal $p^*$ in view $v'>v+1$ strictly extends proposal $p$ and satisfies $p^*.\block.\qc.\view \geq v$. 
    Consequently, any (valid) tip $T$ of a fresh proposal from view $v' > v + 1$ also strictly extends proposal $p$ and satisfies $T.\blockheader.\qc.\view \geq v$.
     
    \item $Y(v')$: If any validator $i \in S$ updates its local tip to some tip $T$ in view $v' \geq v + 1$, then $T.\view \geq v + 1$ and $T.\blockheader.\qc.\view \geq v$.

    \item $Q(v')$: For any QC $\mathit{qc}$ with $\mathit{qc}.\view = v'\geq v+1$, any proposal $p_{\mathit{qc}}$ with $p_{\mathit{qc}}.\blockid = \mathit{qc}.\blockid$ strictly extends proposal $p$.

    \item $W(v')$: Let $\mathit{tc}$ be any TC with $\mathit{tc}.\view = v'\geq v+1$ such that $\mathit{tc}.\highQC \neq \bot$.
    Then, any proposal $p_{\mathit{tc}}$ with $p_{\mathit{tc}}.\blockid = \mathit{tc}.\highQC.\blockid$ strictly extends $p$ and $\mathit{tc}.\highQC.\view \geq v + 1$.
    
    \item $Z(v')$: Let $\mathit{tc}$ be any TC with $\mathit{tc}.\view = v'\geq v+1$ and $\mathit{tc}.\hightip \neq \bot$.
    Then, the proposal whose tip is $\mathit{tc}.\hightip$ strictly extends $p$, $\mathit{tc}.\hightip.\view \geq v + 1$ and $\mathit{tc}.\hightip.\blockheader.\qc.\view \geq v$.
\end{itemize}
Note that proving $X(v')$ is equivalent to proving the statement of the lemma. 
Therefore, we start by proving the rest. 

\medskip
\noindent \underline{$Y(v' = v + 1)$:}
Consider any correct validator $i \in S$.  
By the first paragraph of the proof, validator $i$ updates its local tip in view $v + 1$ to a tip whose block header matches that of $p'.\block$ and whose view is $v + 1$.
Since $i$ can update its local tip at most once per view and $p.\view = v$, it follows that $Y(v' = v + 1)$ holds.




\medskip
\noindent \underline{$Q(v' = v + 1)$:}
Let $\mathit{qc}$ be any QC with $\mathit{qc}.\view = v + 1$.
Since $f + 1$ correct validators voted for proposal $p'$ in view $v + 1$ and at least one of these validators must contribute to $\mathit{qc}$, \Cref{lemma:jovan_same_vote} implies that $\mathit{qc}.\blockid = p'.\blockid$.
Hence, the block contained in any proposal $p_{\mathit{qc}}$ with $p_{\mathit{qc}}.\blockid = \mathit{qc}.\blockid$ must be identical to the block contained in $p'$.
Therefore, $p_{\mathit{qc}}$ strictly extends $p$.

\medskip
\noindent \underline{$W(v' = v + 1)$:}
Let $\mathit{tc}$ be any TC with $\mathit{tc}.\view = v + 1$ and $\mathit{tc}.\highQC \neq \bot$.
By quorum intersection (since $|S| \geq f + 1$ and $\mathit{tc}$ contains information from $2f + 1$ validators), $\mathit{tc}$ includes information from at least one validator $i \in S$.
Since $Y(v + 1)$ holds and $i$ updated its local tip to a tip whose view is $v + 1$ and whose block header matches the block header of $p'.\block$ (from the first paragraph of the proof) before contributing to $\mathit{tc}$, $\mathit{tc}.\highQC.\view \geq v + 1$.
Moreover, since $\mathit{tc}.\view = v + 1$, we have $\mathit{tc}.\highQC.\view \leq v$ (due to the $\ValidTC$ function).
Therefore, this case is impossible; that is, no TC from view $v + 1$ can have a non-empty $\highQC$ field.


\medskip
\noindent \underline{$Z(v' = v + 1)$:}
Let $\mathit{tc}$ be any TC with $\mathit{tc}.\view = v + 1$ and $\mathit{tc}.\hightip \neq \bot$.
By quorum intersection (since $|S| \geq f + 1$ and $\mathit{tc}$ contains information about $2f + 1$ validators), $\mathit{tc}$ includes information from at least one validator $i \in S$.
Because $Y(v + 1)$ holds and $i$ updated its local tip to a tip whose view is $v + 1$ and whose block header matches the block header of $p'.\block$ (from the first paragraph of the proof) before contributing to $\mathit{tc}$, $\mathit{tc}.\hightip.\view \geq v + 1$.
Given that $\mathit{tc}.\view = v + 1$, we have that $\mathit{tc}.\hightip.\view \leq v + 1$ (by the $\ValidTC$ function).
Therefore, $\mathit{tc}.\hightip.\view = v + 1$.
Moreover, since the local tip of validator $i \in S$ contains a block header that matches the block header of $p'.\block$, it follows that $\mathit{tc}.\tipsviews[i] = v + 1$ and $\mathit{tc}.\qcsviews[i] = v$.
(Note that validator $i$ reports its local tip and not its $\localhighQC$ variable, since its local tip is from view $v + 1$, whereas its $\localhighQC$ variable cannot hold any QC from any view greater than $v$.)
This further implies that $\mathit{tc}.\hightip.\blockheader.\qc.\view = v$.
As a QC was formed for proposal $p$ in view $v$, \Cref{lemma:jovan_qc_intersection} ensures that a parent of the proposal whose tip is $\mathit{tc}.\hightip$ must be $p$.
Therefore, the proposal whose tip is $\mathit{tc}.\hightip$ strictly extends $p$, which implies that $Z(v' = v + 1)$ holds.

\medskip
\noindent \underline{$X(v' = v + 2)$:}
Let us consider all possible cases for a valid proposal $p^*$ with $p^*.\view = v + 2$:
\begin{itemize}\itemsep2pt
    \item Suppose $p^*$ is a fresh proposal. Consider the following three possibilities: 
    \begin{itemize}\itemsep2pt
        \item Let $p^*.\block.\qc.\view = v + 1$. Then, by quorum intersection, at least one validator $i\in S$ must contribute to $p^*.\block.\qc$. 
        Hence, \Cref{lemma:jovan_same_vote} implies that a parent of proposal $p^*$ must be proposal $p'$.
        Since $p'$ strictly extends proposal $p$, $p^*$ strictly extends proposal $p$, thereby proving $X(v' = v + 2)$ in this case.

        \item Let $p^*.\nec \neq \bot$.
        We first prove that $p^*.\nec.\hightipQCview = v$.
        In order for $p^*.\nec$ to be formed, a correct validator $i$ must have received an NE request accompanied by a TC $\mathit{tc}$ with $\mathit{tc}.\view = v' - 1 = v + 1$ and $\mathit{tc}.\hightip \neq \bot$.
        Given that $Z(v + 1)$ holds, $\mathit{tc}.\hightip$ strictly extends proposal $p$ and $\mathit{tc}.\hightip.\view \geq v + 1$.
        Moreover, given that $\mathit{tc}.\view = v + 1$, it follows that $\mathit{tc}.\hightip.\view = v + 1$ (due to the $\ValidTC$ function).
        Given that $Z(v + 1)$, we additionally have that $\mathit{tc}.\hightip.\blockheader.\qc.\view \geq v$.
        Because $\mathit{tc}.\hightip.\view = v + 1$ and $\mathit{tc}.\hightip.\blockheader.\qc.\view < v + 1$, we have that
        $p^*.\nec.\hightipQCview$ indeed must be equal to $v$.

        Moreover, observe that $p^*.\block.\qc.\view = p^*.\nec.\hightipQCview = v$ (by the $\SafetyCheck$ function).
        Because a QC is formed for proposal $p$ in view $v$, we obtain $p^*.\block.\qc.\blockid = p.\blockid$ (by \Cref{lemma:jovan_qc_intersection}).
        Thus, $p^*$ strictly extends proposal $p$, establishing $X(v' = v + 2)$.


        
        

        \item Let $p^*.\nec = \bot$ and $p^*.\tc.\highQC \neq \bot$.
        Note that $p^*.\tc.\view = v + 1$.
        By quorum intersection (since $|S| \geq f + 1$ and $p^*.\tc$ contains information from $2f + 1$ validators), $\mathit{tc}$ includes information from at least one validator in $S$.
        Given that $Y(v + 1)$ holds and all validators from the set $S$ update their local tip to a tip whose view is $v + 1$ and whose block header matches the block header of $p'.\block$ (from the first paragraph of the proof) before (potentially) contributing to $p^*.\tc$, it follows that $p^*.\tc.\tipsviews$ includes $v + 1$.
        (Note that each validator $i \in S$ reports its local tip and not its $\localhighQC$ variable, since its local tip is from view $v + 1$, whereas its $\localhighQC$ variable cannot hold any QC from any view greater than $v$.)
        Therefore, it must hold that $p^*.\tc.\highQC.\view \geq v + 1$.
        However, the fact that $p^*.\tc.\view = v + 1$, implies that $p^*.\tc.\highQC.\view \leq v$ (due to the $\ValidTC$ function).
        Therefore, this case cannot occur.
    
    \end{itemize}
    
    \item Suppose $p^*$ is a reproposal.
    In this case, $p^*$ is a reproposal of $p^*.\tc.\hightip$.
    Since $p^*.\tc.\view = v + 1$, $Z(v + 1)$ implies that $p^*$ strictly extends $p$ and $p^*.\block.\qc.\view \geq v$, which proves the $X(v' = v + 2)$ statement in this case.
\end{itemize}


\smallskip
We now prove that the inductive step holds.

\smallskip
\noindent \underline{$Y(v' > v + 1)$:}
Consider a validator $i \in S$ that modifies its local tip in view $v' > v + 1$.
Let us consider all possible cases for that to happen:
\begin{itemize}
    \item Let $i$ modify its local tip upon receiving a fresh proposal $p^*$ in view $v' > v + 1$.
    We further consider three possibilities:
    \begin{itemize}
        \item Let $p^*.\block.\qc.\view = p^*.\view - 1$.
        In this case, validator $i$ updates its local tip to $T(p^*)$, the tip of proposal $p^*$, which satisfies: $T(p^*).\view > v + 1$ and $T(p^*).\blockheader.\qc.\view \geq v + 1$ (as $p^*.\block.\qc.\view = v' - 1 \geq v + 1$).
        Thus, $Y(v' > v + 1)$ holds in this case.

        \item Let $p^*.\nec \neq \bot$.
        In this case, validator $i$ updates its local tip to the tip $T(p^*)$ of proposal $p^*$ with $T(p^*).\view > v + 1$ (as $p^*.\view = v' > v + 1$).
        We now prove that $T(p^*).\blockheader.\qc.\view = p^*.\nec.\hightipQCview \geq v$.
        In order for $p^*.\nec$ to be formed, a correct validator $i$ must have received an NE request accompanied by a TC $\mathit{tc}$ with $\mathit{tc}.\view = v' - 1 \geq v + 1$.
        Given that $Z(v' - 1 \geq v + 1)$ holds, $\mathit{tc}.\hightip.\blockheader.\qc.\view \geq v$, which implies that $p^*.\nec.\hightipQCview \geq v$.




        \item Let $p^*.\nec = \bot$ and $p^*.\tc.\highQC \neq \bot$.
        In this case, validator $i$ updates its local tip to the tip $T(p^*)$ of proposal $p^*$ with $T(p^*).\view > v + 1$ (as $p^*.\view = v' > v + 1$).
        Moreover, a parent of $p^*$ is a proposal to which $p^*.\tc.\highQC$ points (guaranteed by the $\SafetyCheck$ function).
        Given that $p^*.\tc.\view = v' - 1 \geq v + 1$ and $W(v' - 1)$ holds,  $p^*.\block.\qc.\view \geq v + 1$.
        Hence, $Y(v' > v + 1)$ is satisfied even in this case.
    \end{itemize}

    \item Let validator $i$ modify its local tip upon receiving a reproposal $p^*$.
    In this case, validator $i$ updates its local tip to $p^*.\tc.\hightip$ (note that $p^*.\tc.\view = v' - 1 \geq v + 1$ and $p^*.\tc.\hightip \neq \bot$ in this case).
    As $Z(v' - 1 \geq v + 1)$ holds, $p^*.\tc.\hightip.\view \geq v + 1$ and $p^*.\tc.\hightip.\blockheader.\qc.\view \geq v$, which proves $Y(v' > v + 1)$.
\end{itemize}
As $Y(v' > v + 1)$ is preserved in all possible cases, the statement holds.

\smallskip
\noindent \underline{$Q(v' > v + 1)$:}
Let $\mathit{qc}$ be any QC with $\mathit{qc}.\view = v' > v + 1$.
We distinguish between two cases:
\begin{itemize}
    \item Suppose that at least one correct validator contributing to $\mathit{qc}$ did so upon receiving a proposal in view $v'$.  
    In this case, the received proposal strictly extends proposal $p$, since $X(v')$ holds. 
    Hence, $Q(v' > v + 1)$ is satisfied.

    \item Suppose that no correct validator contributing to $\mathit{qc}$ did so upon receiving a proposal in view $v'$.  
    In this case, all correct votes issued in view $v'$ are based on timeout messages, each corresponding to the local tip of the validator at that time.  
    Since $Y(x)$ holds for all views $x \in [v + 1, v']$, and each validator in $S$ updates its local tip in view $v + 1$ to a tip that strictly extends $p$ (by the first paragraph of the proof) before potentially issuing any timeout-based vote in view $v'$, it follows that at least one correct validator $i \in S$ must have contributed to $\mathit{qc}$ after having voted for proposal $p'$.
    If validator $i$ updates its local tip for this vote in view $v + 1$, then---since a correct validator updates its local tip at most once per view, and it has already done so in view $v + 1$ to a tip that strictly extends $p$ (by the first paragraph of the proof)---the vote must target a block that strictly extends $p$.  
    Otherwise, the vote must also target a block that strictly extends $p$ by $X(v^*)$, for any $v^* \in [v + 2, v']$.
    In both cases, $Q(v' > v + 1)$ is satisfied.

\end{itemize}


\smallskip
\noindent \underline{$W(v' > v + 1)$:}
Let $\mathit{tc}$ be any TC with $\mathit{tc}.\view = v' > v + 1$ and $\mathit{tc}.\highQC \neq \bot$.
As (1) $|S| \geq f + 1$, and (2) $\mathit{tc}$ includes information about $2f + 1$ validators, information about at least one validator $i \in S$ is included in $\mathit{tc}$.
Furthermore, since (1) validator $i$ updates its local tip to a tip whose view is $v + 1$ and that strictly extends proposal $p$ (from the first paragraph of the proof) before contributing to $\mathit{tc}$, and (2) $Y(x)$ holds for every view $x \in [v + 1, v']$, it follows that $\mathit{tc}.\highQC.\view \geq v + 1$ (as validator $i$ reports either its local tip from a view $\geq v + 1$ or its $\localhighQC$ variable from a view $\geq v + 1$).
Moreover, the validity of $\mathit{tc}$ dictates that $\mathit{tc}.\highQC.\view < v'$ (see the $\ValidTC$ function).
Lastly, since $Q(x)$ holds for every view $x \in [v + 1, v' - 1]$, it must be that a proposal to which $\mathit{tc}.\highQC$ points strictly extends proposal $p$, thus proving $W(v' > v + 1)$.

\smallskip
\noindent \underline{$Z(v' > v + 1)$:}
Let $\mathit{tc}$ be any TC with $\mathit{tc}.\view = v' > v + 1$ and $\mathit{tc}.\hightip \neq \bot$.
Since (1) $|S| \geq f + 1$, and (2) $\mathit{tc}$ includes information about $2f + 1$ validators, it must be that information about at least one validator $i \in S$ is included in $\mathit{tc}$.
Given that (1) validator $i$ updates its local tip to a tip whose view is $v + 1$ and that strictly extends proposal $p$ (from the first paragraph of the proof) before contributing to $\mathit{tc}$, and (2) $Y(x)$ holds for every view $x \in [v + 1, v']$, $\mathit{tc}.\hightip.\view \geq v + 1$ (as validator $i$ reports either its local tip from a view $\geq v + 1$ or its $\localhighQC$ variable from a view $\geq v + 1$).
We now consider two possible cases:
\begin{itemize}
    \item Let $\mathit{tc}.\hightip.\view = v + 1$.
    Given that (1) $\mathit{tc}.\hightip \neq \bot$, (2) $\mathit{tc}.\hightip.\view = v + 1$, (3) validator $i$ updated its local tip to a tip whose view is $v + 1$ and whose block header matches the block header of $p'.\block$ (from the first paragraph of the proof) before contributing to $\mathit{tc}$, and (4) $Y(x)$ holds for every view $x \in [v + 1, v']$, it must follow that $\mathit{tc}.\tipsviews[i] = v + 1$ and $\mathit{tc}.\qcsviews[i] = v$.
    (Note that if $i$ were to include a QC in its timeout message, this would imply that its view is $\geq v + 1$, which in turn would force $\mathit{tc}.\hightip = \bot$.)
    Therefore, $\mathit{tc}.\hightip.\blockheader.\qc.\view = v$.
    Lastly, given that a QC is formed in view $v$ for proposal $p$, \Cref{lemma:jovan_qc_intersection} ensures that a parent of $\mathit{tc}.\hightip$ must be $p$.
    This implies that $\mathit{tc}.\hightip$ strictly extends proposal $p$, which proves $Z(v' > v + 1)$ in this case.

    \item Let $\mathit{tc}.\hightip.\view > v + 1$.
    Note that $\mathit{tc}.\hightip.\view \leq v'$.
    As $X(\mathit{tc}.\hightip.\view)$ is satisfied, we know that $\mathit{tc}.\hightip$ strictly extends proposal $p$ and it holds that $\mathit{tc}.\hightip.\blockheader.\qc.\view \geq v$.
    Thus, $Z(v' > v + 1)$ holds here as well.
\end{itemize}

\smallskip
\noindent \underline{$X(v' > v + 2)$:}
Let us consider all possible cases for a valid proposal $p^*$ with $p^*.\view = v' > v + 2$:
\begin{itemize}
    \item Suppose $p^*$ is a fresh proposal.
    We further consider three possibilities:
    \begin{itemize}
        \item Let $p^*.\block.\qc.\view = v' - 1 \geq v + 2$.
        In this case, a parent of $p^*$ strictly extends proposal $p$ (because $Q(v' - 1 \geq v + 2)$ holds), which implies that $p^*$ strictly extends proposal $p$.
        Moreover, $p^*.\block.\qc.\view \geq v + 2$, which concludes the proof of $X(v' > v + 2)$ in this case.

        \item Let $p^*.\nec \neq \bot$.
        We first prove that $p^*.\nec.\hightipQCview = p^*.\block.\qc.\view \geq v$.
        In order for $p^*.\nec$ to be formed, a correct validator $i$ must have received an NE request accompanied by a TC $\mathit{tc}$ with $\mathit{tc}.\view = v' - 1 \geq v + 2$ and $\mathit{tc}.\hightip \neq \bot$.
        Given that $Z(v' - 1)$ holds, $\mathit{tc}.\hightip.\blockheader.\qc.\view \geq v$.
        Therefore, $p^*.\nec.\hightipQCview = p^*.\block.\qc.\view \geq v$.
        Hence, it is left to prove that $p^*$ strictly extends proposal $p$.
        We now differentiate two cases:
        \begin{itemize}
            \item Let $p^*.\block.\qc.\view = v$.
            In this case, a parent of $p^*$ must be $p$ (by \Cref{lemma:jovan_qc_intersection}), which proves the statement of the invariant.

            \item Let $p^*.\block.\qc.\view > v$.
            In this case, the fact that $Q(p^*.\block.\qc.\view \geq v + 1)$ ensures that the statement of the invariant holds.
        \end{itemize}
            



        \item Let $p^*.\nec = \bot$ and $p^*.\tc.\highQC \neq \bot$.
        Note that $p^*.\tc.\view = v' - 1$.
        By quorum intersection (since $|S| \geq f + 1$ and $p^*.\tc$ contains information from $2f + 1$ validators), $p^*.\tc$ includes information from at least one validator in $S$.
        Given that $Y(x)$ holds for all views $v \in [v + 1, v' - 1]$ and all validators from the set $S$ update their local tip to a tip whose view is $v + 1$ and that strictly extends proposal $p$ (from the first paragraph of the proof) before (potentially) contributing to $p^*.\tc$ (in view $v + 1$), it follows that $p^*.\tc.\highQC.\view \geq v + 1$; note that $p^*.\tc.\highQC.\view < v' - 1$.    
        Next, $Q(p^*.\tc.\highQC.\view)$ shows that $p^*.\tc.\highQC$ points to a proposal that strictly extends proposal $p$.
        Since $p^*.\block.\qc = p^*.\tc.\highQC$ (ensured by the $\SafetyCheck$ function), $p^*$ strictly extends proposal $p$.
        Moreover, $p^*.\block.\qc.\view \geq v + 1$, which concludes the proof of $X(v' > v + 2)$ in this case.
    \end{itemize}
    
    \item Suppose $p^*$ is a reproposal.
    In this case, $p^*$ is a reproposal of $p^*.\tc.\hightip$; note that $p^*.\tc.\view = v' - 1 \geq v + 2$.
    Due to $Z(v' - 1)$, we have that $p^*$ strictly extends proposal $p$ (as $p^*.\tc.\hightip$ does so) and $p^*.\block.\qc.\view \geq v$ (as $p^*.\tc.\hightip.\blockheader.\qc.\view \geq v$).
    Thus, $X(v' > v + 2)$ holds in this case as well.
\end{itemize}
As $X(v' > v + 2)$ holds in all possible scenarios, the proof of the lemma is completed.
\qed
\end{proof}

The following lemma shows that once $f + 1$ correct validators vote for the “second-level’’ QC of a proposal, none of them can later vote for any proposal that does not strictly extend that proposal.

\begin{lemma}\label{lemma:jovan_safety_crucial_2}
Let $p$ be any proposal with $p.\view = v$.  
Suppose that at least $f + 1$ correct validators vote for a proposal $p'$ satisfying
\[
p'.\view = v + 1, 
\quad 
p'.\block.\qc.\view = v, 
\quad \text{and} \quad 
p'.\block.\qc.\blockid = p.\blockid.
\]
Then, if any correct validator that voted for $p'$ in view $v + 1$ later votes for a proposal $p^*$ in some view $v^* > v + 1$, the proposal $p^*$ must strictly extend $p$.
\end{lemma}
\begin{proof}
Let $S$ denote the set of correct validators, with $|S| \geq f + 1$, that voted for fresh proposal $p'$ in view $v + 1$.  
Assume, for the sake of contradiction, that there exists a correct validator $\ell \in S$ that votes for a proposal $p^*$ in some view $v^* > v + 1$ such that $p^*$ does not strictly extend $p$.  
We know that, in view $v + 1$, validator $\ell$’s local tip stores a block header that matches the block header of proposal $p'$ (by \Cref{lemma:jovan_local_tip}).  
The fact that $\ell$ issues a vote for proposal $p^*$ (that does not strictly extend proposal $p$) in view $v^*$ necessarily implies that $\ell$ updates its local tip to contain a block header that does not strictly extend proposal $p$ after voting for proposal $p'$.  
Indeed, if validator $\ell$ casts its vote for $p^*$ at line~\ref{line:send_vote_normal} of \Cref{Algorithm:Consensus-Execution_1}, it updates its local tip at line~\ref{Alg:update_tip} of the same algorithm;  
otherwise, if the vote is issued via a timeout message, it is created at line~\ref{line:create_vote_timeout} of \Cref{Algorithm:Utilities}.  
If validator $\ell$ had not updated its local tip to contain a block header that does not strictly extend proposal $p$ after voting for $p'$, then $p^*$ would necessarily strictly extend $p$.  
Hence, validator $\ell$ must have updated its local tip to contain a block header that does not strictly extend proposal $p$ after voting for $p'$.  

Note that this update cannot occur in view $v + 1$: once $\ell$ votes for $p'$, it sets its $\mathsf{highest\_voted\_view}$ variable to $v + 1$ (line~\ref{line:highest_voted_view_update_proposal} or line~\ref{line:update_highest_voted_view_timeout} of \Cref{Algorithm:Consensus-Execution_1}), which prevents it from updating its local tip in view $v + 1$ (line~\ref{Alg:update_tip} of \Cref{Algorithm:Consensus-Execution_1}).  
Therefore, validator $\ell$ must update its local tip to contain a block header that does not strictly extend proposal $p$ in some view strictly greater than $v + 1$.
However, by \Cref{lemma:jovan_safety_crucial}, every valid proposal from any view greater than $v + 1$---and hence every proposal that $\ell$ could use to update its local tip---must strictly extend $p$.
This contradicts the assumption that validator $\ell$ updates its local tip in some view greater than $v + 1$ to contain a block header that does not strictly extend $p$, and it completes the proof.
\qed
\end{proof}

Finally, we are ready to prove that \sysname satisfies safety.

\begin{theorem} [Safety] 
\label{theorem:jovan_safety}
No two correct validators commit different blocks at the same log position in their local logs.
\end{theorem}
\begin{proof}
Suppose, toward contradiction, that two correct validators $i$ and $j$ commit distinct blocks $B_i$ and $B_j$ at the same log position, and that these blocks are committed as part of proposals $p_i$ and $p_j$, respectively.
Note that proposals $p_i$ and $p_j$ are conflicting; that is, neither extends the other.
(Otherwise, the blocks $B_i$ and $B_j$ could not have been committed at the same log position.)
Without loss of generality, let $p_i.\view \leq p_j.\view$.
This implies the existence of a QC for proposal $p_i$ in view $p_i.\view$, another QC for the child of $p_i$ in view $p_i.\view + 1$, and a QC for proposal $p_j$ in view $p_j.\view$ (by the commitment rules of \sysname).
We now consider three possible cases:
\begin{itemize}
    \item $p_i.\view = p_j.\view$: 
    \Cref{lemma:jovan_qc_intersection} shows that $p_i.\blockid = p_j.\blockid$.
    However, this is impossible as the blocks $B_i$ and $B_j$ could not have been committed at the same log position.
    

    \item $p_i.\view + 1 = p_j.\view$: 
    Again, \Cref{lemma:jovan_qc_intersection} proves that $p_j.\blockid = c_i.\blockid$, where $c_i$ denotes the child of proposal $p_i$. 
    This cannot happen as the blocks $B_i$ and $B_j$ could not have been committed at the same log position.

    \item $p_i.\view + 1 < p_j.\view$: 
    In this case, there exists at least one correct validator that contributed both to the QC for the child of $p_i$ in view $p_i.\view + 1$ and to the QC for proposal $p_j$ in view $p_j.\view$.  
By \Cref{lemma:jovan_safety_crucial_2}, this implies that $p_j$ strictly extends $p_i$.  
This is impossible, however, because block $B_j$ conflicts with block $B_i$.
\end{itemize}
As neither of these cases can occur, the proof is complete.
\qed
\end{proof}

\subsection{Tail-Forking Resistance}
\label{sub:tf-analysis}
Recall from \Cref{sub:consensus_properties} that the tail-forking resistance property asserts the following:
if a leader proposes a block $B$ as part of a fresh proposal $p$, at least $f + 1$ correct validators vote for $p$, and the leader does not equivocate, then whenever any correct validator commits a block $B'$, it must hold that either $B$ extends $B'$ or $B'$ extends $B$.
We begin by showing that, under the conditions stated above, no NEC can be formed.

\begin{lemma} \label{lemma:jovan_no_nec}
Let $p$ be any fresh proposal with $p.\view = v$, issued by the leader of view $v$.
Then, no correct validator that voted for $p$ in view $v$ issues an $\NE$ message upon receiving a message  
$\langle \NERequest, \mathit{tc} \rangle$, where $\mathit{tc}$ is a TC whose $\mathit{tc}.\hightip$ satisfies $\mathit{tc}.\hightip.\blockid = p.\blockid$.
\end{lemma}
\begin{proof}
By contradiction, assume that some correct validator $i$ that voted for proposal $p$ in view $v$ issues an $\NE$ message upon receiving the above $\langle \NERequest, \mathit{tc} \rangle$ message.
Since $p$ is a fresh proposal from view $v$, the $\ValidTC$ check implies that $\mathit{tc}.\view \geq v$ (as $\mathit{tc}.\hightip.\view = v$).
Observe that validator $i$ must have voted for $p$ before receiving the $\NERequest$.
Indeed, if $i$ processed the $\NERequest$ first, it would increment its $\curView$ to a value $\geq v + 1$ (line~\ref{line:increment_view_ne_request} of \Cref{Algorithm:Consensus-Execution_2}), and therefore could not subsequently cast a vote for $p$ in view $v$.

Thus, validator $i$ votes for $p$ and only later receives the $\NERequest$.
However, in this case $i$ cannot issue an $\NE$ message, because the condition at line~\ref{Has_not_voted_for_higthip} of \Cref{Algorithm:Recovery} fails.
This yields a contradiction, and the lemma follows.
\qed
\end{proof}

Next, we show that any fresh proposal by a non-equivocating leader receiving at least $f+1$ votes from correct validators can never be abandoned.

\begin{lemma} \label{lemma:jovan_crucial_tail_forking}
Let $p$ be any fresh proposal with $p.\view = v$, issued by the leader of view $v$.
Suppose that the leader of view $v$ does not equivocate and that at least $f + 1$ correct validators vote for $p$ in view $v$.
Then, any valid proposal $p^*$ issued in a view $v^* \geq v + 1$ (i.e., $p^*.\view = v^* \geq v + 1$) must extend $p$.
\end{lemma}
\begin{proof}
Let $S$ denote the set of correct validators, with $|S| \geq f + 1$, that voted for fresh proposal $p$ in view $v$.  
For every correct validator $i \in S$, the following holds.  
At the moment $i$ votes for proposal $p$ in view $v$, \Cref{lemma:jovan_local_tip} implies that  
(1) $\localhightip.\blockheader$ matches the block header of $p.\block$, and  
(2) $\localhightip.\view = v$.  
Hence, since validator $i$ cannot update its $\localhightip$ variable to a tip with view $v$ in any view smaller than $v$, it follows that $i$ must have updated its local tip in view $v$ to a tip that extends proposal $p$.

We prove the claim by induction on the view number $v'$.
Consider the following four induction hypotheses, $X(v')$, $Y(v')$, $Q(v')$, and $Z(v')$:
\begin{itemize}
    \item $X(v')$: Any (valid) proposal $p^*$ in view $v'\geq v$ extends $p$.
    Consequently, any (valid) tip $T$ of a fresh proposal from view $v' \geq v$ also extends $p$.
    Moreover, if $v' > v$, then $T.\blockheader.\qc.\view \geq v$.
     
    \item $Y(v')$: If any validator $i \in S$ updates its local tip to some tip $T$ in view $v'\ge v$, then $T.\view \geq v$.

    \item $Q(v')$: For any QC $\mathit{qc}$ with $\mathit{qc}.\view = v'\ge v$, any proposal $p_{\mathit{qc}}$ with $p_{\mathit{qc}}.\blockid = \mathit{qc}.\blockid$ extends proposal $p$.

    
    \item $Z(v')$: Let $\mathit{tc}$ be any TC with $\mathit{tc}.\view = v' \ge v$ and $\mathit{tc}.\hightip \neq \bot$. Then, the proposal whose tip is $\mathit{tc}.\hightip$ extends $p$ and $\mathit{tc}.\hightip.\view \geq v$.
    Moreover, if $\mathit{tc}.\hightip.\view > v$, then $\mathit{tc}.\hightip.\blockheader.\qc.\view \geq v$.
\end{itemize}

We begin by proving the base case. Note that $X(v')$ is equivalent to the lemma’s statement.

\medskip
\noindent \underline{$X(v' = v)$:}
This follows directly from the fact that the leader of view $v$ issues only proposal $p$ in that view.

\medskip
\noindent \underline{$Y(v' = v)$:}
Consider any correct validator $i \in S$.  
By the first paragraph of the proof, validator $i$ updates its local tip in view $v$ to a tip whose view is $v$.  
Since $i$ can update its local tip at most once per view, it follows that $Y(v' = v)$ holds.


\medskip
\noindent \underline{$Q(v' = v)$:}
Let $\mathit{qc}$ be any QC with $\mathit{qc}.\view = v$.
Given that at least $f + 1$ correct validators voted for proposal $p$ in view $v$, \Cref{lemma:jovan_same_vote} ensures that $\mathit{qc}$ points to proposal $p$.
Thus, $Q(v' = v)$ holds.

%

\medskip
\noindent \underline{$Z(v' = v)$:}
Let $\mathit{tc}$ be any TC with $\mathit{tc}.\view = v$ and $\mathit{tc}.\hightip \neq \bot$.
The standard quorum intersection (as $|S| \geq f + 1$ and $\mathit{tc}$ contains information about $2f + 1$ validators) ensures that $\mathit{tc}$ includes information about at least one validator $i \in S$.
Because $Y(v)$ holds and $i$ updated its local tip to a tip whose view is $v$ and that extends proposal $p$ (from the first paragraph of the proof) before contributing to $\mathit{tc}$ (in view $v$), $\mathit{tc}.\hightip.\view \geq v$.
(Note that validator $i$ reports its local tip and not its $\localhighQC$ variable, since its local tip is from view $v$, whereas its $\localhighQC$ variable cannot hold any QC from any view greater than $v - 1$.)
Since $\mathit{tc}.\hightip.\view \leq \itc.\view = v$ (by the $\ValidTC$ function), we get that $\mathit{tc}.\hightip.\view = v$.
As the only proposal issued in view $v$ is $p$, it follows that $\mathit{tc}.\hightip$ extends proposal $p$, which implies that $Z(v' = v)$ holds.

\smallskip
We now prove that the inductive step holds.

\smallskip
\noindent \underline{$X(v' > v)$:}
Let us consider all possible cases for a valid proposal $p^*$ with $p^*.\view = v' > v$:
\begin{itemize}
    \item Suppose $p^*$ is a fresh proposal.
    We further consider three possibilities:
    \begin{itemize}
        \item Let $p^*.\block.\qc.\view = v' - 1 \geq v$.
        In this case, a parent of $p^*$ extends proposal $p$ (because $Q(v' - 1 \geq v)$ holds), which implies that $p^*$ extends proposal $p$.
        Thus, $X(v' > v)$ holds in this case.

        \item Let $p^*.\nec \neq \bot$.
        By quorum intersection, there exists a correct validator $i \in S$ that issues an NE message upon receiving a message $\langle \NERequest, \mathit{tc} \rangle$ in view $v' > v$, where $\mathit{tc}.\view = v' - 1 \geq v$ and $\mathit{tc}.\hightip \neq \bot$.  
        Since $Z(v' - 1 \geq v)$ holds, $\mathit{tc}.\hightip$ extends $p$ and $\mathit{tc}.\hightip.\view \geq v$.
        We now show that $\mathit{tc}.\hightip.\view > v$.
        Assume, for contradiction, that $\mathit{tc}.\hightip.\view = v$.
        Since $p$ is the unique proposal issued in view $v$, we have $\mathit{tc}.\hightip.\blockid = p.\blockid$, contradicting the fact that validator $i$ issued an $\NE$ message (by \Cref{lemma:jovan_no_nec}).
        
        Next, by $Z(v' - 1 \geq v)$, $\mathit{tc}.\hightip.\blockheader.\qc.\view \geq v$.
        Thus, $p^*.\nec.\hightipQCview \geq v$.  
        At the same time, the validity of $p^*.\nec$ implies that 
        \[
        p^*.\nec.\hightipQCview < v' - 1.
        \]
        Since a parent of $p^*$ is a proposal to which a QC from view $p^*.\nec.\hightipQCview$ points (as enforced by the $\SafetyCheck$ function), and since $Q(p^*.\nec.\hightipQCview \in [v,\, v'-2])$ holds, it follows that $p^*$ extends $p$.

            

        
        \item Let $p^*.\nec = \bot$ and $p^*.\tc.\highQC \neq \bot$.
        Note that $p^*.\tc.\view = v' - 1 \geq v$.
        By quorum intersection (since $|S| \geq f + 1$ and $p^*.\tc$ contains information from $2f + 1$ validators), $p^*.\tc$ includes information from at least one validator in $S$.
        Given that $Y(x)$ holds for all views $v \in [v, v' - 1]$ and all validators from the set $S$ update their local tip to a tip whose view is $v$ and that extends proposal $p$ (from the first paragraph of the proof) before contributing to $p^*.\tc$ (in view $v$), it follows that $p^*.\tc.\highQC.\view \geq v$; note that $p^*.\tc.\highQC.\view < v' - 1$ (by the $\ValidTC$ function). 
        Next, $Q(p^*.\tc.\highQC.\view \in [v, v' - 2])$ shows that $p^*.\tc.\highQC$ points to a proposal that extends proposal $p$.
        Since $p^*.\block.\qc = p^*.\tc.\highQC$ (ensured by the $\SafetyCheck$ function), $p^*$ extends proposal $p$, which concludes the proof of $X(v' > v)$ in this case.
    \end{itemize}
    
    \item Suppose $p^*$ is a reproposal.
    In this case, $p^*$ is a reproposal of $p^*.\tc.\hightip$; note that $p^*.\tc.\view = v' - 1 \geq v$.
    Due to $Z(v' - 1)$, we have that $p^*$ extends proposal $p$ (as $p^*.\tc.\hightip$ does so).
    Thus, $X(v' > v)$ holds in this case as well.
\end{itemize}
Therefore, $X(v' > v)$ holds in all possible scenarios.

\smallskip
\noindent \underline{$Y(v' > v)$:}
Consider a validator $i \in S$ that modifies its local tip in view $v' > v$.
Let us consider all possible cases for that to happen:
\begin{itemize}
    \item Let $i$ modify its local tip upon receiving a fresh proposal $p^*$ in view $v' > v$.
    In this case, $i$ updates its local tip to some tip $T$ with $T.\view = v' > v$, which proves $Y(v' > v)$ in this scenario.
    
    \item Let validator $i$ modify its local tip upon receiving a reproposal $p^*$.
    Here, $p^*.\tc \neq \bot$.
    In this case, validator $i$ updates its local tip to $p^*.\tc.\hightip$ (note that $p^*.\tc.\view = v' - 1 \geq v$).
    By quorum intersection (since $|S| \geq f + 1$ and $p^*.\tc$ contains information from $2f + 1$ validators), $p^*.\tc$ includes information from at least one validator in $S$.
    Given that $Y(x)$ holds for all views $v \in [v, v' - 1]$ and all validators from the set $S$ update their local tip to a tip whose view is $v$ and that extends proposal $p$ (from the first paragraph of the proof) before contributing to $p^*.\tc$ (in view $v$), it follows that $p^*.\tc.\tipsviews$ includes $\geq v$.
    Hence, $p^*.\tc.\hightip.\view \geq v$, which proves $Y(v' > v)$ even in this case.
\end{itemize}
As $Y(v' > v)$ is preserved in all possible cases, the statement holds.

\smallskip
\noindent \underline{$Q(v' > v)$:}
Let $\mathit{qc}$ be any QC with $\mathit{qc}.\view = v' > v$.  
We distinguish between two cases:
\begin{itemize}
    \item Suppose that at least one correct validator contributing to $\mathit{qc}$ did so upon receiving a proposal in view $v'$.  
    In this case, the received proposal extends proposal $p$, since $X(v')$ holds.  
    Hence, $Q(v' > v)$ is satisfied.

    \item Suppose that no correct validator contributing to $\mathit{qc}$ did so upon receiving a proposal in view $v'$.  
    In this case, all correct votes issued in view $v'$ are based on timeout messages, each corresponding to the local tip of the validator at that time.  
    Since $Y(x)$ holds for all views $x \in [v, v']$, and each validator in the set $S$ updates its local tip to a tip whose view is $v$ and that extends proposal $p$ (from the first paragraph of the proof) before contributing to $\mathit{qc}$, it follows that at least one correct validator $i \in S$ must have contributed to $\mathit{qc}$ after having voted for proposal $p$.
    If validator $i$ updates its local tip for this vote in view $v$, then---since a correct validator updates its local tip at most once per view, and it has already done so in view $v$ to a tip that extends $p$ (by the first paragraph of the proof)---the vote must target a proposal that extends $p$.  
    Otherwise, the vote must also target a proposal that extends $p$ by $X(v^*)$, for any $v^* \in [v + 1, v']$.
    In both cases, $Q(v' > v)$ is satisfied.

\end{itemize}
Thus, $Q(v' > v)$ holds in both cases.


\smallskip
\noindent \underline{$Z(v' > v)$:}
Let $\mathit{tc}$ be any TC with $\mathit{tc}.\view = v' > v$ and $\mathit{tc}.\hightip \neq \bot$.
Since (1) $|S| \geq f + 1$, and (2) $\mathit{tc}$ includes information about $2f + 1$ validators, it must be that information about at least one validator $i \in S$ is included in $\mathit{tc}$.
Given that (1) validator $i$ updates its local tip to a tip whose view is $v$ and that extends proposal $p$ (from the first paragraph of the proof) before contributing to $\mathit{tc}$, and (2) $Y(x)$ holds for every view $x \in [v, v']$, a view $\geq v$ is included in $\mathit{tc}.\tipsviews$.
This further means that $\mathit{tc}.\hightip.\view \geq v$.
We now consider two possible cases:
\begin{itemize}
    \item Let $\mathit{tc}.\hightip.\view = v$.
    Since the only proposal issued in view $v$ is $p$, it follows that $\mathit{tc}.\hightip$ extends proposal $p$.
    Therefore, $Z(v' > v)$ is satisfied.

    \item Let $\mathit{tc}.\hightip.\view > v$.
    Note that $\mathit{tc}.\hightip.\view \leq v'$.
    As $X(\mathit{tc}.\hightip.\view)$ is satisfied, we know that $\mathit{tc}.\hightip$ extends proposal $p$ and $\mathit{tc}.\hightip.\blockheader.\qc.\view \geq v$.
    Thus, $Z(v' > v)$ holds here as well.
\end{itemize}
Given that all invariants have been established, the proof is complete.
\qed
\end{proof}

The following lemma shows that if a non-equivocating leader issues a fresh proposal $p$ and at least $f + 1$ correct validators vote for it, then none of these validators can later vote for any proposal that does not extend proposal $p$.

\begin{lemma} \label{lemma:jovan_tail_forking_crucial_2}
Let $p$ be any fresh proposal with $p.\view = v$, issued by the leader of view $v$.
Suppose that the leader of view $v$ does not equivocate and that at least $f + 1$ correct validators vote for $p$ in view $v$.
Then, if any correct validator that voted for $p$ in view $v$ later votes for a proposal $p^*$ in some view $v^* \geq v + 1$, the proposal $p^*$ must extend $p$.
\end{lemma}
\begin{proof}
Let $S$ denote the set of correct validators, with $|S| \geq f + 1$, that voted for fresh proposal $p$ in view $v$.
Assume, for the sake of contradiction, that there exists a correct validator $\ell \in S$ that votes for a proposal $p^*$ in some view $v^* \geq v + 1$ such that $p^*$ does not extend $p$.  
We know that, in view $v$, validator $\ell$’s local tip stores a block header that matches the block header of proposal $p$'s block (by \Cref{lemma:jovan_local_tip}).  
The fact that $\ell$ issues a vote for proposal $p^*$ (that does not extend proposal $p$) in view $v^*$ necessarily implies that $\ell$ updates its local tip to contain a block header that does not extend proposal $p$ after voting for proposal $p$.  
Indeed, if validator $\ell$ casts its vote for $p^*$ at line~\ref{line:send_vote_normal} of \Cref{Algorithm:Consensus-Execution_1}, it updates its local tip at line~\ref{Alg:update_tip} of the same algorithm;  
otherwise, if the vote is issued via a timeout message, it is created at line~\ref{line:create_vote_timeout} of \Cref{Algorithm:Utilities}.  
If validator $\ell$ had not updated its local tip to contain a block header that does not extend proposal $p$ after voting for $p$, then $p^*$ would necessarily extend $p$.
Hence, validator $\ell$ must have updated its local tip to contain a block header that does not extend proposal $p$ after voting for $p$.  

Note that this update cannot occur in view $v$: once $\ell$ votes for $p$, it sets its $\mathsf{highest\_voted\_view}$ variable to $v$ (line~\ref{line:highest_voted_view_update_proposal} or line~\ref{line:update_highest_voted_view_timeout} of \Cref{Algorithm:Consensus-Execution_1}), which prevents it from updating its local tip in view $v$ (line~\ref{Alg:update_tip} of \Cref{Algorithm:Consensus-Execution_1}).  
Therefore, validator $\ell$ must update its local tip to contain a block header that does not extend proposal $p$ in some view strictly greater than $v$.
However, \Cref{lemma:jovan_crucial_tail_forking} proves that every valid proposal from any view greater than $v$---and hence every proposal that $\ell$ could use to update its local tip---must extend $p$.
This contradicts the assumption that validator $\ell$ updates its local tip in some view greater than $v$ to contain a block header that does not extend $p$, and it completes the proof.
\qed
\end{proof}

We are now ready to prove the tail-forking resistance property.

\begin{theorem} [Tail-forking resistance] \label{theorem:jovan_tail_forking}
If a leader proposes a block $B$ as part of a fresh proposal $p$ and at least $f + 1$ correct validators vote for it, then---unless the leader equivocates---any block $B'$ committed by a correct validator must satisfy that either $B$ extends $B'$ or $B'$ extends $B$.
\end{theorem}
\begin{proof} 
Let $v$ denote $p$'s view.
First, observe that every correct validator that votes for proposal $p$ in view $v$ must have received some proposal $p^*$ in view $v$ whose block is identical to $p.\block$ (by \Cref{lemma:jovan_local_tip}).

Assume, for contradiction, that the leader does not equivocate and a correct validator $i$ commits a block $B'$ such that $B'$ neither extends $B$ nor is extended by $B$.
Let validator $i$ commit a block $B'$ as part of proposal $p'$, where $p'$ contains $B'$ or strictly extends a proposal that contains $B'$.
This implies the existence of a QC for proposal $p'$ in view $p'.\view$ and another QC for the child of $p'$ in view $p'.\view + 1$.
Note that proposals $p$ and $p'$ are conflicting (as $B$ and $B'$ do not extend each other).
%

From \Cref{lemma:jovan_tail_forking_crucial_2} and the fact that the leader does not equivocate, we have that $p'.\view \leq p.\view$. 
Now consider three possible cases:
\begin{itemize}
    \item $p'.\view = p.\view$.
    There exists at least one correct validator that votes for both proposal $p$ and proposal $p'$ in view $p'.\view = p.\view$.  
    Hence, $p'.\blockid = p.\blockid$ by \Cref{lemma:jovan_same_vote}.
    However, this is impossible as $B$ and $B'$ do not extend each other.
    

    \item $p'.\view + 1 = p.\view$: 
    There exists at least one correct validator that votes for both proposal $p$ and for the child $c'$ of proposal $p'$ in view $p'.\view + 1 = p.\view$.  
    Therefore, $c'.\blockid = p.\blockid$ by \Cref{lemma:jovan_same_vote}, which further implies that $B$ extends $B'$, which violates the starting assumption.
    
    
    \item $p'.\view + 1 < p.\view$:
    In this case, since every correct validator that voted for proposal $p$ in view $v$ receives some proposal $p^*$ with $p^*.\block = p.\block$ (by the first paragraph of the proof), \Cref{lemma:jovan_safety_crucial} implies that $p^*$ strictly extends proposal $p'$.
    Consequently, the block $B$ contained in $p^*$ strictly extends the block $B'$.
    This is, however, impossible as $B$ conflicts with $B'$.
\end{itemize}
As neither of the above cases can occur, this completes the proof. 
\qed
\end{proof}

An immediate corollary of \Cref{theorem:jovan_tail_forking} is that a speculatively finalized block contained in a fresh proposal that has collected a QC can only be reverted if the leader that issued the proposal equivocated.

\begin{corollary} [Reversion of speculatively committed blocks] 
\label{corollary:jovan_reversion}
Suppose a block $B$ contained in a fresh proposal $p$ is speculatively finalized by a correct validator, and another block $B'$ is committed by a correct validator such that neither $B$ extends $B'$ nor $B'$ extends $B$.
Then, the leader of view $p.\view$ equivocated.
\end{corollary}
\begin{proof}
Since block $B$ contained in a fresh proposal $p$ is speculatively finalized, at least $f + 1$ correct validators must have voted for $p$.
Assume, for contradiction, that the leader of view $p.\view$ does not equivocate and that some block $B'$ is committed by a correct validator such that neither $B$ extends $B'$ nor $B'$ extends $B$.
By \Cref{theorem:jovan_tail_forking}, this situation cannot occur (as either $B$ extends $B'$ or $B'$ extends $B$), which contradicts our assumption that the leader did not equivocate.
Hence, the theorem follows.
\qed
\end{proof}
\subsection{Liveness} 
\label{sub:analysis-liveness}

%

This subsection establishes that \sysname satisfies the liveness property: 
every block issued by a correct leader after GST is eventually committed.  
We emphasize that this guarantee applies to \emph{all} blocks, including those proposed by a correct leader as part of a reproposal.
We begin by defining what it means for a correct validator to enter a view.

\begin{definition} [Entering a view] \label{definition:jovan_view_entering}
A correct validator is said to \emph{enter} view $v$ when it assigns $v$ to its $\curView$ variable (line~\ref{line:enter_view} of \Cref{Algorithm:Pacemaker}).  
If this assignment occurs at real time $\tau$, we say that the validator enters view $v$ \emph{at time $\tau$}.  
\end{definition}

Next, we define what it means for a correct validator to accept a QC or a TC.

\begin{definition} [Accepting a QC or a TC]
A correct validator \emph{accepts} a QC $\mathit{qc}$ or a TC $\mathit{tc}$ if and only if it invokes $\IncrementView(\mathit{qc})$ or $\IncrementView(\mathit{tc})$, and this invocation updates its $\curView$ variable to $\mathit{qc}.\view + 1$ or $\mathit{tc}.\view + 1$, respectively.
\end{definition}
Recall that a correct validator $i$ accepts a QC $\mathit{qc}$ in one of the following situations:
\begin{itemize}
    \item upon receiving a proposal that carries it (line~\ref{line:increment_view_proposal} of \Cref{Algorithm:Consensus-Execution_1});
    
    \item upon forming it locally (line~\ref{line:increment_view_form_qc} of \Cref{Algorithm:Consensus-Execution_1} and line~\ref{line:increment_view_form_qc_timeout_message} of \Cref{Algorithm:Consensus-Execution_2});
    
    \item upon receiving it directly from another validator $j$ (line~\ref{line:increment_view_receive_qc} of \Cref{Algorithm:Consensus-Execution_1} if validator $j$ is the leader of view $\mathit{qc}.\view$ and line~\ref{line:increment_view_receive_qc_next_leader} of \Cref{Algorithm:Consensus-Execution_1} if validator $i$ is the leader of view $\mathit{qc}.\view + 1$);
    
    \item upon receiving it within a timeout message (line~\ref{line:increment_view_receive_qc_timeout_message} of \Cref{Algorithm:Consensus-Execution_2}).
\end{itemize}
Similarly, a correct validator accepts a TC in one of the following situations:
\begin{itemize}
    \item upon receiving a proposal that carries it (line~\ref{line:increment_view_proposal} of \Cref{Algorithm:Consensus-Execution_1}); 
    
    \item upon receiving it within a timeout message (line~\ref{line:increment_view_receive_qc_timeout_message} of \Cref{Algorithm:Consensus-Execution_2});
    
    \item upon forming it locally (line~\ref{line:increment_view_form_tc} of \Cref{Algorithm:Consensus-Execution_2});

    \item upon receiving it directly (line~\ref{line:increment_view_receive_tc} of \Cref{Algorithm:Consensus-Execution_2});
    
    \item upon receiving it within a proposal request or an NE request (line~\ref{line:increment_view_proposal_request} of \Cref{Algorithm:Consensus-Execution_2} and line~\ref{line:increment_view_ne_request} of \Cref{Algorithm:Consensus-Execution_2}).
\end{itemize}

Let $\tau_i(v)$ denote the time at which a correct validator $i$ enters view $v$, with $\tau_i(v) = \bot$ if validator $i$ never enters that view.
In the remainder of this section, we say that a correct validator \emph{fully times out} from a view if it issues a timeout triggered by its own local clock---namely, after $\viewduration$ time has elapsed according to that clock\footnote{Note that before GST, although a validator's local clock may measure $\viewduration$ time, the actual elapsed real time may differ, since local clocks are permitted to drift prior to GST.}---rather than due to the premature timeout invocation at line~\ref{line:premature_timeout} of \Mref{Algorithm:Pacemaker}.
Similarly, we say that a correct validator \emph{prematurely times out} if it triggers the timeout event at line~\ref{line:premature_timeout} of \Cref{Algorithm:Consensus-Execution_2}.

First, we show that the first correct validator to send a timeout message for any view $v$ must have fully timed out from that view (line~\ref{line:timeout_from_view} of \Cref{Algorithm:Consensus-Execution_1}) before doing so.

\begin{lemma} \label{lemma:jovan_first_timeout}
Let $i$ be the first correct validator that sends a timeout message for any view $v \geq 1$.
Then, validator $i$ has fully timed out from view $v$ before sending the message.
\end{lemma}
\begin{proof}
Assume, for the sake of contradiction, that validator $i$ does not send a timeout message upon fully timing out from view $v$.
The only remaining possibility is that $i$ sends its timeout message upon prematurely timing out from view $v$ (line~\ref{line:premature_timeout} of \Cref{Algorithm:Pacemaker}).
In this case, validator $i$ must have already received $f + 1$ timeout messages associated with view $v$.
Hence, $i$ cannot be the first correct validator to send a timeout message for view $v$, which concludes the proof.
\qed
\end{proof}

Next, we prove that if any correct validator issues a timeout message associated with any view $v$, then the validator has previously entered view $v$.

\begin{lemma} \label{lemma:jovan_timeout_message_enter_before}
If a correct validator $i$ issues a timeout message for some view $v \geq 1$, then $i$ must have entered view $v$ beforehand.
\end{lemma}
\begin{proof}
For validator $i$ to issue a timeout message for view $v$ (line~\ref{line:broadcast_timeout_message} of \Cref{Algorithm:Consensus-Execution_1}), it must have either fully or prematurely timed out from view $v$ beforehand.  
If $i$ fully times out from view $v$, then it necessarily entered view $v$ earlier.  
Otherwise, $i$ prematurely times out upon receiving $f + 1$ valid timeout messages (line~\ref{line:premature_timeout} of \Cref{Algorithm:Pacemaker}).  
Upon receiving the first such message, validator $i$ enters view $v$ (line~\ref{line:increment_view_receive_qc_timeout_message} of \Cref{Algorithm:Consensus-Execution_2}).  
Therefore, the claim holds in both cases.
\qed
\end{proof}

The following lemma establishes that a correct validator never issues an invalid timeout message.  
In particular, if a correct validator issues a timeout message for some view $v \geq 1$, then that message must include either a QC or a TC from view $v - 1$.

\begin{lemma} \label{lemma:jovan_valid_timeout}
Suppose that a correct validator $i$ issues a timeout message $m$ associated with some view $v \geq 1$.  
Then, $m$ is a valid timeout message.
\end{lemma}
\begin{proof}
If $v = 1$, then $m$ is trivially valid, as it includes a genesis QC from view $0$.


Consider $v > 1$.  
For message $m$ to be broadcast (line~\ref{line:broadcast_timeout_message} of \Mref{Algorithm:Consensus-Execution_1}), validator~$i$ must have already entered view~$v$ (by \Cref{lemma:jovan_timeout_message_enter_before}).
For validator~$i$ to enter view~$v$, it must have first accepted either a QC or a TC corresponding to view~$v - 1$ (see the $\IncrementView$ procedure at line~\ref{alg:Increment_view} of \Mref{Algorithm:Pacemaker}).
This QC or TC is subsequently included in the timeout message (see the $\CreateTimeoutMsg$ function at line~\ref{line:view_certificate_create_timeout_message} of \Cref{Algorithm:Utilities}), thereby completing the proof.
\qed
\end{proof}

We will implicitly rely on the above lemma throughout the proof.  
In particular, it guarantees that all timeout messages issued by correct validators are valid and therefore will be processed by every other correct validator.

The following lemma shows that for any view $v > 1$, if a correct validator enters view $v$, then there must exist some correct validator that previously entered view $v - 1$.

\begin{lemma} \label{lemma:jovan_enter_before}
Consider any correct validator $i$ and any view $v > 1$ such that $\tau_i(v) \neq \bot$.  
Then, there exists a correct validator $j$ for which $\tau_j(v-1) \neq \bot$ and $\tau_j(v-1) \leq \tau_i(v)$.
\end{lemma}
\begin{proof}
If a correct validator $i$ enters view $v > 1$, then a QC or a TC must have been previously formed for view $v - 1$.
Therefore, let us consider two possible cases:
\begin{itemize}
    \item Suppose that a QC is formed in view $v - 1$ before validator $i$ enters view $v$.  
    Then at least $f + 1$ correct validators must have voted in view $v - 1$.  
    Since a correct validator can cast a vote in view $v - 1$ only after entering that view---whether the vote is issued upon receiving a proposal (line~\ref{line:send_vote_normal} of \Cref{Algorithm:Consensus-Execution_1}) or through a timeout message (by \Cref{lemma:jovan_timeout_message_enter_before})---the lemma follows in this case.




    \item Suppose that a TC is formed in view $v - 1$ before validator $i$ enters view $v$.  
    This implies that at least $f + 1$ correct validators have issued timeout messages in view $v - 1$.  
    Each of these validators must have entered view $v - 1$ before sending its timeout message (by \Cref{lemma:jovan_timeout_message_enter_before}), and thus the lemma follows in this case as well.

    
\end{itemize}
Since the lemma holds in both possible cases, the proof is complete.
\qed
\end{proof}




The following lemma proves that every view is eventually entered by a correct validator, i.e., no views are ``skipped''.

\begin{lemma} \label{lemma:jovan_all_views_entered}
For every view $v \geq 1$, at least one correct validator enters view $v$.
\end{lemma}
\begin{proof}
We prove this by contradiction.  
Assume that there exists a view that no correct validator enters, and let $v_{\min}$ denote the smallest such view.  
Clearly, $v_{\min} > 1$, since every correct validator starts at view~$1$.
By \Cref{lemma:jovan_enter_before}, no correct validator can enter any view greater than $v_{\min}$.  
Moreover, by the minimality of $v_{\min}$, every view smaller than $v_{\min}$ is entered by at least one correct validator.  
Let $i$ be a correct validator that enters view~$v_{\min} - 1$.  
Since validator~$i$ never enters any view greater than $v_{\min} - 1$, it eventually issues a timeout message for view~$v_{\min} - 1$ (line~\ref{line:broadcast_timeout_message} of \Cref{Algorithm:Consensus-Execution_1}) upon timing out from view $v_{\min} - 1$.

Now, consider two cases: (i) $v_{\min} = 2$ and (ii) $v_{\min} > 2$. 
When $v_{\min} = 2$, then note that all correct validators start at $v_{\min} - 1 = 1$.
Alternatively, when $v_{\min} > 2$, the aforementioned timeout message associated with view $v_{\min} - 1$ carries either a QC or a TC from view~$v_{\min} - 2$ (as validator $i$ entered view $v_{\min} - 1$ upon accepting a QC or a TC from view $v_{\min} - 2$, by \Cref{lemma:jovan_valid_timeout}). 
Upon receiving this message, every correct validator that has not yet entered view $v_{\min} - 1$ enters view $v_{\min} - 1$ (line~\ref{line:increment_view_receive_qc_timeout_message} of \Cref{Algorithm:Consensus-Execution_2}).
Therefore, in both possible cases, all correct validators eventually enter view $v_{\min} - 1$.

Since no correct validator ever enters any view greater than $v_{\min} - 1$, all correct validators eventually broadcast their timeout messages for view~$v_{\min} - 1$ (line~\ref{line:broadcast_timeout_message} of \Cref{Algorithm:Consensus-Execution_1}, either because they fully or prematurely time out from view $v_{\min} - 1$).  
Given that there are at least $2f + 1$ correct validators and $2f + 1$ timeout messages are required to form a TC, at least one correct validator eventually constructs a QC or a TC for view $v_{\min} - 1$ and enters view $v_{\min}$ (line~\ref{line:increment_view_form_qc_timeout_message} or line~\ref{line:increment_view_form_tc} of \Cref{Algorithm:Consensus-Execution_2}).  
This contradicts the starting assumption and concludes the proof.
\qed
\end{proof}


\Cref{lemma:jovan_all_views_entered} is important because it guarantees that every view is eventually entered by at least one correct validator. We assume this lemma implicitly in the remainder of the proof.

\medskip 
Let $\tau(v)$ denote the time at which the first correct validator enters view $v$, i.e., 
\[
    \tau(v) \deq\,  \min \bigl( \{ \tau_i(v) \mid i \text{ is a correct validator and } \tau_i(v) \neq \bot \} \bigr).
\]
By \Cref{lemma:jovan_all_views_entered}, $\tau(v)$ is well-defined for every view $v$. Next, we define a post-GST view.

\begin{definition} [Post-GST view] \label{definition:jovan_post_gst_view}
A view $v$ is a \emph{post-GST view} if and only if $\tau(v) \geq \text{GST}$. 
\end{definition}
Without loss of generality, we assume that all correct validators start \sysname\ before GST. Therefore, any post-GST view $v$ must satisfy $v > 1$ (because correct validators start at view $1$ implying $\tau(1) < \text{GST}$).

Let $\viewduration$ denote the local time each correct validator waits before issuing a timeout message for a given view. Stating differently, $\viewduration$ represents the maximum time a correct leader needs to successfully collect and disseminate a QC in any post-GST view. 
In \sysname, we require $\viewduration$ to be
\[
 \viewduration= 3\Delta + \recoveryduration + 3\Delta,
\] 
where $\Delta$ is the upper bound on message delays after GST and $\recoveryduration$ denotes the time (after GST) required for a correct leader to construct or recover a proposal once it has obtained a QC or a TC from the preceding view (after GST).
%
%
Intuitively, the first $3\Delta$ term accounts for the delay that may occur before a correct leader enters the view --- this can take up to $3\Delta$ after the first correct validator has entered (as proven by \Cref{lemma:jovan_enter_views_with_correct_leader}). 
Additionally, the leader may require up to $\recoveryduration$ time to reconstruct the block it must repropose or to generate an NEC, in the case of the unhappy path. 
Concretely, $\recoveryduration = \left( \left\lceil \frac{n}{\batch} \right\rceil - 1 \right) \cdot \intervalduration + 2\Delta$, where $\batch$ and $\intervalduration$ are parameters of the recovery algorithm (see \Cref{Algorithm:Recovery}): proposal requests are disseminated to $\batch$ validators at a time, with $\intervalduration$ time between successive batches. 
Therefore, $\left( \left\lceil \frac{n}{\batch} \right\rceil - 1 \right) \cdot \intervalduration$ is sufficient to send proposal and NE requests to all validators. 
These requests are received by all validators within an additional $\Delta$ time, and their responses are received by the leader within another $\Delta$, for a total of $2\Delta$ communication delay. 
Thus, the total time required for this recovery process is $\recoveryduration = \left( \left\lceil \frac{n}{\batch} \right\rceil - 1 \right) \cdot \intervalduration + 2\Delta$.

Recall that when the leader obtains a QC, it is assumed to simultaneously obtain the proposal referenced by that QC, allowing it to issue a new proposal immediately upon obtaining the QC.
Finally, the leader needs an additional $\Delta$ time to disseminate its proposal, $\Delta$ time to collect votes and form a QC, and another $\Delta$ time to disseminate the newly formed QC.

The following lemma establishes that no correct validator issues a timeout message associated with any post-GST view $v$ before time $\tau(v) + \viewduration$.

\begin{lemma} \label{lemma:jovan_first_timeout_time}
Let $v > 1$ be any post-GST view.
Then, no correct validator issues a timeout message for view $v$ before time $\tau(v) + \viewduration$.
\end{lemma}
\begin{proof}
By \Cref{lemma:jovan_first_timeout}, the first correct validator to issue a timeout message for view $v$ must do so only after fully timing out from that view.  
Since no correct validator enters view $v$ before time $\tau(v)$, this event cannot occur earlier than $\tau(v) + \viewduration$.
\qed
\end{proof}

We now prove that no TC can be formed for any post-GST view $v$ before time $\tau(v) + \viewduration$.

\begin{lemma} \label{lemma:jovan_tc_time}
Let $v > 1$ be any post-GST view.
Then, no TC for view $v$ can be formed before time $\tau(v) + \viewduration$.
\end{lemma}
\begin{proof}
For a TC to be formed, at least $f + 1$ timeout messages must be issued by correct validators. 
Thus, the lemma follows directly from \Cref{lemma:jovan_first_timeout_time}.
\qed
\end{proof}

The following lemma establishes that, for any post-GST view $v$, no correct validator can enter any view $v' > v$ before time $\tau(v) + \viewduration$, unless at least one correct validator accepts a QC for view $v$ prior to $\tau(v) + \viewduration$.

\begin{lemma} \label{lemma:jovan_early_qc}
Let $v > 1$ be any post-GST view.  
If there exists a view $v' > v$ such that $\tau(v') < \tau(v) + \viewduration$, then there must exist a correct validator that accepts a QC for view $v$ before time $\tau(v) + \viewduration$.
\end{lemma}
\begin{proof}
We prove this by contradiction.  
Suppose there exists a view $v' > v$ that is entered by a correct validator before time $\tau(v) + \viewduration$, and that no correct validator accepts a QC for view $v$ before $\tau(v) + \viewduration$.  
By \Cref{lemma:jovan_enter_before}, it follows that $\tau(v + 1) < \tau(v) + \viewduration$.  
Let $i$ be a correct validator that enters view $v + 1$ at time $\tau(v + 1)$.  
We consider two possibilities:
\begin{itemize}
    \item Suppose that validator $i$ enters view $v + 1$ upon accepting a QC for view $v$.  
    This contradicts the assumption that no correct validator accepts a QC for view $v$ before $\tau(v) + \viewduration$, rendering this case impossible.

    \item Suppose that validator $i$ enters view $v + 1$ upon accepting a TC for view $v$.
    This case is impossible due to \Cref{lemma:jovan_tc_time}.
\end{itemize}
As neither case is possible, the contradiction completes the proof.
\qed
\end{proof}


The following lemma establishes that every correct validator enters a post-GST view $v$ with a correct leader within at most $3\delta$ time from the moment the first correct validator entered view $v$, or that at least one correct validator has accepted a QC for view $v$ by that time.  
Recall that $\delta \leq \Delta$ denotes the actual network delay.

\begin{lemma} \label{lemma:jovan_enter_views_with_correct_leader}
%
Let $v > 1$ be any post-GST view with a correct leader. 
If no correct validator accepts a QC for view $v$ by time $\tau(v) + 3\delta$, then every correct validator enters view $v$ by time $\tau(v) + 3\delta$.
\end{lemma}
\begin{proof}
We first prove that if no correct validator accepts a QC for view $v$ by time $\tau(v) + 3\delta$, then no correct validator can enter any view $\geq v + 1$ by that time.
The reasoning is as follows.  
By \Cref{lemma:jovan_tc_time}, no TC can be formed for view $v$ before time $\tau(v) + \viewduration > \tau(v) + 3\Delta \geq \tau(v) + 3\delta$.  
Hence, no correct validator can enter view $v + 1$ by time $\tau(v) + 3\delta$, since a validator can enter a new view only upon accepting either a QC or a TC associated with view $v$.  
Consequently, by \Cref{lemma:jovan_enter_before}, no correct validator enters a view $\geq v + 1$ by time $\tau(v) + 3\delta$.

Since a correct validator enters view $v$ at time $\tau(v)$, it must have accepted either a QC or a TC associated with view $v - 1$ at that time.  
To establish the lemma, we therefore consider two possible cases:
\begin{itemize}[itemsep=3pt]
    \item \emph{Case 1:} Suppose that some correct validator $i$ accepts a QC for view $v - 1$ by time $\tau(v) + \delta$ at line~\ref{line:increment_view_proposal} of \Cref{Algorithm:Consensus-Execution_1} or at line~\ref{line:increment_view_form_qc} of \Cref{Algorithm:Consensus-Execution_1} or at line~\ref{line:increment_view_receive_qc} of \Cref{Algorithm:Consensus-Execution_1} or at line~\ref{line:increment_view_receive_qc_next_leader} of \Cref{Algorithm:Consensus-Execution_1} or at line~\ref{line:increment_view_receive_qc_timeout_message} of \Cref{Algorithm:Consensus-Execution_2}.
    This accounts for all situations in which a correct validator may accept a QC, except for the case where it is formed from $2f + 1$ timeout messages (line~\ref{line:increment_view_form_qc_timeout_message} of \Cref{Algorithm:Consensus-Execution_2}), which will be addressed separately in the second case.

    We now differentiate two scenarios:
    \begin{itemize}
        \item Let validator $i$ be the leader of view $v$.
        Note that validator $i$ cannot accept a QC for view $v - 1$ at line~\ref{line:increment_view_proposal} of \Cref{Algorithm:Consensus-Execution_1}, since this would require the leader to have received a proposal from itself.  
        However, issuing a proposal is only possible after the validator has already entered view $v$ (see \Cref{Algorithm:Consensus-Execution_1} and \Cref{Algorithm:Consensus-Execution_2}).

        In all remaining cases, the leader issues its proposal immediately upon accepting a QC (line~\ref{line:propose_form_qc} or line~\ref{line:propose_receive_qc} of \Cref{Algorithm:Consensus-Execution_1}, or line~\ref{line:propose_receive_qc_timeout_message} of \Cref{Algorithm:Consensus-Execution_2}).  
        Consequently, all correct validators enter view $v$ at line~\ref{line:increment_view_proposal} of \Cref{Algorithm:Consensus-Execution_1} by time $\tau(v) + 2\delta$, unless they have already done so, and thus the lemma follows in this case.

\item Let validator $i$ be a non-leader in view $v$.
In this case, we first show that the leader of view $v$ issues a proposal (or a proposal request or an NE request) by time $\tau(v) + 2\delta$.  
If validator $i$ accepts a QC at line~\ref{line:increment_view_proposal} of \Cref{Algorithm:Consensus-Execution_1}, then the corresponding proposal must have been received from the leader of view $v$, implying that the leader issued its proposal by time $\tau(v) + \delta < \tau(v) + 2\delta$.  
If validator $i$ accepts a QC at any other point in the protocol (covered by this case), it forwards that QC to the leader of view $v$ (line~\ref{line:disseminate_backup_qc} or line~\ref{line:handle_backup_qc_receive_qc} of \Cref{Algorithm:Consensus-Execution_1} or line~\ref{line:forward_qc_timeout_message} of \Cref{Algorithm:Consensus-Execution_2}).  
Hence, the leader of view $v$ receives a QC associated with view $v - 1$ by time $\tau(v) + 2\delta$.  
If the leader has not yet entered view $v$ and issued a proposal (or a proposal request or an NE request) by this time, it accepts the QC (line~\ref{line:increment_view_receive_qc_next_leader} of \Cref{Algorithm:Consensus-Execution_1}) and then issues its proposal (line~\ref{line:propose_receive_qc} of \Cref{Algorithm:Consensus-Execution_1}).

Therefore, every correct validator receives a proposal (or a proposal request or an NE request) from the leader of view $v$ by time $\tau(v) + 3\delta$.  
Consequently, if not already done so, each correct validator enters view $v$ by time $\tau(v) + 3\delta$---either at line~\ref{line:increment_view_proposal} of \Cref{Algorithm:Consensus-Execution_1}, or at line~\ref{line:increment_view_proposal_request} or line~\ref{line:increment_view_ne_request} of \Cref{Algorithm:Consensus-Execution_2}.  

\end{itemize}

    \item \emph{Case 2:} Suppose that no correct validator $i$ accepts a QC for view $v - 1$ by time $\tau(v) + \delta$ at line~\ref{line:increment_view_proposal} of \Cref{Algorithm:Consensus-Execution_1} or at line~\ref{line:increment_view_form_qc} of \Cref{Algorithm:Consensus-Execution_1} or at line~\ref{line:increment_view_receive_qc} of \Cref{Algorithm:Consensus-Execution_1} or at line~\ref{line:increment_view_receive_qc_next_leader} of \Cref{Algorithm:Consensus-Execution_1} or at line~\ref{line:increment_view_receive_qc_timeout_message} of \Cref{Algorithm:Consensus-Execution_2}.
    In this case, any correct validator that entered view $v$ at time $\tau(v)$ must have done so either by accepting a TC for view $v - 1$ or by accepting a QC for view $v - 1$ at line~\ref{line:increment_view_form_qc_timeout_message} of \Cref{Algorithm:Consensus-Execution_2} upon receiving $2f + 1$ timeout messages associated with view $v - 1$.
    We again consider two possible scenarios:
    \begin{itemize}
        \item Suppose that the leader of view $v$ issues its proposal (or proposal request or NE request) by time $\tau(v) + \delta$.
        Then, every correct validator receives this message by time $\tau(v) + 2\delta$, ensuring that all correct validators enter view $v$ by that time (at line~\ref{line:increment_view_proposal} of \Cref{Algorithm:Consensus-Execution_1} or at line~\ref{line:increment_view_proposal_request} of \Cref{Algorithm:Consensus-Execution_2} or at line~\ref{line:increment_view_ne_request} of \Cref{Algorithm:Consensus-Execution_2}).

        \item Suppose that the leader of view $v$ does not issue its proposal (or proposal request or NE request) by time $\tau(v) + \delta$.
        In this case, we show that every correct validator either issues a timeout message associated with view $v - 1$ by time $\tau(v) + \delta$ or broadcasts a TC by the same time.  
        Since at least $f + 1$ timeout messages are sent by correct validators by time $\tau(v)$, every correct validator receives $f + 1$ timeout messages by time $\tau(v) + \delta$.  
        Upon receiving the first $f + 1$ timeout messages for view $v - 1$, one of the following must occur for each correct validator $i$:
        \begin{itemize}
            \item Validator $i$ sends a timeout message for view $v - 1$ after fully or prematurely timing out from view $v - 1$ by this moment.

            \item Otherwise, validator $i$ has already entered view $v$ (as the check at line~\ref{line:receive_timeout_message} of \Cref{Algorithm:Consensus-Execution_2} fails).  
            Recall that no validator can enter any view $\geq v + 1$ by this time, as argued in the first paragraph of the proof.  
            Since no correct validator accepts a QC for view $v - 1$ at line~\ref{line:increment_view_proposal} of \Cref{Algorithm:Consensus-Execution_1} or at line~\ref{line:increment_view_form_qc} of \Cref{Algorithm:Consensus-Execution_1} or at line~\ref{line:increment_view_receive_qc} of \Cref{Algorithm:Consensus-Execution_1} or at line~\ref{line:increment_view_receive_qc_next_leader} of \Cref{Algorithm:Consensus-Execution_1} or at line~\ref{line:increment_view_receive_qc_timeout_message} of \Cref{Algorithm:Consensus-Execution_2} nor does the leader issue any proposal (or proposal request or NE request) by time $\tau(v) + \delta$, validator $i$ must have entered view $v$ upon accepting a TC at line~\ref{line:increment_view_receive_qc_timeout_message} or line~\ref{line:increment_view_receive_tc} of \Cref{Algorithm:Consensus-Execution_2}.  
            (Note that if validator $i$ moved to view $v$ upon accepting a QC or a TC at line~\ref{line:increment_view_form_qc_timeout_message} or line~\ref{line:increment_view_form_tc} of \Cref{Algorithm:Consensus-Execution_2}, then it must have already issued a timeout message for view $v - 1$---a situation already covered above---since it must first time out from view $v - 1$ before transitioning to view $v$.)
            In this case, if validator $i$ has not already issued a timeout message for view $v - 1$, it broadcasts the accepted TC (line~\ref{line:broadcast_tc_timeout} or line~\ref{line:broadcast_tc_receive_tc} of \Cref{Algorithm:Consensus-Execution_2}).
        \end{itemize}
        Hence, every correct validator either issues a timeout message associated with view $v - 1$ by time $\tau(v) + \delta$ or broadcasts a TC associated with view $v - 1$ by that time.
        Finally, we distinguish two cases:
\begin{itemize}
    \item Suppose that some correct validator broadcasts a TC by time $\tau(v) + \delta$.  
    Then every correct validator receives a TC associated with view $v - 1$ by time $\tau(v) + 2\delta$.  
    In this case, unless already in view $v$, each correct validator accepts the TC and enters view $v$ (line~\ref{line:increment_view_receive_tc} of \Cref{Algorithm:Consensus-Execution_2}).

    \item Otherwise, all correct validators receive $2f + 1$ timeout messages by time $\tau(v) + 2\delta$. 
    If a correct validator has not yet entered view $v$ by this time, it does so at line~\ref{line:increment_view_form_qc_timeout_message} or line~\ref{line:increment_view_form_tc} of \Cref{Algorithm:Consensus-Execution_2}, upon constructing a TC or a QC associated with view $v - 1$ (line~\ref{line:form_tc_qc} of \Cref{Algorithm:Consensus-Execution_2}).
\end{itemize}

    \end{itemize}

\end{itemize}
Since both cases satisfy the lemma’s claim, the proof is complete.
\qed
\end{proof}

Next, we show that for any post-GST view $v$ with a correct leader, the leader enters view $v$ by time $\tau(v) + 3\delta$.

\begin{lemma} \label{lemma:jovan_correct_enters}
Let $v > 1$ be any post-GST view with a correct leader.
Then, the leader of view $v$ enters view $v$ by time $\tau(v) + 3\delta$.
\end{lemma}
\begin{proof}
Let us consider two possible cases: 
\begin{itemize}
    \item \emph{Case 1:} No correct validator accepts a QC for view $v$ by time $\tau(v) + 3\delta$.
    
    In this case, by \Cref{lemma:jovan_enter_views_with_correct_leader}, every correct validator (including the leader) enters view $v$ by time $\tau(v) + 3\delta$.

    \item \emph{Case 2:} Some correct validator accepts a QC for view $v$ by time $\tau(v) + 3\delta$. 
    
    Note that no correct validator issues a timeout message for view $v$ before time $\tau(v) + \viewduration > \tau(v) + 3\delta$ (by \Cref{lemma:jovan_first_timeout_time}). 
    Therefore, every correct validator that contributed to the aforementioned QC must have done so upon receiving the leader's proposal for view $v$.  
    Consequently, the leader of view $v$ must have entered view $v$ prior to that time. 
    \qed
\end{itemize}
\end{proof}

We now show that, in every post-GST view $v$ with a correct leader, the leader issues a proposal by time $\tau(v) + 3\delta + \recoveryduration + \delta$.

\begin{lemma} \label{lemma:jovan_correct_leader_proposal_time}
Let $v > 1$ be any post-GST view with a correct leader.  
If the leader of view $v$ enters view $v$ upon accepting a QC associated with view $v - 1$, then the leader issues its proposal by time
\[
\tau(v) + 3\delta.
\]
Alternatively, if the leader of view $v$ enters view $v$ upon accepting a TC associated with view $v - 1$, then the leader issues its proposal by time
\[
\tau(v) + 3\delta + \recoveryduration.
\]
\end{lemma}
\begin{proof}
By \Cref{lemma:jovan_correct_enters}, the leader $L$ of view $v$ enters view $v$ by time $\tau(v) + 3\delta$.  
If the leader $L$ enters view $v$ upon accepting a QC from view $v - 1$, then $L$ issues its proposal immediately. 
This directly establishes the claim of the lemma in this case. 
(Recall that we assume the leader, upon obtaining a QC, simultaneously receives the proposal referenced by that QC. 
This enables the leader to issue a new proposal without delay.)

If the leader $L$ enters view $v$ upon accepting a timeout certificate $\mathit{tc}$ from view $v - 1$, and either (i) $\mathit{tc}.\highQC \neq \bot$, or (ii) the leader already possesses a block corresponding to $\mathit{tc}.\hightip$, then $L$ issues its proposal immediately (see the $\NextLeaderReceivesTC$ function).

Finally, it remains to consider the case where the leader does not possess a block corresponding to $\mathit{tc}.\hightip$.  
Note that in this case, every correct validator receives both an NE request and a proposal request from the leader $L$ by time $\tau(v) + 3\delta + \recoveryduration - \Delta$ (see \Cref{Algorithm:Recovery}).
In this situation, we distinguish between two subcases:
\begin{itemize}
    \item \emph{Case 1:} Some correct validator enters some view $> v$ by time $\tau(v) + 3\delta + \recoveryduration - \Delta$.

    By \Cref{lemma:jovan_enter_before}, we have 
\[
\tau(v + 1) < \tau(v) + 3\delta + \recoveryduration - \Delta < \tau(v) + \viewduration.
\]
Therefore, some correct validator must have accepted either a QC or a TC for view $v$ by time $\tau(v) + 3\delta + \recoveryduration - \Delta$.  
However, by \Cref{lemma:jovan_tc_time}, no TC for view $v$ can be formed before time $\tau(v) + \viewduration > \tau(v) + 3\delta + \recoveryduration - \Delta$.  
It follows that some correct validator must have accepted a QC for view $v$ by time $\tau(v) + 3\delta + \recoveryduration - \Delta$.
    Since no correct validator issues a timeout message before time $\tau(v) + \viewduration > \tau(v) + 3\delta + \recoveryduration - \Delta$ (by \Cref{lemma:jovan_first_timeout_time}), it follows that every correct validator that contributed to the aforementioned QC did so upon receiving a proposal from the leader, which concludes the proof.

    \item \emph{Case 2:} No correct validator enters any view $> v$ by time $\tau(v) + 3\delta + \recoveryduration - \Delta$.  

    In this case, every correct validator sends either the block corresponding to $\mathit{tc}.\hightip$ or an $\mathsf{NE}$ message back to the leader by time $\tau(v) + 3\delta + \recoveryduration - \Delta$.  
    (This is because, upon receiving an NE request or a proposal request, the validator cannot have entered any view $> v$, as these messages are received by time $\tau(v) + 3\delta + \recoveryduration - \Delta$.)
    Hence, by time $\tau(v) + 3\delta + \recoveryduration$, the leader $L$ either obtains the block corresponding to $\mathit{tc}.\hightip$ (if at least one correct validator possesses the block) or constructs an NEC (if no correct validator possesses the block).  
    Consequently, by the same time, the leader $L$ issues its proposal.
    \qed
\end{itemize}
\end{proof}

Next, we show that for any post-GST view $v$ with a correct leader, if no correct validator accepts a QC for view $v$ by time $\tau(v) + 3\delta + \recoveryduration + \delta$, then all correct validators cast their votes for a proposal associated with view $v$ by that time.

\begin{lemma} \label{lemma:jovan_issue_vote_time_tc}
Let $v > 1$ be any post-GST view with a correct leader, and let $p$ be the proposal issued by that leader in view $v$. 
If $p.\tc = \bot$ (resp., $p.\tc \neq \bot$) and no correct validator accepts a QC for view $v$ by time $\tau(v) + 3\delta + \delta$ (resp., by time $\tau(v) + 3\delta + \recoveryduration + \delta$), then every correct validator issues its vote for proposal $p$  
by time $\tau(v) + 3\delta + \delta$ (resp., by time $\tau(v) + 3\delta + \recoveryduration + \delta$).
\end{lemma}
\begin{proof}
Note that, by~\Cref{lemma:jovan_tc_time}, no TC for view $v$ can be formed by time $\tau(v) + 3\delta + \recoveryduration + \delta$.  
Since a validator can enter a new view only upon accepting either a QC or a TC, this fact, together with the lemma’s assumption, implies that no correct validator enters view $v + 1$ or any higher view (as proven by \Cref{lemma:jovan_enter_before}) by time $\tau(v) + 3\delta + \delta$ if $p.\tc = \bot$ and by time $\tau(v) + 3\delta + \recoveryduration + \delta$ if $p.\tc \neq \bot$.

Let us now differentiate two cases:
\begin{itemize}
    \item Let $p.\tc = \bot$.
    By \Cref{lemma:jovan_correct_leader_proposal_time}, the leader issues its proposal $p$ by time $\tau(v) + 3\delta$.  
    This proposal is received by every correct validator by time $\tau(v) + 3\delta + \delta$, upon which they cast their votes for proposal $p$, entering view $v$ first if they have not already done so.  
    From the first paragraph, we know that no correct validator enters any view beyond $v$ before this time.  
    Hence, whenever a correct validator receives proposal $p$, it enters view $v$ (unless it has already done so) and subsequently votes for it.
    Thus, the lemma follows in this case.

    \item Let $p.\tc \neq \bot$.
    By \Cref{lemma:jovan_correct_leader_proposal_time}, the leader issues its proposal $p$ by time $\tau(v) + 3\delta + \recoveryduration$.  
    This proposal is received by every correct validator by time $\tau(v) + 3\delta + \recoveryduration + \delta$, upon which they cast their votes for proposal $p$, entering view $v$ first if they have not already done so.  
    From the first paragraph, we know that no correct validator enters any view beyond $v$ before this time.  
    Hence, whenever a correct validator receives proposal $p$, it enters view $v$ (unless it has already done so) and subsequently votes for it.
    Thus, the lemma follows.
\qed
\end{itemize}
\end{proof}


The following lemma proves that when leaders of two consecutive post-GST views are honest, then the leader of the latter view receives a QC within a bounded amount of time.

\begin{lemma} \label{lemma:jovan_two_leaders}
Let $v > 1$ and $v + 1$ be two consecutive post-GST views, each led by a correct leader.  
Suppose that the leader of view $v$ enters view $v$ upon accepting a QC associated with view $v - 1$.
Then, the leader of view $v + 1$ enters view $v + 1$ at some time $\tau_L$ upon accepting a QC for view $v$.  
Moreover, the following bounds hold:
\[
\tau_L \leq \tau(v) + 3\delta + 2\delta,
\]
and moreover,
\[
\tau_L \leq \tau(v + 1) + \delta.
\]
\end{lemma}
\begin{proof}
Let $L_v$ denote the leader of view $v$ and $L_{v + 1}$ denote the leader of view $v + 1$.  
Recall that $\viewduration > 3\delta + 2\delta$.
By \Cref{lemma:jovan_tc_time}, a TC for view $v$ cannot be formed before $\tau(v) + \viewduration > \tau(v) + 3\delta + 2\delta$.

We distinguish two cases:
\begin{itemize}
    \item \emph{Case 1:} Suppose some correct validator accepts a QC for view $v$ before time $\tau(v) + 3\delta + \delta$.
    
    Let $i$ denote the first correct validator to accept a QC for view $v$, occurring at some time $\tau_i^{\mathit{qc}} < \tau(v) + 3\delta + \delta$.  
    Observe that $\tau(v + 1) = \tau_i^{\mathit{qc}}$, as no TC for view $v$ can be formed by this time and $i$ is the first correct validator to accept a QC for view $v$.
    Moreover, that validator $i$ can accept this QC in one of the following ways:
    (i) upon receiving a proposal that carries this QC (line~\ref{line:increment_view_proposal} of \Cref{Algorithm:Consensus-Execution_1}),
    (ii) upon forming the QC locally (line~\ref{line:increment_view_form_qc} of \Cref{Algorithm:Consensus-Execution_1}),
    (iii) upon receiving a QC directly (line~\ref{line:increment_view_receive_qc} or line~\ref{line:increment_view_receive_qc_next_leader} of \Cref{Algorithm:Consensus-Execution_1}), or
    (iv) upon receiving a timeout message containing this QC (line~\ref{line:increment_view_receive_qc_timeout_message} of \Cref{Algorithm:Consensus-Execution_2}).
    Validator $i$ cannot, however, accept this QC at line~\ref{line:increment_view_form_qc_timeout_message} of \Cref{Algorithm:Consensus-Execution_2}, as doing so would require that at least one correct validator had already sent a timeout message.
    This is impossible, since no correct validator issues a timeout message before time $\tau(v) + \viewduration > \tau(v) + 3\delta + \delta$ (by \Cref{lemma:jovan_first_timeout_time}).

    If validator $i$ is the leader of view $v + 1$, the lemma follows immediately.  
    Otherwise, suppose that $i$ is not the leader of view $v + 1$.

    We now show that the leader $L_{v + 1}$ of view $v + 1$ receives or accepts a QC associated with view $v$ by time $\tau_i^{\mathit{qc}} + \delta$.  
Consider all possible points at which validator $i$ could have accepted a QC for view $v$:
\begin{itemize}
    \item If $i$ accepts the QC at line~\ref{line:increment_view_proposal} of \Cref{Algorithm:Consensus-Execution_1}, then it must have received a proposal carrying this QC from the leader $L_{v + 1}$, which implies that the leader had already accepted the QC beforehand.

    \item If $i$ accepts the QC at line~\ref{line:increment_view_form_qc} of \Cref{Algorithm:Consensus-Execution_1}, it forwards the QC to all validators, including $L_{v + 1}$ (line~\ref{line:disseminate_backup_qc} of \Cref{Algorithm:Consensus-Execution_1}), ensuring that the leader receives a QC for view $v$ by time $\tau_i^{\mathit{qc}} + \delta$.  
    Recall that $i \neq L_{v + 1}$.

    \item If $i$ accepts the QC at line~\ref{line:increment_view_receive_qc} of \Cref{Algorithm:Consensus-Execution_1}, it forwards the QC to the leader $L_{v + 1}$ (line~\ref{line:handle_backup_qc_receive_qc} of \Cref{Algorithm:Consensus-Execution_1}), which again ensures that the leader receives a QC for view $v$ by time $\tau_i^{\mathit{qc}} + \delta$.

    \item If $i$ accepts the QC at line~\ref{line:increment_view_receive_qc_timeout_message} of \Cref{Algorithm:Consensus-Execution_2}, it forwards the QC to $L_{v + 1}$ (line~\ref{line:forward_qc_timeout_message} of \Cref{Algorithm:Consensus-Execution_2}), and thus the same bound applies.
\end{itemize}

Therefore, in all cases, the leader $L_{v + 1}$ either receives or accepts a QC associated with view $v$ by time $\tau_i^{\mathit{qc}} + \delta$.  
If it accepts such a QC by this time, the lemma follows.  
Therefore, suppose that $L_{v + 1}$ does not accept any QC for view $v$ by time $\tau_i^{\mathit{qc}} + \delta$.
In this case, it still receives the QC for view $v$ by that time.



Let $\tau_{L_{v + 1}}^{\mathit{qc}} \leq \tau_i^{\mathit{qc}} + \delta \leq \tau(v + 1) + \delta < \tau(v) + 3\delta + 2\delta$ denote the first time $L_{v + 1}$ receives a QC from view $v$.  
We now distinguish two possible cases concerning the time $\tau_{L_{v + 1}}^{\mathit{qc}}$:
\begin{itemize}
    \item Suppose the current view of $L_{v + 1}$ at time $\tau_{L_{v + 1}}^{\mathit{qc}}$ is $v + 1$.  
    Since no TC for view $v$ can be formed before time $\tau(v) + 3\delta + 2\delta > \tau_{L_{v + 1}}^{\mathit{qc}}$, it would imply that $L_{v + 1}$ entered view $v + 1$ upon accepting a QC from view $v$.  
    However, this is impossible, as $L_{v + 1}$ does not accept any QC before time $\tau_{L_{v + 1}}^{\mathit{qc}}$.
    

    \item Suppose that the current view of $L_{v + 1}$ at time $\tau_{L_{v + 1}}^{\mathit{qc}}$ is greater than $v + 1$.  
    Hence, \Cref{lemma:jovan_enter_before} proves that $\tau(v + 2) \leq \tau_{L_{v + 1}}^{\mathit{qc}} < \tau(v) + 3\delta + 2\delta$.
    Since no TC can be formed for view $v + 1$ before time $\tau(v + 1) + \viewduration \geq \tau(v) + \viewduration > \tau_{L_{v + 1}}^{\mathit{qc}}$, any correct validator that has entered view $v + 2$ must have done so due to a QC from view $v + 1$.  
    Moreover, as no correct validator issues a timeout message for view $v + 1$ before $\tau(v + 1) + \viewduration > \tau_{L_{v + 1}}^{\mathit{qc}}$, every correct validator contributing to the aforementioned QC from view $v + 1$ must have done so after receiving a proposal from $L_{v + 1}$.  
    This implies that $L_{v + 1}$ must have entered view $v + 1$ before $\tau_{L_{v + 1}}^{\mathit{qc}}$, since this QC for view $v + 1$ can only be formed after $L_{v + 1}$ has issued its proposal.  
    Furthermore, $L_{v + 1}$ must have entered view $v + 1$ upon accepting a QC for view $v$ (as no TC associated with view $v$ can be formed by this time), which contradicts the fact that $L_{v + 1}$ does not accept any QC associated with view $v$ by time $\tau_{L_{v + 1}}^{\mathit{qc}}$.
    
\end{itemize}

Since both cases are impossible, the current view of $L_{v + 1}$ upon receiving the first QC from view $v$ at time $\tau_{L_{v + 1}}^{\mathit{qc}} \leq \tau(v + 1) + \delta$ must be smaller than $v + 1$.
Hence, the leader $L_{v + 1}$ accepts a QC at this time (line~\ref{line:increment_view_receive_qc_next_leader} of \Cref{Algorithm:Consensus-Execution_1}).


    \item \emph{Case 2:} Suppose no correct validator accepts a QC for view $v$ before time $\tau(v) + 3\delta + \delta$.  

    Given that no TC can be formed before time $\tau(v) + \viewduration > \tau(v) + 3\delta + \delta$, we have that $\tau(v + 1) \geq \tau(v) + 3\delta + \delta$.
    Moreover, by \Cref{lemma:jovan_issue_vote_time_tc}, all correct validators vote for the proposal issued by the leader of view $v$ by that time.  
    Consequently, the leader of view $v + 1$ receives $2f + 1$ votes from correct validators by time $\tau(v) + 3\delta + 2\delta$.
    Let $\tau_{L_{v + 1}}^{\mathit{qc}} \leq \tau(v) + 3\delta + 2\delta$ denote the earliest time at which the leader $L_{v + 1}$ receives $2f + 1$ votes for view $v$ from correct validators.  
    If validator $i$ accepts a QC for view $v$ before time $\tau_{L_{v + 1}}^{\mathit{qc}}$, the lemma follows immediately.  
    Otherwise, assume that validator $i$ does not accept any QC for view $v$ before time $\tau_{L_{v + 1}}^{\mathit{qc}}$.
    Again, we consider two possibilities:
    \begin{itemize}
    \item Suppose the current view of $L_{v + 1}$ at time $\tau_{L_{v + 1}}^{\mathit{qc}}$ is $v + 1$.  
    Since no TC for view $v$ can be formed by time $\tau(v) + 3\delta + 2\delta \geq \tau_{L_{v + 1}}^{\mathit{qc}}$, it follows that $L_{v + 1}$ would have accepted a QC for view $v$ by that time, which contradicts the above assumption.

    \item Suppose that the current view of $L_{v + 1}$ at time $\tau_{L_{v + 1}}^{\mathit{qc}}$ is greater than $v + 1$.  
    Hence, \Cref{lemma:jovan_enter_before} proves that $\tau(v + 2) \leq \tau_{L_{v + 1}}^{\mathit{qc}} < \tau(v) + 3\delta + 2\delta$.
    Since no TC can be formed for view $v + 1$ before time $\tau(v + 1) + \viewduration \geq \tau(v) + \viewduration > \tau_{L_{v + 1}}^{\mathit{qc}}$, any correct validator that has entered view $v + 2$ must have done so due to a QC from view $v + 1$.  
    Moreover, as no correct validator issues a timeout message for view $v + 1$ before $\tau(v + 1) + \viewduration > \tau_{L_{v + 1}}^{\mathit{qc}}$, every correct validator contributing to the aforementioned QC from view $v + 1$ must have done so after receiving a proposal from $L_{v + 1}$.  
    This implies that $L_{v + 1}$ must have entered view $v + 1$ before $\tau_{L_{v + 1}}^{\mathit{qc}}$, since this QC for view $v + 1$ can only be formed after $L_{v + 1}$ has issued its proposal.  
    Furthermore, $L_{v + 1}$ must have entered view $v + 1$ upon accepting a QC for view $v$ (as no TC associated with view $v$ can be formed by this time), which contradicts the fact that $L_{v + 1}$ does not accept any QC associated with view $v$ by time $\tau_{L_{v + 1}}^{\mathit{qc}}$.
\end{itemize}
Since both cases are impossible, the current view of $L_{v + 1}$ at time $\tau_{L_{v + 1}}^{\mathit{qc}}$ must be smaller than $v + 1$.  
Therefore, at this moment, the leader $L_{v + 1}$ forms a QC for view $v$ (line~\ref{line:form_qc} of \Cref{Algorithm:Consensus-Execution_1}) and enters view $v + 1$ upon accepting this QC (line~\ref{line:increment_view_form_qc} of \Cref{Algorithm:Consensus-Execution_1}).
\end{itemize}
Since the statement holds in both cases, the proof is complete.
\qed
\end{proof}

Next, we show that in any sequence of $3f + 3$ consecutive views, there exists at least one group of three consecutive views in which each view has a correct leader.  
This is crucial, as the commitment rule of the \sysname protocol requires two consecutive QCs to commit a block.  
Hence, the existence of three consecutive views with correct leaders guarantees that these two consecutive QCs are both constructed (during the first two views) and subsequently disseminated (during the third view).

\begin{lemma} \label{lemma:jovan_consecutive_views}
In any sequence of $3f + 3$ consecutive views, there exists a group of three consecutive views such that each view has a correct leader.
\end{lemma}
\begin{proof}
Consider any sequence of $3f + 3$ consecutive views.  
This sequence contains $3f + 3 - 2 = 3f + 1$ distinct groups of three consecutive views.  
By the round-robin schedule, each of the $3f + 1$ validators appears in exactly three of these groups.  
Since there are at most $f$ faulty validators, the pigeonhole principle implies that at most $3f$ groups can include a view led by a faulty validator.
Therefore, there exists at least one group in which all three views have correct leaders.
\qed
\end{proof}


Next, we prove that no TC can be constructed for any post-GST view with a correct leader.

\begin{lemma} \label{lemma:jovan_correct_leader_no_tc}
Let $v > 1$ be any post-GST view with a correct leader.
Then, no TC can be formed for view $v$.
Furthermore, at least $f + 1$ correct validators do not issue timeout messages associated with view $v$.
\end{lemma}
\begin{proof}
Let us consider the two following cases:
\begin{itemize}
    \item \emph{Case 1:} At least $f + 1$ correct validators enter views greater than $v$ before time $\tau(v) + \viewduration$.
    
    Since no correct validator can issue a timeout message before $\tau(v) + \viewduration$ (by \Cref{lemma:jovan_first_timeout_time}), those validators that advance to higher views before this time will never issue a timeout message for view $v$: they can neither fully nor prematurely time out from view $v$ prior to time $\tau(v) + \viewduration$.
    As at least $2f + 1$ timeout messages are required to form a TC, no TC for view $v$ can be constructed in this case, which proves the lemma.

    \item \emph{Case 2:} At most $f$ correct validators enter views greater than $v$ before time $\tau(v) + \viewduration$. 
    
    In this case, we first prove that no correct validator can enter any view $> v + 1$ before time $\tau(v) + \viewduration$. 
    Since $\tau(v + 1) \geq \tau(v)$ (by \Cref{lemma:jovan_enter_before}) and no TC can be formed for view $v + 1$ before time $\tau(v + 1) + \viewduration \geq \tau(v) + \viewduration$ (by \Cref{lemma:jovan_tc_time}), the only way a correct validator could enter a view $> v + 1$ before time $\tau(v) + \viewduration$ is by accepting a QC from view $v + 1$.  
    Moreover, since no correct validator times out from view $v + 1$ before $\tau(v + 1) + \viewduration \geq \tau(v) + \viewduration$ (by \Cref{lemma:jovan_first_timeout_time}), every correct validator that contributes to any QC from view $v + 1$ formed before time $\tau(v) + \viewduration$ must have done so upon processing the leader’s proposal for view $v + 1$ (line~\ref{line:send_vote_normal} of \Cref{Algorithm:Consensus-Execution_1}).  
    As forming such a QC requires at least $f + 1$ correct validators to have voted, this would imply that $f + 1$ or more correct validators entered views greater than $v$ before time $\tau(v) + \viewduration$ (line~\ref{line:increment_view_proposal} of \Cref{Algorithm:Consensus-Execution_1}), which contradicts the case assumption.
    Therefore, all correct validators that transition to future views before time $\tau(v) + \viewduration$ do so only to view $v + 1$.  

    Let $L$ denote the leader of view $v$.  
    We further distinguish two scenarios:
    \begin{itemize}
        \item Suppose that $L$ disseminates a QC formed for view $v$ to all validators before time $\tau(v) + \viewduration - \Delta$ (by invoking $\CurrentLeaderReceivesQC$ at line~\ref{line:disseminate_backup_qc_proposal} or line~\ref{line:disseminate_backup_qc} or line~\ref{line:disseminate_backup_qc_2} of \Cref{Algorithm:Consensus-Execution_1} or line~\ref{line:disseminate_backup_qc_timeout} of \Cref{Algorithm:Consensus-Execution_2}).

        In this case, every correct validator $i$ receives from the leader $L$ a QC formed for view $v$ by time $\tau(v) + \viewduration$.  
        If the current view of validator $i$ is $v + 1$, then $i$ has already entered view $v + 1$ and therefore never issues a timeout message for view $v$.  
        If the current view of validator $i$ is $< v + 1$, validator $i$ accepts the received QC (line~\ref{line:increment_view_receive_qc} of \Cref{Algorithm:Consensus-Execution_1}) and enters view $v$, which again implies that it will not issue a timeout message for that view.  
        (Recall that the current view of validator $i$ cannot exceed $v + 1$.)  
        Hence, in both cases, validator $i$ never issues a timeout message for view $v$, and therefore no correct validator does so, implying that no TC can be formed for view $v$.

        \item Suppose that $L$ does not disseminate a QC formed for view $v$ to all validators before time $\tau(v) + \viewduration - \Delta$ (by invoking $\CurrentLeaderReceivesQC$ at line~\ref{line:disseminate_backup_qc_proposal} or line~\ref{line:disseminate_backup_qc} or line~\ref{line:disseminate_backup_qc_2} of \Cref{Algorithm:Consensus-Execution_1} or line~\ref{line:disseminate_backup_qc_timeout} of \Cref{Algorithm:Consensus-Execution_2}).

        For this scenario, we analyze two possible cases according to \Cref{lemma:jovan_issue_vote_time_tc}:
        \begin{itemize}
            \item Suppose that some correct validator $i$ accepts a QC for view $v$ by time $\tau(v) + 3\delta + \recoveryduration + \delta$.
            
            Note that this correct validator $i$ cannot accept a QC for view $v$ at line~\ref{line:increment_view_form_qc_timeout_message} of \Cref{Algorithm:Consensus-Execution_2}, since it is impossible to receive $2f + 1$ timeout messages for view $v$ by time $\tau(v) + 3\delta + \recoveryduration + \delta < \tau(v) + \viewduration$ (as no correct validator issues a timeout message associated with view $v$ before time $\tau(v) + \viewduration$ according to \Cref{lemma:jovan_first_timeout_time}).
            We now consider all possible points in the protocol where validator $i$ could accept a QC for view $v$:
            \begin{itemize}
                \item Suppose that validator $i$ accepts a QC for view $v$ at line~\ref{line:increment_view_proposal} of \Cref{Algorithm:Consensus-Execution_1}.  
                In this case, $i$ forwards the accepted QC to the leader $L$ (line~\ref{line:forward_qc_proposal_current_leader} of \Cref{Algorithm:Consensus-Execution_1}).  
                This implies that the leader $L$ receives a QC for view $v$ by time $\tau(v) + 3\delta + \recoveryduration + 2\delta \leq \tau(v) + \viewduration - \Delta$.  
                Upon receiving this QC, if the current view of $L$ is $v + 1$ (recall that it cannot be greater), then $L$ has already disseminated a QC for view $v$ (at line~\ref{line:disseminate_backup_qc_proposal}, line~\ref{line:disseminate_backup_qc}, line~\ref{line:disseminate_backup_qc_2} of \Cref{Algorithm:Consensus-Execution_1}, or line~\ref{line:disseminate_backup_qc_timeout} of \Cref{Algorithm:Consensus-Execution_2}).  
                Otherwise, $L$ processes the received QC and disseminates it at line~\ref{line:disseminate_backup_qc_2} of \Cref{Algorithm:Consensus-Execution_1}.  
                In both cases, we reach a contradiction with the assumption that $L$ does not disseminate a QC formed in view $v$ by time $\tau(v) + \viewduration - \Delta$, making this case impossible.

                \item Suppose that validator $i$ accepts a QC for view $v$ at line~\ref{line:increment_view_form_qc} of \Cref{Algorithm:Consensus-Execution_1}.  
                In this case, $i$ is either the leader of view $v$ or the leader of view $v + 1$.  
                If $i$ is the leader of view $v$ (i.e., $i = L$), then it disseminates the QC formed in view $v$ at line~\ref{line:disseminate_backup_qc} of \Cref{Algorithm:Consensus-Execution_1}, which contradicts the initial assumption that this does not occur.  

                Otherwise, $i$ is the leader of view $v + 1$, in which case it issues its proposal that is received by all correct validators by time $\tau(v) + 3\delta + \recoveryduration + 2\delta < \tau(v) + \viewduration$.  
                Consequently, by time $\tau(v) + 3\delta + \recoveryduration + 2\delta$, all correct validators have entered view $v + 1$, implying that no correct validator issues a timeout message for view $v$.  
                Hence, the lemma follows in this case.

                \item Suppose that validator $i$ accepts a QC for view $v$ at line~\ref{line:increment_view_receive_qc} of \Cref{Algorithm:Consensus-Execution_1}.  
                In this case, $i$ must have received the QC from the leader $L$, which implies that $L$ had previously disseminated the QC formed for view $v$ (at line~\ref{line:disseminate_backup_qc_proposal}, line~\ref{line:disseminate_backup_qc}, line~\ref{line:disseminate_backup_qc_2} of \Cref{Algorithm:Consensus-Execution_1}, or line~\ref{line:disseminate_backup_qc_timeout} of \Cref{Algorithm:Consensus-Execution_2}).  
                This contradicts the initial assumption, and therefore this case cannot occur.

                \item Suppose that validator $i$ accepts a QC for view $v$ at line~\ref{line:increment_view_receive_qc_next_leader} of \Cref{Algorithm:Consensus-Execution_1}.
                In this case, $i$ is the leader of view $v + 1$.
                Therefore, validator $i$ issues its proposal that is received by all correct validators by time $\tau(v) + 3\delta + \recoveryduration + 2\delta < \tau(v) + \viewduration$.  
                Consequently, by time $\tau(v) + 3\delta + \recoveryduration + 2\delta$, all correct validators have entered view $v + 1$, implying that no correct validator issues a timeout message for view $v$.  
                Hence, the lemma follows in this case.

                \item Suppose that validator $i$ accepts a QC for view $v$ at line~\ref{line:increment_view_receive_qc_timeout_message} of \Cref{Algorithm:Consensus-Execution_2}.
                In this case, $i$ forwards the accepted QC to the leader $L$ (line~\ref{line:forward_qc_timeout_message_current_leader} of \Cref{Algorithm:Consensus-Execution_2}).  
                This implies that the leader $L$ receives a QC for view $v$ by time $\tau(v) + 3\delta + \recoveryduration + 2\delta \leq \tau(v) + \viewduration - \Delta$.  
                Upon receiving this QC, if the current view of $L$ is $v + 1$ (recall that it cannot be greater), then $L$ has already disseminated a QC for view $v$ (at line~\ref{line:disseminate_backup_qc_proposal}, line~\ref{line:disseminate_backup_qc}, line~\ref{line:disseminate_backup_qc_2} of \Cref{Algorithm:Consensus-Execution_1}, or line~\ref{line:disseminate_backup_qc_timeout} of \Cref{Algorithm:Consensus-Execution_2}).  
                Otherwise, $L$ processes the received QC and disseminates it at line~\ref{line:disseminate_backup_qc_2} of \Cref{Algorithm:Consensus-Execution_1}.  
                In both cases, we reach a contradiction with the assumption that $L$ does not disseminate a QC formed in view $v$ by time $\tau(v) + \viewduration - \Delta$, making this case impossible.
            \end{itemize}
            Hence, in all possible scenarios, the lemma holds.

            \item Suppose that no correct validator accepts a QC for view $v$ by time $\tau(v) + 3\delta + \recoveryduration + \delta$.

            By \Cref{lemma:jovan_issue_vote_time_tc}, all correct validators issue their votes by time $\tau(v) + 3\delta + \recoveryduration + \delta$, which implies that the leader $L$ receives votes from all correct validators by time $\tau(v) + 3\delta + \recoveryduration + 2\delta \leq \tau(v) + \viewduration - \Delta$.  
            Upon receiving these votes, if the current view of $L$ is $v + 1$ (recall that it cannot be greater), then $L$ has already disseminated a QC for view $v$ (at line~\ref{line:disseminate_backup_qc_proposal}, line~\ref{line:disseminate_backup_qc}, line~\ref{line:disseminate_backup_qc_2} of \Cref{Algorithm:Consensus-Execution_1}, or line~\ref{line:disseminate_backup_qc_timeout} of \Cref{Algorithm:Consensus-Execution_2}).  
            Otherwise, $L$ processes the received votes and disseminates the QC for view $v$ (formed from these votes) at line~\ref{line:disseminate_backup_qc} of \Cref{Algorithm:Consensus-Execution_1}.  
            Hence, in both cases, we obtain a contradiction with the assumption that $L$ does not disseminate a QC formed in view $v$ by time $\tau(v) + \viewduration - \Delta$, making this case impossible.
        \end{itemize}
    \end{itemize}

\end{itemize}
As the statement of the lemma holds in both possible cases, the proof is concluded.
\qed
\end{proof}


Next, we show that no correct validator can (fully or prematurely) time out from a view with a correct leader before every correct validator receives a proposal from that leader.

\begin{lemma} \label{lemma:jovan_timeout_all_proposal}
Let $v > 1$ be any post-GST view with a correct leader.  
Then, no correct validator (fully or prematurely) times out from view $v$ before every correct validator receives a proposal from the leader of view $v$.
\end{lemma} 
\begin{proof}
For the sake of contradiction, suppose that some correct validator $i$ (either fully or prematurely) times out from view $v$ before another correct validator $j$ receives the proposal from the leader of view $v$.  
Observe that validator $i$ cannot (fully or prematurely) time out in view $v$ before time $\tau(v) + \viewduration$ (by \Cref{lemma:jovan_first_timeout_time}).  
By \Cref{lemma:jovan_correct_leader_proposal_time}, the leader of view $v$ issues its proposal by time $\tau(v) + 3\delta + \recoveryduration$.  
Consequently, validator $j$ receives this proposal by time $\tau(v) + 3\delta + \recoveryduration + \delta < \tau(v) + \viewduration$, which contradicts the assumption that validator $i$ times out before validator $j$ receives the leader’s proposal.  
Hence, the claim follows. 
\qed
\end{proof}

The following lemma shows that any QC formed in a post-GST view with a correct leader must correspond to the proposal issued in that same view.

\begin{lemma} \label{lemma:jovan_post_gst_correct_leader_qc}
Let $v > 1$ be any post-GST view with a correct leader, and let $p$ denote the proposal issued by the leader in view $v$.  
If any QC $\mathit{qc}$ is formed in view $v$, then $\mathit{qc}.\blockid = p.\blockid$.
\end{lemma}
\begin{proof}
Assume for contradiction that $\mathit{qc}$ does not point to $p$.  
Note that all correct validators that cast their votes in view $v$ after receiving the leader’s proposal do vote for the received proposal, and therefore must vote for $p$.  
Then, there must exist a correct validator contributing to $\mathit{qc}$ by issuing a timeout message (upon fully or prematurely timing out from view $v$) including a vote for a proposal different from $p$.
However, by \Cref{lemma:jovan_timeout_all_proposal}, that correct validator must have previously received the leader's proposal $p$.
Upon receiving proposal $p$, the correct validator must have voted for $p$ (line~\ref{line:send_vote_normal} of \Cref{Algorithm:Consensus-Execution_1}), since, had its current view been greater than $v$ (which would prevent the aforementioned vote), it would not have issued a timeout message associated with view $v$ .
Consequently, by \Cref{lemma:jovan_same_vote}, the timeout message must include a vote for $p$, which implies that $\mathit{qc}$ points to $p$.
\qed

\end{proof}


The following lemma establishes that if a correct leader issues a proposal $p$ in any post-GST view $v$, then, before the first TC for view $v + 1$ is formed (if any), at least $f + 1$ correct validators must have observed a QC formed in view $v$.  
This, in turn, implies that these validators update their local QC variable ($\localhighQC$) to this QC before the TC for view $v + 1$ is constructed.

\begin{lemma} \label{lemma:jovan_post_gst_correct_leader_tc}
Let $v > 1$ be any post-GST view with a correct leader.  
Then, before any TC for view $v + 1$ is formed (if any), at least $f + 1$ correct validators updated their $\localhighQC$ variable to a QC from view $v$.
\end{lemma}
\begin{proof}
Let $\mathit{tc}$ be any TC such that $\mathit{tc}.\view = v + 1$.  
Since at least $f + 1$ correct validators must have issued timeout messages for view $v + 1$ for $\mathit{tc}$ to exist, and since all timeout messages from correct validators are valid (by \Cref{lemma:jovan_valid_timeout}), each of them must include either a TC or a QC from the previous view $v$.  
By \Cref{lemma:jovan_correct_leader_no_tc}, no TC can be formed for view $v$, which implies that all correctly issued timeout messages must carry a QC from view $v$.
Moreover, by the definition of the $\CreateTimeoutMsg$ function (line~\ref{line:view_certificate_create_timeout_message} of \Cref{Algorithm:Utilities}), this QC must be stored in the $\localhighQC$ variable of each of these $f + 1$ correct validators.
\qed
\end{proof}

We now complement the previous lemma by proving that if there exists any post-GST view $v$ with a correct leader, and some correct validator enters a view greater than $v + 1$, then at least $f + 1$ correct validators must have updated their $\localhighQC$ variable to a QC formed in view $v$.

\begin{lemma} \label{lemma:jovan_post_gst_correct_leader_enter}
Let $v > 1$ be any post-GST view with a correct leader.  
Then, before any correct validator enters any view greater than $v + 1$, at least $f + 1$ correct validators updated their $\localhighQC$ variable to a QC from view $v$.
\end{lemma}
\begin{proof}
According to \Cref{lemma:jovan_enter_before}, for any view $> v + 1$ to be entered, some correct validator must first enter view $v + 2$.  
If this occurs as a result of a QC formed in view $v + 1$, then at least $f + 1$ correct validators must have voted for a proposal in that view.  
Such votes can be cast only after the validators have entered view $v + 1$, regardless of whether the vote was issued upon receiving a proposal or upon timing out from view $v + 1$.
Upon entering view $v + 1$---which can only happen via a QC from view $v$ (by \Cref{lemma:jovan_correct_leader_no_tc})---each of these $f + 1$ correct validators update their $\localhighQC$ variable to this QC (see the $\IncrementView$ function).
Hence, the lemma holds in this case.  
Alternatively, if the transition to view $v + 2$ occurs via a TC for view $v + 1$, then the lemma holds due to \Cref{lemma:jovan_post_gst_correct_leader_tc}.
\qed
\end{proof}



The following lemma establishes that any proposal (whether fresh or a reproposal) issued by a correct leader after GST is strictly extended by every proposal issued in subsequent views.  
This property is essential for ensuring liveness.  
Specifically, once a correct leader issues a proposal in a post-GST view, the lemma guarantees that all future proposals---regardless of whether they are issued by correct or Byzantine leaders---must extend this honest proposal.  
Consequently, once a group of three consecutive correct leaders occurs, the first leader in this group will strictly extend the honest proposal, and since the block proposed by this leader becomes committed by the end of the third view from the group, the original honest proposal is eventually committed as well.

\begin{lemma} \label{lemma:jovan_crucial_liveness}
Consider a proposal $p$ issued by a correct leader in any post-GST view $v > 1$.  
Then, any valid proposal $p^*$ issued in any view $v^* \geq v + 1$ (i.e., $p^*.\view = v^* \geq v + 1$) must strictly extend $p$.
\end{lemma}
\begin{proof}
We prove the claim by induction on the view number $v'$.
Consider the following four induction hypotheses, $X(v')$, $Y(v')$, $Q(v')$, and $Z(v')$:
\begin{itemize}
    \item $X(v')$: Any (valid) proposal $p^*$ in view $v' \geq v + 1$ strictly extends $p$ and satisfies $p^*.\block.\qc.\view \geq v$.
    Consequently, any (valid) tip $T$ (of a fresh proposal) from view $v'$ also strictly extends $p$ and satisfies $T.\blockheader.\qc.\view \geq v$.
     
    \item $Y(v')$: Let $\mathit{tc}$ be any TC with $\mathit{tc}.\view = v' \geq v + 1$ and $\mathit{tc}.\highQC \neq \bot$.
    Then, any proposal $p_{\mathit{tc}}$ with $p_{\mathit{tc}}.\blockid = \mathit{tc}.\highQC.\blockid$ extends $p$ and $\mathit{tc}.\highQC.\view \geq v$.

    \item $Q(v')$: For any QC $\mathit{qc}$ with $\mathit{qc}.\view = v' \geq v + 1$, any proposal $p_{\mathit{qc}}$ with $p_{\mathit{qc}}.\blockid = \mathit{qc}.\blockid$ strictly extends proposal $p$.
    
    \item $Z(v')$: Let $\mathit{tc}$ be any TC with $\mathit{tc}.\view = v' \geq v + 1$ and $\mathit{tc}.\hightip \neq \bot$. 
    Then, the proposal whose tip is $\mathit{tc}.\hightip$ strictly extends $p$ and $\mathit{tc}.\hightip.\blockheader.\qc.\view \geq v$.
\end{itemize}

We begin by proving the base case. 
Note that $X(v')$ is equivalent to the lemma’s statement.

\medskip
\noindent \underline{$X(v' = v + 1)$:}
Let $p^*$ be any proposal such that $p^*.\view = v + 1$.  
By \Cref{lemma:jovan_correct_leader_no_tc}, we have $p^*.\tc = \bot$, which implies that $p^*.\block.\qc.\view = v$.  
Due to \Cref{lemma:jovan_post_gst_correct_leader_qc}, it follows that a parent of proposal $p^*$ is $p$, thereby proving that $X(v' = v + 1)$ holds.

\medskip
\noindent \underline{$Y(v' = v + 1)$:}
Let $\mathit{tc}$ be any TC such that $\mathit{tc}.\view = v + 1$ and $\mathit{tc}.\highQC \neq \bot$.  
By \Cref{lemma:jovan_post_gst_correct_leader_tc}, at least $f + 1$ correct validators must have updated their $\localhighQC$ variable to a QC for view $v$ before $\mathit{tc}$ was formed.  
Moreover, as $\mathit{tc}.\view = v + 1$, the preceding argument, together with the validity of $\mathit{tc}$, implies that $\mathit{tc}.\highQC.\view = v$.
Due to \Cref{lemma:jovan_post_gst_correct_leader_qc}, $\mathit{tc}.\highQC.\blockid = p.\blockid$, thereby proving $Y(v' = v + 1)$.

\medskip
\noindent \underline{$Q(v' = v + 1)$:}
Let $\mathit{qc}$ be any QC with $\mathit{qc}.\view = v' = v + 1$.  
We distinguish between two cases:
\begin{itemize}
    \item Suppose that at least one correct validator contributing to $\mathit{qc}$ did so upon receiving a proposal in view $v'$.  
    In this case, the received proposal strictly extends proposal $p$, since $X(v' = v + 1)$ holds.  
    Hence, $Q(v' = v + 1)$ is satisfied.

    \item Suppose that no correct validator contributing to $\mathit{qc}$ did so upon receiving a proposal in view $v'$.  
    In this case, all correct votes issued in view~$v'$ are based on timeout messages, each reflecting the validator’s local tip at the time of issuance.  
    Now consider any correct validator~$i$ contributing to~$\mathit{qc}$.  
    Such a contribution must occur when validator~$i$ (either fully or prematurely) times out from view~$v + 1$.  
    This implies that validator~$i$ had previously entered view~$v + 1$, which, by \Cref{lemma:jovan_post_gst_correct_leader_qc}, can only happen after accepting a QC from view~$v$.  
    Therefore, at the time of issuing the timeout message, validator~$i$'s $\localhighQC$ variable must contain a QC from view~$v$.
    The only circumstance under which validator~$i$ would issue a vote contributing to~$\mathit{qc}$ is if it updated its local tip in view~$v + 1$ based on a fresh proposal---since only in that case would $\localhightip.\view > \localhighQC.\view$, prompting a vote.  
    However, in this scenario, the condition $X(v' = v + 1)$ implies that $Q(v' = v + 1)$ holds.


\end{itemize}
Thus, $Q(v' = v + 1)$ holds in both cases.

\medskip
\noindent \underline{$Z(v' = v + 1)$:}
Let $\mathit{tc}$ be any TC such that $\mathit{tc}.\view = v + 1$ and $\mathit{tc}.\hightip \neq \bot$.  
By \Cref{lemma:jovan_post_gst_correct_leader_tc}, at least $f + 1$ correct validators must have updated their $\localhighQC$ variable to a QC for view $v$ before $\mathit{tc}$ was formed. 
As $\mathit{tc}.\view = v + 1$ and at least one correct validator contributing to $\mathit{tc}$ has its $\localhighQC$ variable containing a QC from view $v$, it follows that $\mathit{tc}.\hightip.\view = v + 1$.  
As $X(v' = v + 1)$ holds, $\mathit{tc}.\hightip$ must strictly extend proposal $p$ and $\mathit{tc}.\hightip.\blockheader.\qc.\view \geq v$, thereby proving $Z(v' = v + 1)$.

\smallskip
We now prove that the inductive step holds.

\medskip
\noindent \underline{$X(v' > v + 1)$:}
Let us consider all possible cases for a valid proposal $p^*$ with $p^*.\view = v' > v + 1$:
\begin{itemize}
    \item Suppose $p^*$ is a fresh proposal.
    We further consider three possibilities:
    \begin{itemize}
        \item Let $p^*.\block.\qc.\view = v' - 1 \geq v + 1$.
        In this case, a parent of $p^*$ strictly extends proposal $p$ (because $Q(v' - 1 \geq v + 1)$ holds), which implies that $p^*$ strictly extends proposal $p$.
        Thus, $X(v' > v + 1)$ holds in this case.

        \item Let $p^*.\nec \neq \bot$.
        We first prove that $p^*.\nec.\hightipQCview = p^*.\block.\qc.\view \geq v$.
        In order for $p^*.\nec$ to be formed, a correct validator $i$ must have received an NE request accompanied by a TC $\mathit{tc}$ with $\mathit{tc}.\view = v' - 1 \geq v + 1$.
        Given that $Z(v' - 1)$ holds, $\mathit{tc}.\hightip$ strictly extends proposal $p$ and $\mathit{tc}.\hightip.\blockheader.\qc.\view \geq v$.
        Therefore, $p^*.\nec.\hightipQCview = p^*.\block.\qc.\view \geq v$.
        Hence, it is left to prove that $p^*$ strictly extends proposal $p$.
        We now differentiate two cases:
        \begin{itemize}
            \item Let $p^*.\block.\qc.\view = v$.
            In this case, a parent of $p^*$ must be $p$ (by \Cref{lemma:jovan_post_gst_correct_leader_qc}), which proves the statement of the invariant.

            \item Let $p^*.\block.\qc.\view > v$.
            In this case, the fact that $Q(p^*.\block.\qc.\view \geq v + 1)$ ensures that the statement of the invariant holds.
        \end{itemize}

        
        \item Let $p^*.\nec = \bot$ and $p^*.\tc.\highQC \neq \bot$.
        Note that $p^*.\tc.\view = v' - 1 \geq v + 1$.
        In this case, a parent of $p^*$ extends $p$ and $p^*.\tc.\highQC.\view \geq v$ (as ensured by $Y(v' - 1 \geq v + 1)$), which proves $X(v' > v + 1)$ in this scenario.
    \end{itemize}
    
    \item Suppose $p^*$ is a reproposal.
    In this case, $p^*$ is a reproposal of $p^*.\tc.\hightip$; note that $p^*.\tc.\view = v' - 1 \geq v + 1$.
    Due to $Z(v' - 1 \geq v + 1)$, we have that $p^*$ strictly extends proposal $p$ (as $p^*.\tc.\hightip$ does so) and $p^*.\tc.\hightip.\blockheader.\qc.\view \geq v$.
    Thus, $X(v' > v + 1)$ holds in this case as well.
\end{itemize}
As $X(v' > v + 1)$ holds in all possible scenarios, $X(v' > v + 1)$ is satisfied.

\medskip
\noindent $Y(v' > v + 1)$:
Let $\mathit{tc}$ be any TC such that $\mathit{tc}.\view = v' > v + 1$ and $\mathit{tc}.\highQC \neq \bot$.  
Since $\mathit{tc}$ is formed and the first correct validator to issue a timeout message must do so upon fully timing out from view $v' > v + 1$ (by \Cref{lemma:jovan_first_timeout}),  
\Cref{lemma:jovan_post_gst_correct_leader_enter} ensures that at least $f + 1$ correct validators must have updated their $\localhighQC$ variable to a QC for view $v$ before $\mathit{tc}$ was formed.  
Hence, $\mathit{tc}.\highQC.\view \geq v$. 
Moreover, due to the validity of $\mathit{tc}$, $\mathit{tc}.\highQC.\view < v'$.
We now consider two cases:
\begin{itemize}
    \item If $\mathit{tc}.\highQC.\view = v$, then by \Cref{lemma:jovan_post_gst_correct_leader_qc}, $\mathit{tc}.\highQC$ points to proposal $p$,  
    thereby establishing that $Y(v' > v + 1)$ holds in this case.

    \item If $\mathit{tc}.\highQC.\view > v$, then $Y(v' > v + 1)$ follows from $Q(\mathit{tc}.\highQC.\view)$.
    

\end{itemize}

\medskip
\noindent \underline{$Q(v' > v + 1)$:}
Let $\mathit{qc}$ be any QC with $\mathit{qc}.\view = v' > v + 1$.  
We distinguish between two cases:
\begin{itemize}
    \item Suppose that at least one correct validator contributing to $\mathit{qc}$ did so upon receiving a proposal in view $v'$.  
    In this case, the received proposal strictly extends proposal $p$, since $X(v' > v + 1)$ holds.  
    Hence, $Q(v' > v + 1)$ is satisfied.

    \item Suppose that no correct validator contributing to $\mathit{qc}$ did so upon receiving a proposal in view $v'$.  
    In this case, all correct votes issued in view~$v'$ are based on timeout messages, each reflecting the validator’s local tip at the time of issuance.  
By \Cref{lemma:jovan_post_gst_correct_leader_enter}, at least $f + 1$ correct validators must have updated their $\localhighQC$ variable to a QC from view~$v$ before $\mathit{qc}$ was formed.
Let $i$ be any such correct validator who both updated its $\localhighQC$ to a QC from view~$v$ prior to the formation of~$\mathit{qc}$ and contributed a vote to~$\mathit{qc}$.  
This implies that validator~$i$'s local tip~$T$ must belong to a view $\geq v + 1$, as the vote is cast for that proposal.  
(Otherwise, $i$ would have instead issued a timeout message containing its local QC and would not have voted.)
Given that $X(T.\view \in [v + 1, v'])$ holds, it follows that $\mathit{qc}$ must point to a proposal that strictly extends~$p$.

\end{itemize}
Thus, $Q(v' > v + 1)$ holds in both cases.

\medskip
\noindent \underline{$Z(v' > v + 1)$:}
Let $\mathit{tc}$ be any TC with $\mathit{tc}.\view = v' > v$ and $\mathit{tc}.\hightip \neq \bot$.
Since $\mathit{tc}$ is formed and the first correct validator to issue a timeout message must do so upon timing out from view $v' > v + 1$,  
\Cref{lemma:jovan_post_gst_correct_leader_enter} ensures that at least $f + 1$ correct validators must have updated their $\localhighQC$ variable to a QC for view $v$ before $\mathit{tc}$ was formed.  
Hence, $\mathit{tc}.\hightip.\view > v$; note that $\mathit{tc}.\hightip.\view \leq v'$.  
Therefore, $Z(v' > v + 1)$ holds due to $X(\mathit{tc}.\hightip.\view)$.
\qed
\end{proof}

Next, we derive an upper bound on the difference $\tau(v + 1) - \tau(v)$ for a post-GST view $v$.

\begin{lemma} \label{lemma:jovan_bound}
Let $v > 1$ be any post-GST view.  
Then, the following bound holds:
\[
\tau(v + 1) - \tau(v) \leq 2\viewduration + 2\delta.
\]
\end{lemma}
\begin{proof}
By contradiction, suppose that $\tau(v + 1) - \tau(v) > 2\viewduration + 2\delta$.
Let us now consider two different cases:
\begin{itemize}
    \item Suppose that the leader of view $v$ is correct.  
    By \Cref{lemma:jovan_issue_vote_time_tc}, every correct validator issues its vote for a proposal from view $v$ by time $\tau(v) + 3\delta + \recoveryduration + \delta < \tau(v) + 2\viewduration + 2\delta$ (otherwise, it would hold that $\tau(v + 1) < \tau(v) + 3\delta + \recoveryduration + \delta$, contradicting our initial assumption).  
    By time $\tau(v) + 3\delta + \recoveryduration + 2\delta$, the leader receives votes from all correct validators.  
    Once this occurs, the leader forms a QC (line~\ref{line:form_qc} of \Cref{Algorithm:Consensus-Execution_1}) for view $v$ and proceeds to enter view $v + 1$ (line~\ref{line:increment_view_form_qc} of \Cref{Algorithm:Consensus-Execution_1}).  
    Hence, $\tau(v + 1) - \tau(v) \leq 3\delta + \recoveryduration + 2\delta$, which again contradicts our starting assumption.
    This implies that the leader of $v$ cannot be correct.

    \item Suppose that the leader of view $v$ is faulty.
    First, observe that a correct validator $i$ that entered view $v$ at time $\tau(v)$ issues a valid timeout message at time $\tau(v) + \viewduration$, which carries either a QC or a TC from view $v - 1$.  
    (Otherwise, validator $i$ would have already advanced to a future view, contradicting the initial assumption.)  
    Consequently, all correct validators enter view $v$ by time $\tau(v) + \viewduration + \delta$, since upon receiving any valid timeout message associated with view $v$, a correct validator immediately enters that view (line~\ref{line:increment_view_receive_qc_timeout_message} of \Cref{Algorithm:Consensus-Execution_2}).  
    Subsequently, each correct validator issues its own timeout message by time $\tau(v) + 2\viewduration + \delta$ (as none have progressed to later views by this point), and all such messages are received by every correct validator by time $\tau(v) + 2\viewduration + 2\delta$.  
    At this stage, a correct validator can form a QC or a TC associated with view $v$ (line~\ref{line:form_tc_qc} of \Cref{Algorithm:Consensus-Execution_2}), enabling it to move to the next view $v + 1$ (line~\ref{line:increment_view_form_qc_timeout_message} or line~\ref{line:increment_view_form_tc} of \Cref{Algorithm:Consensus-Execution_2}).  
    Thus, we once again reach a contradiction with the initial assumption.
\end{itemize}
Since neither of the two cases can occur, our initial assumption must be false, thereby completing the proof of the lemma.
\qed
\end{proof}

Next, we show that every correct validator enters a post-GST view within $(f + 1) \cdot (2\viewduration + 2\delta) + 4\delta + \recoveryduration$ time after GST.

\begin{lemma} \label{lemma:jovan_liveness_weird}
Let $v_{\max}$ denote the highest view that is not a post-GST view.  
Then, every correct validator enters a view $> v_{\max}$ by time
\[
\text{GST} + (f + 1) \cdot (2\viewduration + 2\delta) + 4\delta + \recoveryduration.
\]
\end{lemma}
\begin{proof}
Since $v_{\max}$ is the highest non-post-GST view, it follows that all views $> v_{\max}$ are post-GST views.  
Consider any correct validator $i$.  
Among the views $\{v_{\max} + 1, v_{\max} + 2, \dots, v_{\max} + f + 1\}$, there exists a view $v_i$ with a correct leader.
By \Cref{lemma:jovan_correct_enters}, the leader of view~$v_i$ enters view~$v_i$.  
Then, by \Cref{lemma:jovan_correct_leader_proposal_time}, it issues its proposal by time $\tau(v_i) + 3\delta + \recoveryduration$, which is received by validator~$i$ no later than time $\tau(v_i) + 4\delta + \recoveryduration$.
Upon receiving this proposal, if $\curView$ at validator~$i$ is less than~$v_i$, then validator~$i$ enters view~$v_i > v_{\max}$.  
Otherwise, validator~$i$ has already entered a view higher than~$v_{\max}$.
Thus, it remains to bound $\tau(v_i) + 4\delta + \recoveryduration - \text{GST}$, which we do by first bounding $\tau(v_{\max} + 1) - \text{GST}$.

\smallskip
\noindent\emph{Bounding $\tau(v_{\max} + 1) - \text{GST}$.}
Assume, for the sake of contradiction, that 
\[
\tau(v_{\max} + 1) - \text{GST} > 2\viewduration + 2\delta.
\]
Let $i$ be a correct validator that enters view $v_{\max}$ at time $\tau(v_{\max})$.  
Validator $i$ issues a valid timeout message associated with view $v_{\max}$ by time $\text{GST} + \viewduration$, carrying either a QC or a TC from view $v_{\max} - 1$.
(Otherwise, validator $i$ would have already progressed to a future view, contradicting the initial assumption.)
Consequently, all correct validators enter view $v_{\max}$ by time $\text{GST} + \viewduration + \delta$, since upon receiving any valid timeout message for view $v_{\max}$, a correct validator immediately enters that view (if it has not already done so) at line~\ref{line:increment_view_receive_qc_timeout_message} of \Cref{Algorithm:Consensus-Execution_2}.  
Each correct validator then issues its own timeout message for view $v_{\max}$ by time $\text{GST} + 2\viewduration + \delta$ (as none have advanced to any future view), and these messages are received by every correct validator by time $\text{GST} + 2\viewduration + 2\delta$.  
At this point, a correct validator can form a QC or a TC associated with view $v_{\max}$ (line~\ref{line:form_tc_qc} of \Cref{Algorithm:Consensus-Execution_2}), enabling it to move to view $v_{\max} + 1$.  
Hence,
\[
\tau(v_{\max} + 1) - \text{GST} \leq 2\viewduration + 2\delta.
\]

\smallskip
\noindent \emph{Bounding $\tau(v_i) + 4\delta + \recoveryduration - \text{GST}$.}
Recall that $\tau(v_{\max} + f + 1) \geq \tau(v_{\max} + 1)$ by \Cref{lemma:jovan_enter_before}.
Finally, we bound the time by which validator $i$ enters its leader view $v_i$ and issues its proposal: 
\[
\begin{aligned}
\tau(v_i) + 4\delta + \recoveryduration - \text{GST}
\;\;&\leq 
\underbrace{\big(\tau(v_{\max} + 1) - \text{GST}\big)}_{\leq 2\viewduration + 2\delta}
+ 
\underbrace{\big(\tau(v_{\max} + 2) - \tau(v_{\max} + 1)\big)}_{\leq 2\viewduration + 2\delta} \\
&\quad + \dots + 
\underbrace{\big(\tau(v_{\max} + f + 1) - \tau(v_{\max} + f)\big)}_{\leq 2\viewduration + 2\delta}
+ 4\delta + \recoveryduration \\[1ex]
&\leq (f + 1) \cdot (2\viewduration + 2\delta) + 4\delta + \recoveryduration.
\end{aligned}
\]
This establishes the desired bound and completes the proof.
\qed
\end{proof}

We are now ready to prove that \sysname satisfies the liveness property.  
Specifically, we show that any block issued by a correct validator after time 
\[
\text{GST} + (f + 1) \cdot (2\viewduration + 2\delta) + 4\delta + \recoveryduration
\]
is eventually committed.
The need for this additional delay beyond $\text{GST}$ arises from the instability of the network immediately after GST, as established in \Cref{lemma:jovan_liveness_weird}.  
Prior to this point, the network may still exhibit pre-GST behavior, which can prevent correct proposals from being supported.
For instance, consider a scenario where a correct validator is the leader of a pre-GST view~$v$.  
If this validator issues its proposal exactly at time $\text{GST}$, it is possible that all other correct validators have already advanced to a higher view due to earlier timeouts.  
As a result, the proposal for view~$v$ cannot gather sufficient votes to be committed, despite being correctly formed.
This delay thus ensures that the system has fully transitioned out of any pre-GST instability before requiring commitment guarantees.  
One could interpret this as a semantic adjustment: the ``effective GST'' is the model-defined GST plus this bounded delay, accounting for residual view divergence.

\begin{theorem} [Liveness] \label{theorem:liveness}
After time $\textup{GST} + (f + 1) \cdot (2\viewduration + 2\delta) + 4\delta + \recoveryduration$, every block issued by a correct validator is committed by every correct validator.
\end{theorem}
\begin{proof}
Consider any correct validator $i$ that issues a proposal $p$ after time 
\[
\textup{GST} + (f + 1) \cdot (2\viewduration + 2\delta) + 4\delta + \recoveryduration.
\]
By \Cref{lemma:jovan_liveness_weird}, validator $i$ must have entered a post-GST view by that time.  
Let $v$ denote the post-GST view in which the validator $i$ issues its proposal $p$.
By \Cref{lemma:jovan_consecutive_views}, within the sequence of views from $v + 1$ to $v + 3f + 3$, there exists a group of three consecutive views
\[
S = [v^*,\, v^* + 1,\, v^* + 2],
\]
all of which have correct leaders.  

Since the leader of view $v^*$ is correct, it issues a valid proposal $p^*$ that strictly extends proposal $p$ (by \Cref{lemma:jovan_crucial_liveness}).  
By \Cref{lemma:jovan_correct_enters}, the leader of view~$v^* + 1$ enters view~$v + 1$.  
Furthermore, by \Cref{lemma:jovan_correct_leader_no_tc}, this leader enters view~$v^* + 1$ only after accepting a QC from view~$v^*$.  
By \Cref{lemma:jovan_post_gst_correct_leader_qc}, this implies that the proposal issued by the leader of view~$v^* + 1$ must contain a QC for~$p^*$.
Furthermore, the leader of view $v^* + 2$ issues a proposal that contains a QC for the proposal from view $v^* + 1$ (by \Cref{lemma:jovan_two_leaders}), which itself includes a QC for $p^*$.  
Consequently, every correct validator $j$ receives the proposal from the leader of view $v^* + 2$ by time $\tau(v^* + 2) + 2\delta$.
We now distinguish two cases:
\begin{itemize}
    \item If the current view of validator $j$ is $\leq v^* + 2$ upon receiving this proposal, then validator $j$ commits proposal $p$ (as the proposal from view $v^* + 2$ is processed) at line~\ref{line:commit_spec_commit_proposal} of \Cref{Algorithm:Consensus-Execution_1}, which proves the theorem in this case.
    
    \item Otherwise, validator $j$ has already entered some view $v^{**} > v^* + 2$.  
    Since no TC can be formed for any view between $v^* + 2$ and $v^{**} - 1$ by time $\tau(v^* + 2) + 2\delta < \tau(v^* + 2) + \viewduration$, validator $j$ must have entered view $v^{**}$ upon accepting a QC from view $v^{**} - 1 \geq v^* + 2$.  
    This QC necessarily ``extends'' the proposal from view $v^*$ through a chain of QCs (with no TCs), ensuring that there exist at least two consecutive QCs extending the proposal from view $v^*$.
    Hence, $p$ is committed in this case as well, since the $\CommitSpecCommit$ function is invoked after every QC acceptance.
    \qed
\end{itemize}
\end{proof}

\subsection{Optimistic Responsiveness}

In this subsection, we prove that \sysname satisfies optimistic responsiveness: if all validators are correct, then during any post-GST time interval of duration $d$, each correct validator commits at least $\Omega(d / \delta)$ blocks, where $\delta$ denotes the actual network delay.
To this end, we first determine when the first post-GST view is entered by the first correct validator relative to GST.

\begin{lemma} \label{lemma:jovan_smallest_post_gst}
Let $v > 1$ be the smallest post-GST view.
Then, the following bound holds:
\[
\tau(v) - \text{GST} \leq 2\viewduration + 2\delta.
\]
\end{lemma}
\begin{proof}
(This statement was already established in the first part of the proof of \Cref{lemma:jovan_liveness_weird}.  
For completeness, we provide a self-contained proof here.)
Suppose, for the sake of contradiction, that $\tau(v) - \text{GST} > 2\viewduration + 2\delta$.  
By \Cref{lemma:jovan_enter_before}, no correct validator enters any view greater than $v - 1$ by time $\text{GST} + 2\viewduration + 2\delta$.

Note first that $v - 1$ cannot be a post-GST view, since $v$ is the smallest post-GST view.  
Hence, $\tau(v - 1) < \text{GST}$.  
Let $i$ be a correct validator that enters view $v - 1$ at time $\tau(v - 1)$.  
Validator $i$ issues its timeout message for view $v - 1$ by time $\text{GST} + \viewduration$, as it does not enter any view greater than $v - 1$ by that time.  
By \Cref{lemma:jovan_valid_timeout}, this timeout message must carry either a QC or a TC from view $v - 2$.  
Upon receiving this timeout message (by time $\text{GST} + \viewduration + \delta$), each correct validator enters view $v - 1$, unless it has already done so.  

Subsequently, every correct validator (either fully or prematurely) times out from view $v - 1$ by time $\text{GST} + \viewduration + \delta + \viewduration$, since no correct validator enters any view higher than $v - 1$ before that time.  
Finally, by time $\text{GST} + 2\viewduration + 2\delta$, every correct validator has received timeout messages from all other correct validators and forms either a QC or a TC for view $v - 1$ (line~\ref{line:form_tc_qc} of \Cref{Algorithm:Consensus-Execution_2}), thereby entering view $v$ (line~\ref{line:increment_view_form_qc_timeout_message} or line~\ref{line:increment_view_form_tc} of \Cref{Algorithm:Consensus-Execution_2}).  
This contradicts our initial assumption, completing the proof.
\qed
\end{proof}




We are now ready to establish the optimistic responsiveness property.

\begin{theorem}[Optimistic responsiveness]
\label{thm:optimistic_responsiveness}
If all validators are correct, then during any post-GST time interval $[\tau_1, \tau_2]$ such that $\tau_1 > \text{GST} + 4\viewduration + 4\delta$ and $\tau_2 - \tau_1 \in \omega(\delta)$, each correct validator commits $\Omega\!\left(\frac{\tau_2 - \tau_1}{\delta}\right)$ blocks.
\end{theorem}

\begin{proof}
Let $v_{\max}$ denote the largest post-GST view such that $\tau(v_{\max}) < \tau_1$.  
Note that $v_{\max}$ is well defined by \Cref{lemma:jovan_smallest_post_gst}.
Next, let $\mathcal{V}$ denote the set of views that ``start'' during the time interval $[\tau_1, \tau_2]$:
\[
\mathcal{V} = \{\, v \mid \tau(v) \geq \tau_1 \text{ and } \tau(v) + 2\delta \leq \tau_2 \,\}.
\]

First, we show that $\tau(v_{\max} + 1) < \tau_1 + 5\delta$, i.e., that the smallest view of $\mathcal{V}$ starts at most $5\delta$ time after $\tau_1$.
Let $v_0$ denote the smallest post-GST view.
By \Cref{lemma:jovan_smallest_post_gst}, $\tau(v_0) \leq \text{GST} + 2\viewduration + 2\delta$.
Moreover, \Cref{lemma:jovan_bound} proves that $\tau(v_0 + 1) \leq \tau(v_0) + 2\viewduration + 2\delta \leq \text{GST} + 4\viewduration + 4\delta < \tau_1$.
Hence, $v_{\max} \geq v_0 + 1$, which implies that view $v_{\max} - 1$ is a post-GST view; as all validators are correct, its leader is correct.
By \Cref{lemma:jovan_correct_leader_no_tc}, no TC can be formed for view $v_{\max} - 1$, and thus the leader of view $v_{\max}$ entered view $v_{\max}$ upon accepting a QC for view $v_{\max} - 1$.
Therefore, \Cref{lemma:jovan_two_leaders} applies to the pair of views $v_{\max}$ and $v_{\max} + 1$: the leader of view $v_{\max} + 1$ enters view $v_{\max} + 1$ upon accepting a QC for view $v_{\max}$ by time $\tau(v_{\max}) + 3\delta + 2\delta$.
Consequently, $\tau(v_{\max} + 1) \leq \tau(v_{\max}) + 5\delta < \tau_1 + 5\delta$; recall that $\tau(v_{\max} + 1) \geq \tau_1$ by the definition of $v_{\max}$.
As $\tau_2 - \tau_1 \in \omega(\delta)$, the smallest view $v'$ of $\mathcal{V}$ is indeed $v' = v_{\max} + 1$.
Note that the leader of view $v'$ enters view $v'$ upon accepting a QC for view $v' - 1 = v_{\max}$.

By \Cref{lemma:jovan_two_leaders}, the number of views in~$\mathcal{V}$ satisfies
\[
|\mathcal{V}| \in \Omega\!\left(\frac{\tau_2 - \tau_1}{\delta}\right).
\]
Specifically, the smallest view of $\mathcal{V}$ starts by time $\tau_1 + 5\delta$ and, for each view~$v \in \mathcal{V} \setminus{ \{ \text{the largest view of $\mathcal{V}$} \}}$, \Cref{lemma:jovan_two_leaders} establishes that the time difference between consecutive views is bounded, i.e., $\tau(v + 1) - \tau(v) \in O(\delta)$; recall that, by \Cref{lemma:jovan_two_leaders}, the leader of each view of $\mathcal{V}$ enters that view upon accepting a QC for the preceding view, which ensures that \Cref{lemma:jovan_two_leaders} indeed applies to every pair of consecutive views of $\mathcal{V}$.

Let $v^*$ denote the largest view in $\mathcal{V}$.  
By \Cref{lemma:jovan_two_leaders} (and the assumption that all validators are correct), all correct validators receive the proposal $p^*$ issued by the leader of view $v^*$ by time $\tau(v^*) + 2\delta \leq \tau_2$.  
Furthermore, by \Cref{lemma:jovan_correct_leader_no_tc}, all proposals issued in views within $\mathcal{V}$ (including $v^*$) are ``connected'' through QCs, since no TC can be formed for any view in $\mathcal{V}$.

We now distinguish two cases depending on the state of any correct validator $j$ upon receiving the proposal $p^*$ from the leader of view $v^*$ (which occurs by time $\tau_2$):
\begin{itemize}
    \item If the current view of validator $j$ is at most $v^*$ when the proposal is received, then $j$ commits the grandparent of $p^*$ (as the proposal from view $v^*$ is processed) at line~\ref{line:commit_spec_commit_proposal} of \Cref{Algorithm:Consensus-Execution_1}.
    Therefore, all proposals issued in the views of $\mathcal{V}$---except for the last two---are committed.
    Since $|\mathcal{V}| \in \Omega\left(\frac{\tau_2 - \tau_1}{\delta}\right)$ and we commit $|\mathcal{V}| - 2$ blocks, the number of committed blocks is $|\mathcal{V}| - 2 \in \Omega\left(\frac{\tau_2 - \tau_1}{\delta}\right) - 2 = \Omega\left(\frac{\tau_2 - \tau_1}{\delta}\right)$, where the last equality holds because $\tau_2 - \tau_1 \in \omega(\delta)$.
    
    \item Otherwise, validator $j$ has already entered some view $v^{**} > v^*$.  
    Since no TC can be formed for any view greater than or equal to $v^*$, validator $j$ must have entered view $v^{**}$ upon accepting a QC from view $v^{**} - 1 \geq v^*$.  
    This QC necessarily ``extends'' the proposal from view $v^*$ through a continuous chain of QCs (without any TCs).
    Hence, the grandparent of $p^*$ is committed in this case as well, since the $\CommitSpecCommit$ function is invoked after every QC acceptance.  
    This, in turn, implies that all proposals issued in the views of $\mathcal{V}$---except for the last one---are committed.
    By the same reasoning as in the previous case, the number of committed blocks is $|\mathcal{V}| - 1 \in \Omega\left(\frac{\tau_2 - \tau_1}{\delta}\right)$.
    \qed
\end{itemize}
\end{proof}

As shown, \Cref{thm:optimistic_responsiveness} captures the case where all validators are correct.
However, we can state a stronger guarantee.
Assume there is a sequence $G = [v, v + 1, \dots, v + x - 1]$ of $x$ consecutive post-GST views, each led by a correct validator.
Then, by the time a correct validator receives the proposal issued by the leader of view $v + x - 1$, it commits all blocks proposed in the views of $G$, except possibly those from the last two views $v + x - 2$ and $v + x - 1$.
Moreover, this occurs within $O(x\delta)$ time from the moment the first correct validator enters view $v + 1$.
Crucially, if the leader of view $v$ enters view $v$ upon accepting a QC (rather than a TC), the same guarantee is achieved within $O(x\delta)$ time from the moment the first correct validator enters view $v$.
Note that this guarantee does not require all validators to be correct.

\begin{theorem}[Optimistic responsiveness; stronger guarantee]
Let $G$ be a sequence of $x$ consecutive post-GST views $[v, v + 1, \dots, v + x - 1]$, each led by a correct validator.
Consider any correct validator $i$, and let $\tau_i^{v + x - 1}$ denote the time at which $i$ receives the proposal issued by the leader of view $v + x - 1$.
Then, by time $\tau_i^{v + x - 1}$, validator $i$ commits all blocks proposed in the views of $G$, except possibly those from the last two views $v + x - 2$ and $v + x - 1$.
Furthermore, the time difference $\tau_i^{v + x - 1} - \tau(v + 1)$ lies in $O(x\delta)$.
Finally, if the leader of view $v$ enters view $v$ upon accepting a QC associated with view $v - 1$, then $\tau_i^{v + x - 1} - \tau(v)$ also lies in $O(x\delta)$.
\end{theorem}
\begin{proof}
First, observe that the leader of view $v + 1$ enters view $v + 1$ upon accepting a QC for view $v$.  
By \Cref{lemma:jovan_correct_enters}, the leader of view $v + 1$ indeed enters view $v + 1$.  
Since a validator can only enter a new view upon accepting either a QC or a TC from the previous view, and \Cref{lemma:jovan_correct_leader_no_tc} ensures that no TC can be formed for view $v$, it follows that the leader of view $v + 1$ must have entered view $v + 1$ upon accepting a QC for view v.

Next, for each view $v' \in [v + 1, v + 2, \dots, v + x - 2]$, \Cref{lemma:jovan_two_leaders} shows that the time gap between consecutive views is bounded, i.e., $\tau(v' + 1) - \tau(v') \in O(\delta)$.
Consequently, $\tau(v + x - 1) - \tau(v + 1) \in O(x\delta)$.
Furthermore, if the leader of view $v$ enters view $v$ upon accepting a QC associated with view $v - 1$, then \Cref{lemma:jovan_two_leaders} also proves that $\tau(v + 1) - \tau(v) \in O(\delta)$, which implies that $\tau(v + x - 1) - \tau(v) \in O(x\delta)$.

Let $v^* = v + x - 1$ denote the largest view in $G$.  
By \Cref{lemma:jovan_two_leaders} (and the assumption that all leaders of views from $G$ are correct), each correct validator $i$ receives the proposal $p^*$ issued by the leader of view $v^*$ by time $\tau(v^*) + 2\delta$.
(Note that $\tau(v^*) + 2\delta - \tau(v + 1) \in O(x\delta)$ and, if the leader of view $v$ enters view $v$ upon accepting a QC associated with view $v - 1$, $\tau(v^*) + 2\delta - \tau(v) \in O(x\delta)$.)
Furthermore, by \Cref{lemma:jovan_correct_leader_no_tc}, all proposals issued in views within $G$ (including $v^*$) are ``connected'' through QCs, since no TC can be formed for any view in $G$.

We now distinguish two cases depending on the state of validator $i$ upon receiving the proposal $p^*$ from the leader of view $v^*$ (which occurs by time $\tau(v^*) + 2\delta$):
\begin{itemize}
    \item If the current view of validator $i$ is at most $v^*$ when the proposal is received, then $i$ commits the grandparent of $p^*$ (as the proposal from view $v^*$ is processed) at line~\ref{line:commit_spec_commit_proposal} of \Cref{Algorithm:Consensus-Execution_1}.
    Therefore, all proposals issued in the views of $G$---except for the last two---are committed.
    
    \item Otherwise, validator $i$ has already entered some view $v^{**} > v^*$.  
    Since no TC can be formed for any view $\geq v^*$ by time $\tau(v^*) + 2\delta$ (by \Cref{lemma:jovan_tc_time}), validator $i$ must have entered view $v^{**}$ upon accepting a QC from view $v^{**} - 1 \geq v^*$.  
    This QC necessarily ``extends'' the proposal from view $v^*$ through a continuous chain of QCs (without any TCs).
    Hence, the grandparent of $p^*$ is committed in this case as well, since the $\CommitSpecCommit$ function is invoked after every QC acceptance, which also proves that all proposals issued in the views of $G$---except for the last two---are committed.
    \qed
\end{itemize}
\end{proof}
\subsection{Leader Fault Isolation}

In this subsection, we prove that \sysname satisfies the leader fault isolation property, which states that, after GST, if a view $v$ led by a Byzantine leader is ``sandwiched'' between two views led by correct leaders, then the duration of view $v$ is at most $O(\delta) + \viewduration$.
To this end, we first show that all correct validators ``complete'' any post-GST view led by a correct leader within $O(\delta)$ time of one another.

\begin{lemma} \label{lemma:jovan_fault_isolation}
Let $v > 1$ be any post-GST view with a correct leader.
Then, (1) all correct validators enter a view greater than $v$ by time $\tau(v + 1) + 2\delta$, or (2) $\tau(v + 2) \leq \tau(v + 1) + 2\delta$.
\end{lemma}
\begin{proof}
Let $L$ denote the leader of view $v$.
By \Cref{lemma:jovan_correct_leader_no_tc}, no TC can be formed for view $v$, and at least $f + 1$ correct validators do not issue timeout messages associated with this view.  
Hence, any correct validator $i$ that enters view $v + 1$ at time $\tau(v + 1)$ must have done so upon accepting a QC associated with view $v$. 
Observe that, since at least $f + 1$ correct validators did not emit timeout messages for view $v$, validator $i$ could not have accepted a QC for view $v$ at line~\ref{line:increment_view_form_qc_timeout_message} of \Cref{Algorithm:Consensus-Execution_2}, as this requires $2f + 1$ timeout messages.
To establish the lemma, we now examine all possible circumstances under which validator $i$ could have accepted a QC for view $v$:
\begin{itemize}
    \item Suppose that $i$ accepts a QC for view $v$ at line~\ref{line:increment_view_proposal} of \Cref{Algorithm:Consensus-Execution_1}.

    In this case, validator $i$ forwards the accepted QC to the leader $L$ (line~\ref{line:forward_qc_proposal_current_leader} of \Cref{Algorithm:Consensus-Execution_1}), which is received by $L$ by time $\tau(v + 1) + \delta$.  
    Upon receiving this QC, if the current view of $L$ is $\leq v$, the leader disseminates the QC to all correct validators (line~\ref{line:disseminate_backup_qc_2} of \Cref{Algorithm:Consensus-Execution_1}), that receive it by time $\tau(v + 1) + 2\delta$, thereby establishing the lemma in this case. 
    Otherwise, we consider two different cases:
    \begin{itemize}
        \item Suppose that the current view of $L$ is $v + 1$.
        Then, the leader has already accepted a QC associated with view $v$ and disseminated it by time $\tau(v + 1) + \delta$, which again ensures that the lemma holds.

        \item Suppose that that the current view of $L$ is $> v + 1$.
        In this case, \Cref{lemma:jovan_enter_before} proves that $\tau(v + 2) \leq \tau(v + 1) + 2\delta$.
    \end{itemize}

    \item Suppose that $i$ accepts a QC for view $v$ at line~\ref{line:increment_view_form_qc} of \Cref{Algorithm:Consensus-Execution_1}.
    
    In this case, validator $i$ is either the leader of view $v$ or of view $v + 1$.  
    If $i$ is the leader of view $v$, it disseminates the accepted QC (line~\ref{line:disseminate_backup_qc} of \Cref{Algorithm:Consensus-Execution_1}), which all correct validators receive by time $\tau(v + 1) + \delta$, thereby establishing the lemma.  
    Otherwise, if $i$ is the leader of view $v + 1$, it disseminates its proposal for view $v + 1$ (line~\ref{line:propose_form_qc} of \Cref{Algorithm:Consensus-Execution_1}), which is likewise received by all correct validators by time $\tau(v + 1) + \delta$, thus confirming the lemma in this case as well.

    \item Suppose that $i$ accepts a QC for view $v$ at line~\ref{line:increment_view_receive_qc} of \Cref{Algorithm:Consensus-Execution_1}.

    In this case, the leader $L$ has already disseminated this QC, implying that all correct validators receive it by time $\tau(v + 1) + \delta$, thereby establishing the lemma.

    \item Suppose that $i$ accepts a QC for view $v$ at line~\ref{line:increment_view_receive_qc_next_leader} of \Cref{Algorithm:Consensus-Execution_1}.

    In this case, $i$ is the leader of view $v + 1$.
    Hence, it disseminates its proposal for view $v + 1$ (line~\ref{line:propose_receive_qc} of \Cref{Algorithm:Consensus-Execution_1}), which is received by all correct validators by time $\tau(v + 1) + \delta$, thus confirming the lemma in this case as well.

    \item Suppose that $i$ accepts a QC for view $v$ at line~\ref{line:increment_view_receive_qc_timeout_message} of \Cref{Algorithm:Consensus-Execution_2}.

    In this case, validator $i$ forwards the accepted QC to the leader $L$ (line~\ref{line:forward_qc_timeout_message_current_leader} of \Cref{Algorithm:Consensus-Execution_2}), which is received by $L$ by time $\tau(v + 1) + \delta$.  
    Upon receiving this QC, if the current view of $L$ is $\leq v$, the leader disseminates the QC to all correct validators (line~\ref{line:disseminate_backup_qc_2} of \Cref{Algorithm:Consensus-Execution_1}), that receive it by time $\tau(v + 1) + 2\delta$, thereby establishing the lemma in this case. 
    Otherwise, we consider two different cases:
    \begin{itemize}
        \item Suppose that the current view of $L$ is $v + 1$.
        Then, the leader has already accepted a QC associated with view $v$ and disseminated it by time $\tau(v + 1) + \delta$, which again ensures that the lemma holds.

        \item Suppose that that the current view of $L$ is $> v + 1$.
        In this case, \Cref{lemma:jovan_enter_before} proves that $\tau(v + 2) \leq \tau(v + 1) + 2\delta$.
    \end{itemize}
\end{itemize}
Since the lemma holds in all possible cases, the proof is complete.
\qed
\end{proof}

We are ready to prove that \sysname satisfies the leader fault isolation property.

\begin{theorem} [Leader fault isolation] \label{thm:leader_fault_isolation}
Consider any three consecutive post-GST views $v$, $v + 1$, and $v + 2$ such that the leaders of views $v$ and $v + 2$ are correct.  
Then, the following bound holds:
\[
\bigl(\text{the earliest time at which all correct validators enter a view } > v + 1\bigr) - \tau(v + 1) \leq O(\delta) + \viewduration.
\]
\end{theorem}
\begin{proof}
By \Cref{lemma:jovan_fault_isolation}, one of the following holds:  
(1) all correct validators enter view $v + 1$ by time $\tau(v + 1) + 2\delta$, or  
(2) $\tau(v + 2) \leq \tau(v + 1) + 2\delta$.  
We first establish that $\tau(v + 2) - \tau(v + 1) \leq O(\delta) + \viewduration$.  
This bound follows immediately in the second case.  
Consider now the first case.  
Assume, for the sake of contradiction, that no correct validator enters view $v + 2$ (or any higher view) by time $\tau(v + 1) + 2\delta + \viewduration + \delta$.  
(Otherwise, if some correct validator does enter a view $\geq v + 2$ by that time, \Cref{lemma:jovan_enter_before} guarantees that $\tau(v + 2) - \tau(v + 1) \leq O(\delta) + \viewduration$.)
Since all correct validators enter view $v + 1$ by time $\tau(v + 1) + 2\delta$, and no correct validator enters any view higher than $v + 1$ by time $\tau(v + 1) + 2\delta + \viewduration + \delta$, each correct validator must issue a timeout message associated with view $v + 1$ by time $\tau(v + 1) + 2\delta + \viewduration$---either due to a full or premature timeout.  
Consequently, by time $\tau(v + 1) + 2\delta + \viewduration + \delta$, every correct validator has received timeout messages from all other correct validators and thus forms either a QC or a TC associated with view $v + 1$ (line~\ref{line:form_tc_qc} of \Cref{Algorithm:Consensus-Execution_2}), entering view $v + 2$ at line~\ref{line:increment_view_form_qc_timeout_message} or line~\ref{line:increment_view_form_tc} of \Cref{Algorithm:Consensus-Execution_2} (unless it has already entered a view greater than view $v + 1$).  
Hence, in all possible cases, we conclude that $\tau(v + 2) - \tau(v + 1) \leq O(\delta) + \viewduration$.

Having established that $\tau(v + 2) - \tau(v + 1) \leq O(\delta) + \viewduration$,  
\Cref{lemma:jovan_correct_enters} proves that the leader of view $v + 2$ enters view $v + 2$ by time $\tau(v + 2) + 3\delta$.
Therefore, the leader of view $v + 2$ disseminates either a proposal, a proposal request, or an NE request by time $\tau(v + 2) + 3\delta$. 
This implies that all correct validators receive the proposal, proposal request, or NE request by time $\tau(v + 2) + 4\delta$. 
  Consequently, all correct validators transition to a view higher than $v + 1$ by time $\tau(v + 2) + 4\delta \leq \tau(v + 2) + O(\delta)$. Hence, the statement of the theorem follows.
\qed
\end{proof}




\end{document}